\documentclass{article}
\usepackage[applemac]{inputenc}
\usepackage[colorinlistoftodos]{todonotes}
\usepackage{authblk}
\usepackage{enumerate}
\usepackage{tikz}
  \usetikzlibrary{patterns}
  \usetikzlibrary{plotmarks}
\usepackage{graphicx}
\usepackage{caption}        
\usepackage{subcaption}
\usepackage{fancybox}
\usepackage{hyperref}
\usepackage{standalone}     
\usepackage{float}        
\usepackage{algorithm}
\usepackage{algcompatible}

  \algblockdefx{FORALLP}{ENDFORALLP}[1]%
    {\textbf{for all }#1 \textbf{do in parallel}}%
    {\textbf{end for}}
  \algloopdefx{RETURN}[1][]{\textbf{return} #1}  
\usepackage{color, colortbl}
\usepackage{array}
\usepackage{amsmath,amsfonts,amssymb,xspace,amsthm}
\usepackage{ bbold } 
\newtheorem{theorem}{Theorem}[section]

\newtheorem{lemma}[theorem]{Lemma}
\newtheorem{proposition}[theorem]{Proposition}
\theoremstyle{definition}
\newtheorem{definition}{Definition}[section]

\DeclareRobustCommand{\Trianguloooo}[1]{%
\protect\tikz{
  \begin{tikzpicture}[inner sep=0pt, baseline=(base),scale=0.25]%
    \draw[line width=1.5pt,fill=#1] (1.5,0.33012701892) -- (2.0,1.19615242271) -- (2.5,0.33012701892) -- (1.5,0.33012701892);
    \node (base) at (0.75,0.33012701892) {};
  \end{tikzpicture}}%
  \hspace*{-8pt}
}  

\newcommand{\defproblemu}[3]{
  \vspace{1mm}
\noindent\fbox{
  \begin{minipage}{0.95\textwidth}
  #1 \\
  {\bf{Input:}} #2  \\
  {\bf{Question:}} #3
  \end{minipage}
  }
  \vspace{1mm}
}

\newcounter{TODOcounter}

\newcommand{\stability}{\textsc{Stability}\xspace}

\newcommand{\NC}{\textbf{NC}\xspace}
\newcommand{\Pt}{\textbf{P}\xspace}
\newcommand{\FNSM}{234}

\newcommand{\cO}{\mathcal{O}}
\def\minus{%
  \setbox0=\hbox{-}%
  \vcenter{%
    \hrule width\wd0 height \the\fontdimen8\textfont3%
  }%
}

\providecommand{\keywords}[1]{\textbf{\textit{Keywords:}} #1}
\title{On the Complexity of the Stability Problem of Binary  Freezing Totalistic Cellular Automata}
\author[1,2,3]{Eric Goles}
\author[2]{Diego Maldonado}
\author[1]{Pedro Montealegre}
\author[2]{Nicolas Ollinger}

\affil[1]{Facultad de Ingenieria y Ciencias, Universidad Adolfo Iba\~nez, Santiago, Chile.}
\affil[2]{Univ. Orl\'{e}ans, LIFO EA 4022, FR-45067 Orl\'{e}ans, France}
\affil[3]{LE STUDIUM, Loire Valley Institute for Advanced Studies, Orl\'eans, France}
\date{December 30, 2018}
\begin{document}

\maketitle

\begin{abstract}
In this paper we study the family of two-state Totalistic Freezing Cellular Automata (TFCA) defined over the triangular and square grids with von Neumann neighborhoods. We say that a Cellular Automaton is Freezing and Totalistic if the active cells remain unchanged, and the new value of an inactive cell depends only on the sum of its active neighbors. 

We classify all the Cellular Automata in the class of TFCA,  grouping them in five different classes: the Trivial rules, Turing Universal rules, Algebraic rules, Topological rules and Fractal Growing rules.  At the same time, we study in this family the  \stability problem, consisting in deciding whether an inactive cell becomes active, given an initial configuration. We exploit the properties of the automata in each group to show that:
\begin{itemize}
\item For Algebraic and Topological Rules the \stability problem is in \NC.
\item For Turing Universal rules the \stability problem is \Pt-Complete.
\end{itemize}

\end{abstract}

\keywords{
Computational Complexity,
Freezing Cellular Automata ,
Totalistic Cellular Automata ,
Fast parallel algorithms,
P-Complete
}

\section{Introduction}

%
%
 
A \emph{cellular automata} (CA) is a discrete dynamical system in time and space. The space is defined in our case as two-dimensional regular grid of \emph{cells}. Each cell on the grid has a binary state (active or inactive) which evolves in time accordingly to a local function. The local function depends on a set of \emph{neighbors}, and it is the same for every site in the grid. The global dynamics consists to update all sites synchronously.

Recently,   a particular family of CA, namely the \emph{freezing} CA (FCA) have been introduced and studied in \cite{goles:hal-01294144}.
These are the CA where the cells can only evolve into a state bigger than their current state (in some pre-defined order). In the case of binary cellular automata, where the states are \emph{active} and \emph{inactive}, the freezing dynamics imply that every active cell remains active in subsequent states. It is direct that, every initial configuration on a binary freezing cellular automaton converges in at most $N$ steps to a fixed point (where $N$ is the number of cells).

One challenging problem related to freezing dynamics consists in computing the fixed point reached by a FCA, given an initial configuration. Observe that this problem is equivalent to compute the inactive cells that remain inactive in subsequent states. We call such cells \emph{stable cells}, and {\stability} the problem consisting in deciding if a given cell is stable, given an initial configuration of a FCA. 

Of course, one can solve {\stability} in linear time, simply simulating the dynamics of the FCA. The interesting question is how fast one can determine a solution to {\stability}, and in particular if we may answer faster than simply simulating the automaton. That leads us to study the {\stability} problem in the context of the \emph{Computational Complexity Theory}.  Consider the class {\Pt} of problems that can be solved in polynomial time on a deterministic Turing machine.  Conventionally {\Pt} is considered the class of \emph{feasible problems}. Observe that the simple simulation of a FCA until it reaches a fixed point leads to an algorithm that solves {\stability} in \emph{polynomial time}, i.e. {\stability} is in {\Pt}. 

Conversely, {\NC} is the class containing all problems that can be solved in poly-logarithmic time in a parallel machine using a polynomial number of processors. Informally, {\NC} is the class of problems solvable by a \emph{fast-parallel algorithm}. It is known that $\NC\subseteq\Pt$, and is a wide-believed conjecture that this inclusion is proper. If the conjecture is correct, then there are problems in {\Pt} that are not contained in {\NC}. Such problems are called \emph{intrinsically sequential}. The problems that are the most likely to be intrinsically sequential are the {\Pt}-Complete problems. A problem $L$ is {\Pt}-Complete if every problem in the class \Pt can be reduced $L$ by a function computable in $\NC$.

We aim to classify FCA according to the complexity of its {\stability} problem. More precisely, we either seek for a fast-parallel algorithm that solves the {\stability} problem, or give evidence that  no such algorithm exists showing that the problem is {\Pt}-Complete.  Therefore, our goal is to classify a FCA in two groups, those where {\stability} is in \NC (in that case we say the problem is easy) and those where the problem is \Pt-Complete (so we say that the problem is \emph{hard}).

In this work we study the two simplest ways to tessellate the bidimensional grid: tessellation with triangles (where each cell has three neighbors) and with squares (where each cell has four neighbors, i.e. the two dimensional CA with von Neumann neighborhood).
In each one of those grids we study the family of \emph{freezing totalistic cellular automata} (FTCA). The name totalistic means that the new value of a cell only depends on the sum of its neighbors. We show that this family of CAs exhibits a broad and rich range of behaviors. More precisely, we classify FTCAs in five groups: 
\begin{itemize}
\item Simple rules:  Rules that exhibit very simple dynamics, which reach fixed points in a constant number of steps.
\item Topological rules: Rules where the stability of a cell depends on some topological property given by the initial configuration. 
\item Algebraic rules: Rules where the dynamics can be accelerated, exploiting some algebraic properties given by the rule. 
\item Turing Universal rules: Rules capable of simulating Boolean Circuits and capable to simulate Turing Computation. 
\item Fractal growing rules: Rules that produce patterns which grow forming fractal shapes. 
\end{itemize}

\subsection{Related work}

To our knowledge, the first study related with the computational complexity of Cellular Automata was done by E. Banks. In his PhD thesis he studied the possibility for simple Cellular Automata in two dimensional grids, to simulate logical gates.  If such simulation is possible, the automaton is capable of universal Turing computation \cite{conf/focs/Banks70}. Directly in the context of prediction problems C. Moore et al \cite{RePEc:wop:safiwp:96-08-060} studied the \emph{Majority Automata} (next state of a site will be the most represented in the neighborhood). He proved that predicting $T$ steps in the evolution of the majority automaton is \Pt-Complete in three and more dimensions. The complexity remains open in two dimensions.


More related with freezing dynamics, in \cite{goles:hal-01294144}  it is shown that the \emph{stability} problem is in \NC, for every one-dimensional freezing cellular automata. In order to find a FCA with  a higher complexity, the result of \cite{goles:hal-01294144} shows that it is necessary to study FCA in more than one dimensions. In this context, we should mention that D. Griffeath and C. Moore studied the \emph{Life without Death} Automaton (i.e. the \emph{Game of Life} restricted to freezing dynamics), showing that the {\stability} problem for this rule is {\Pt}-Complete  \cite{RePEc:wop:safiwp:97-05-044}. 

On other hand, in \cite{goles:hal-00914603}, it was studied the \emph{freezing majority cellular automaton}, also known as \emph{bootstrap percolation model}, in arbitrary undirected graph. In this case, an inactive cell becomes active if and only if the active cells are the most represented in its neighborhood. It was proved that \stability is \Pt-Complete over graphs such that its maximum degree (number of neighbors) $\geq5$. Otherwise (graphs with maximum degree $\leq4$), the problem is in \NC. This clearly includes the two dimensional case, with von Neumann neighborhood.

Other approaches on the relation of computational complexity and cellular automata consider other problems different than {\stability}. For example in \cite{Sutner1995, Dennunzio2017} is studied the \emph{reachability} problem: given two configurations of a cellular automata, namely $x$ and $y$, decide if $y$ is in the orbit of configuration $x$. An algorithm solving the {\stability} problem for a FCA can be used to solve the reachability problem when configuration $y$ is the fixed point reached from $x$. Moreover, most of the algorithms presented in this article can be used to solve the reachability problem for FTCA.

\subsection{Structure of the article}

The paper is structured as follows. In Section 2 definitions and notations are introduced. In Section 3, the FTCA for the triangular grid are studied. In Section 4, we study the FTCA on the square grid. Finally, in Section 5 we give some conclusions.


\section{Preliminaries}
Consider the plane tessellated by triangles, as depicted in Figure \ref{fig: grid3}, or tessellated by squares, as depicted in Figure \ref{fig: grid4}. We call such tessellations the \emph{triangular grid} and \emph{square grid}, respectively. 
In each case, a triangle or square is called a \emph{cell} or \emph{site}. In the triangular grid each cell (triangle) has three adjacent cells, and in the square grid each cell (square) has four adjacent cells. A cell  that is adjacent to a cell $u$ is called a \emph{neighbor} of $u$. The set of neighbors of $u$ is denoted by $N(u)$. In the triangular and square grid, this definition of neighbors is called the \emph{von Neumann Neighborhood}, and it is denoted $N(0,0)$. 

\begin{figure}
  \centering
  \begin{subfigure}[t]{0.475\textwidth}
    \centering
      \includegraphics[width=.95\textwidth]{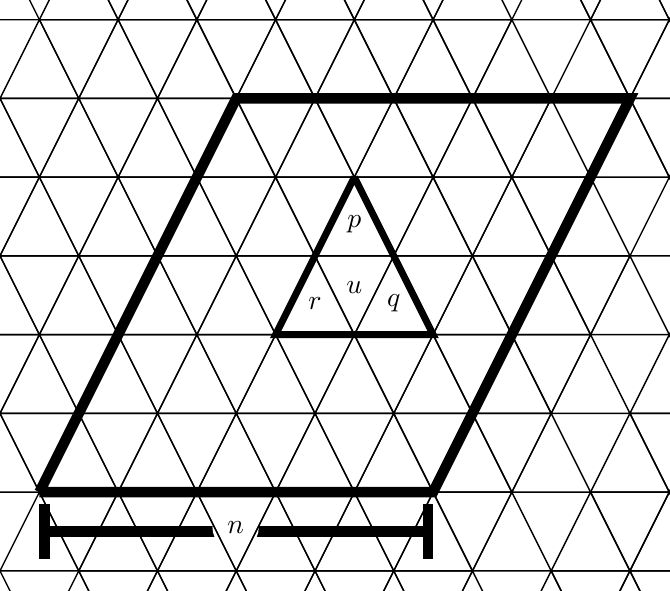}  
      \caption{Triangular grid. Cells $p$, $q$ and $r$ are the neighbors of the cell $u$.}
      \label{fig: grid3}
  \end{subfigure}%
  \hfill  
  \begin{subfigure}[t]{0.475\textwidth}
    \centering
      \includegraphics[width=.95\textwidth]{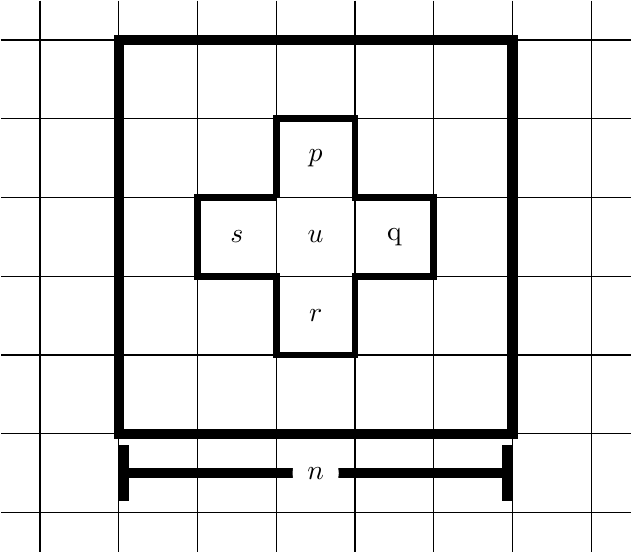}  
      \caption{Square grid. Cells $p$, $q$, $r$ and $s$ are the neighbors of the cell $u$.}
      \label{fig: grid4}
  \end{subfigure}%
  \caption{ Triangular (a) and square (b) grids with the von Neumann neighborhood of a cell $u$.}     
  \label{fig: grids}
\end{figure} 

Each cell in a grid has two possible \emph{states}, which are denoted $0$ and $1$. We say that a site in state $1$ is \emph{active} and a site in state $0$ is \emph{inactive}. A configuration of a grid (triangular or square) is a function that assigns a state to every cell. In a square grid, a \emph{finite configuration} $x$ of dimension $n\times n$  is a function that assigns values in $\{0,1\}$ to square-shaped area of $n^2$ cells. Analogously, in a triangular grid, a finite configuration $x$ of dimension $n\times n$ is a function that assigns values in $\{0,1\}$ to a rhomboid shaped area of $2n^2$ cells. The value of the cell $u$ in the configuration $x$ is denoted $x_u$ (See Figure \ref{fig: grids}). We remark that a finite configuration $x$ of dimension $n\times n$ has $2n^2$ cells in a triangular grid and has $n^2$ cells in a square grid. In both cases the number of covered cells is $\cO(n^2)$. 

Given a finite configuration $x$ of dimension $n\times n$, the periodic configuration $c=c(x)$ is an infinite configuration over the grid, obtained by repetitions of $x$ in all directions. The configuration $c(x)$ is a spatially periodic, and will be interpreted as a torus, where each cell in the boundary of $x$ has a neighbor placed in the opposite boundary of $x$. 

We call $\mathcal{C}$ is the set of all possible configurations over a (triangular or square) grid.
A \emph{cellular automaton} (CA)  with  set of states $\{0,1\}$ is a function $F:\mathcal{C} \to \mathcal{C}$, defined by a \emph{ local function} $f:\{0,1\}^{N(0,0)}\to \{0,1\}$ as $F(c)_u=f(c_{N(u)})$. 
Computing $F$ is equivalent to compute synchronously in each site of the grid, the application of the local function $f$ cell by cell.
A cellular automaton is called \emph{freezing} \cite{goles:hal-01294144}  (FCA) if the local rule $f$ satisfies that the active cells always remains active. A cellular automaton is called \emph{totalistic} \cite{RevModPhys.55.601} (TCA) if the local rule $f$ satisfies $f(c_{N(u)})=f(c_{u},\sum_{v\in N(u)}c_{v})$, i.e. it depends only on the sum of the states in the neighborhood of a cell.

We call FTCA the family of two-state freezing totalistic cellular automata, over the square and triangular grids, with von Neumann neighborhood. In this family, the active cells remain active, because the rule is freezing, and the inactive cells become active depending only in the sum of their neighbors. Notice that this sum of the states of the neighbors of a site is at most the size of the neighborhood, that we call $|N(0,0)|$, and equals $3$ in the case of the triangular grid, and $4$ in the case of the square grid. 

Let $F$ be a FTCA. We can identify $F$ with a set $\mathcal{I}_F \subseteq \{1, \dots, |N(0,0)|\}$ such that, for every configuration $c$ and site $u$: 
$$f(c_{N(u)}) = \left\{ \begin{array}{cl} 1 & \textrm{if } (c_u = 1) \vee (\sum_{v\in N(u)}c_{v}\in \mathcal{I}_F), \\ 0 & \textrm{otherwise.}  \end{array} \right.$$
Notice that  $\mathcal{I}_F\subseteq \{0,1,2,3\}$ in the triangular grid and $\mathcal{I}_f\subseteq \{0,1,2,3,4\}$ in the square grid. 
We will name the FTCAs according to the elements contained in $\mathcal{I}_F$, as the concatenation of the elements of $\mathcal{I}_F$ in increasing order (except when $\mathcal{I}_F = \emptyset$, that we call $\phi$). 
For example, let $Maj$ be the freezing majority vote CA, where an inactive cell becomes active if the majority of its neighbors is active. Note that $\mathcal{I}_{Maj} = \{2,3\}$ in the triangular grid and  $\mathcal{I}_{Maj} = \{2,3,4\}$ in the square grid. We call then $Maj$ the rule $23$ in the first case and $234$ in the later.

We deduce that there are $2^{|N(0,0)|}$ different FTCA, each one of them represented by the corresponding set $\mathcal{I}_F$. Notice that the number of different FTCA is $16$ in the triangular grid and $32$ different in the square grid. We will focus our analysis in the FTCAs where the inactive state is a quiescent state, which means that the inactive sites where the sum of their neighborhoods is $0$ remain inactive. Therefore, we will consider $8$ different FTCA in the triangular grid, and $16$ in the square grid.

Recall that in an FTCA the active cells remain always active. We will be interested in the inactive cells that always remain inactive.

\begin{definition}
Given a configuration $c \in \{0,1\}^{\mathbb{Z}^2}$, we say that a site $v$ is \emph{stable} if and only if $c_v= 0$ and it remains inactive after any iterated application of the rule, i.e., $F^t(c)_v= 0$ for all $t\geq 0$.
\end{definition}

From the previous definition, we consider the {\stability} problem, which consists in deciding if a cell on a periodic configuration $c$ is stable. More formally, if $F$ is a cellular automaton, then:

\medskip

\defproblemu{\stability}{A finite configuration $x$ of dimensions $n\times n$ and a site $u \in [n]\times [n]$ such that $x_u = 0$. }{Is $u$ stable for configuration $c=c(x)$? }

\medskip

In other words, the answer of $\stability$ is \emph{no} if there exists $T>0$ such that $F^T(c(x))_u = 1$.
Our goal is to understand the difficulty of $\stability$ in terms of its computational complexity, for every FTCA  defined over a triangular or square grid.
We consider two classes of problems: \Pt and \NC.

The class \Pt is the class of problems that can be solved by a deterministic Turing machine in  time $n^{\cO(1)}$, where $n$ is the size of the input.
Let $F$ be a freezing cellular automaton (FCA) and $x$ be a finite configuration of dimensions $n\times n$ cells. Notice that the dynamics of $F$ over $c(x)$ reach a \emph{fixed point} (a configuration $c'$ such that $F(c')=c'$) in $\cO(n^2)$ steps. Indeed, after each application of $F$ before reaching the fixed point, at least one inactive site become active in each copy of $x$. The application of one step of any FCA can be simulated in polynomial time, simply computing the local function of every cell. Therefore, for every FCA (and then for every FTCA) $F$  the $\stability(F)$ problem is in~\Pt.

The class {\NC} is a subclass of {\Pt}, consisting of all problems  solvable by a \emph{fast-parallel algorithm}. A fast-parallel algorithm is one that runs in a parallel random access machine (PRAM) in poly-logarithmic time (i.e. in time $(\log{n})^{\cO(1)}$)  using $n^{\cO(1)}$ processors.
It is direct that $\NC \subseteq \Pt$, and it is a wide-believed conjecture that the inclusion is proper \cite{book:879825}. Indeed,  $\NC = \Pt$ would imply that for any  problem solvable in polynomial time, there is a parallel algorithm solving that problem \emph{exponentially faster}. Back in our context, the fact that for some FTCA the  $\stability$ problem belongs to {\NC} will imply that one can solve the problem significantly faster than simply simulating the steps of the automaton.

The problems in {\Pt} that are the most likely to not belong to {\NC} are the {\Pt}-Complete problems. A problem $p$ is {\Pt}-Complete if it is contained in {\Pt} and every other problem in {\Pt} can be reduced to $p$ via a function computable in logarithmic-space. For further details we refer to the book of \cite{book:879825}.

\subsection{Some graph terminology}

 For a  set of cells $S \subseteq \mathbb{Z}^2$, we call $G[S] = (S, E)$ the graph defined with vertex set $S$, where two vertices are adjacent if the corresponding sites are neighbors for the von Neumann neighborhood. 

For a graph $G = (V,E)$, a sequence of vertices $P = v_1, \dots, v_k$ is called a $v_1, v_k$- \emph{path} if $\{v_i, v_{i+1}\}$ is an edge of $G$, for each $i \in [k-1]$. Two $u, v$-paths $P_1$, $P_2$ are called \emph{disjoint} if  $P_1 \cap P_2 = \{u, v\}$. A $u,v$-path where $u$ and $v$ are adjacent is called a \emph{cycle}.

\begin{definition} A graph $G$ is called $k$-\emph{connected} if for every pair of vertices $u,v \in V(G)$, $G$ contains $k$ disjoint $u,v$-paths. A $1$-connected graph is simply called \emph{connected}, a $2$-connected graph is called \emph{bi-connected} and a $3$-connected graph is called \emph{tri-connected}
\end{definition}

A maximal set of vertices of a graph $G$ that induces a $k$-connected subgraph is called a \emph{$k$-connected component} of $G$.

Two cells $v_1, v_2$ are at \emph{distance} $r$ if a shortest path connecting $v_1$ and $v_2$ is of length $r$.  The \emph{ball of radius $r$ centered in $u$}, denoted $B_r(u)$, is the the set of all cells at distance at most $r$ from $u$. On the other hand, the \emph{disc of radius $r$ centered in $u$}, denoted $D_r(u)$, is the set of all cells at distance exactly $r$ from $u$. Observe that $D_r(u) = B_r(u) \setminus B_{r-1}(u)$. When $u$ is the cell at the origin (the cell with coordinates $(0,0)$), these sets are denoted $B_r$ and $D_r$, respectively.

\subsection{Parallel subroutines}

In this subsection, we will give some \NC algoirhtms that we will use as subroutines of our fast-parallel algorithm solving \stability.

\subsubsection{Prefix-sum}
First, we will study a general way to compute in \NC called \emph{prefix sum algorithm} \cite{JaJa:1992:IPA:133889}.
Given an associative binary operation $*$ defined on a group $G$, and an array $A = (a_1, \dots, a_n)$ of $n$ elements of $G$, the prefix sum of $A$ is the vector $B$ of dimension $n$ such that $B_i = a_1 * \dots * a_i$. 
Computing the prefix sum of a vector is very useful. For example, it can be used to compute the parity of a Boolean array, the presence of a nonzero coordinate in an array, etc.

\begin{proposition}[\cite{JaJa:1992:IPA:133889}]\label{prop: prefix-sum}
There is an algorithm that computes the prefix-sum of an array of $n$ elements in time $\mathcal{O}(\log n)$ with $\mathcal{O}(n)$ processors.
\end{proposition}

\subsubsection{Connected components}

The following propositions state that the connected, bi-connected and tri-connected components of an input graph $G$ can be computed by fast-parallel algorithms. 
\begin{proposition}[\cite{JaJa:1992:IPA:133889}]\label{prop:conjaja}
There is an algorithm that computes the connected components of a graph with $n$ vertices in time $\mathcal{O}(\log^2n)$ with $\mathcal{O}(n{^2}) $ processors.
\end{proposition}

\begin{proposition}[\cite{DBLP:journals/siamcomp/JaJaS82}]\label{prop:bijaja}
There is an algorithm that computes the bi-connected components of a graph with $n$ vertices in time $\mathcal{O}(\log^2n)$ with $\mathcal{O}(n{^3}/\log n) $ processors.
\end{proposition}

\begin{proposition}[\cite{DBLP:journals/siamcomp/JaJaS82}]\label{prop:trijaja}
There is an algorithm that computes the tri-connected components of a graph in time $\mathcal{O}(\log^2n)$ with $\mathcal{O}(n{^4})$ processors.
\end{proposition}


\subsubsection{Vertex level algorithm}

Given a rooted tree we are interested in computing the level {\it level}($v$) of each vertex $v$, which is the distance (number of edges) between $v$ and the root $r$. The following proposition shows that there is a fast-parallel algorithm that computes the level of every vertex of the graph.  


\begin{proposition}[\cite{JaJa:1992:IPA:133889}]\label{prop:vertex-level}
There is an algorithm that computes, on an input rooted tree $(T,r)$ the $level(v)$ of every vertex $v\in V(T)$ in time $\mathcal{O}(\log n)$ and using $\mathcal{O}(n) $ processors, where $n$ is the size of $T$.
\end{proposition}

\subsubsection{All pairs shortest paths}

Given a graph $G$ of size $n$. Name $v_1, \dots, v_n$ the set of vertices of $G$. A matrix $B$ is called an \emph{All Pairs Shortest Paths matrix} if $B_{i,j}$ corresponds to the length of a shortest path from vertex $v_i$ to vertex $v_j$. The following proposition states that there is a fast-parallel algorithms computing an All Pairs Shortest Path matrix of an input graph $G$.

\begin{proposition}[\cite{JaJa:1992:IPA:133889}]\label{prop:shortest-paths}
There is an algorithm that computes all Pairs Shortest Paths matrix of a graph with $n$ vertices in time $\mathcal{O}(\log^2n)$ with $\mathcal{O}(n{^3}\log n) $ processors.
\end{proposition}

\section{Triangular Grid}%

 We will start our study over the regular grid where each cell has three neighbors (see Figure \ref{fig: grid3}). In this topology, the sixteen FTCA are reduced to eight non-equivalent, considering the inactive state as a quiescent state. According to our classifications, the eight FTCAs in the triangular grid are grouped as follows:

\begin{itemize}
\item Simple rules: $\phi$, $123$ and $3$.
\item Topological rules: $2$ and $23$.
\item Algebraic rule: $12$.
\item Fractal growing rules: $1$ and $13$.
\end{itemize}

It is easy to check that Simple rules are in \NC. For rule $\phi$, we note that every configuration is a fixed point (then \stability for this rule is trivial). For rule $123$, no site is stable unless the configuration consists in every cell inactive.  We can check in time $\cO(\log n)$ and $\cO(n^2)$ processors whether a configuration contains an active cell using a prefix-sum algorithm (sum the states of all cells, and  then decide if the result is different than $0$).  Finally, for rule $3$ we notice that all dynamics reach a fixed point after one step. Therefore, we check if the initial neighborhood of site $u$ makes it active in the first step (this can be decided in $\cO(\log n)$ time in a sequential machine). 


\subsection{Topological Rules}

We say that rules $2$ and $23$ are \emph{topological} because, as we will see, we can characterize the stable sites according to some topological properties of the initial configurations.

As we mentioned before, rule $23$ is a particular case of the freezing majority vote CA (that we called $Maj$). In \cite{goles:hal-00914603} the authors show that \stability for Maj is in {\NC} over any graph with degree at most $4$. This result is based on a characterization of the set of stable cells, that can be verified by a fast-parallel algorithm.
Thus we can apply this result to solve $\stability$ for rule $23$, considering the triangular grid as a graph of degree $3$. 
Then we have the next theorem:

\begin{theorem}[\cite{goles:hal-00914603}]\label{thm: SyncStability 23}
  There is a fast-parallel algorithm that solves \stability for 23 in time $\cO(\log^2 n)$ and $\cO(n^4)$  processors. 
  Then \stability for 23 is in \NC. 
\end{theorem}

\begin{figure}
  \centering
  \begin{subfigure}[t]{0.475\textwidth}
    \centering
      \includegraphics[width=.95\textwidth]{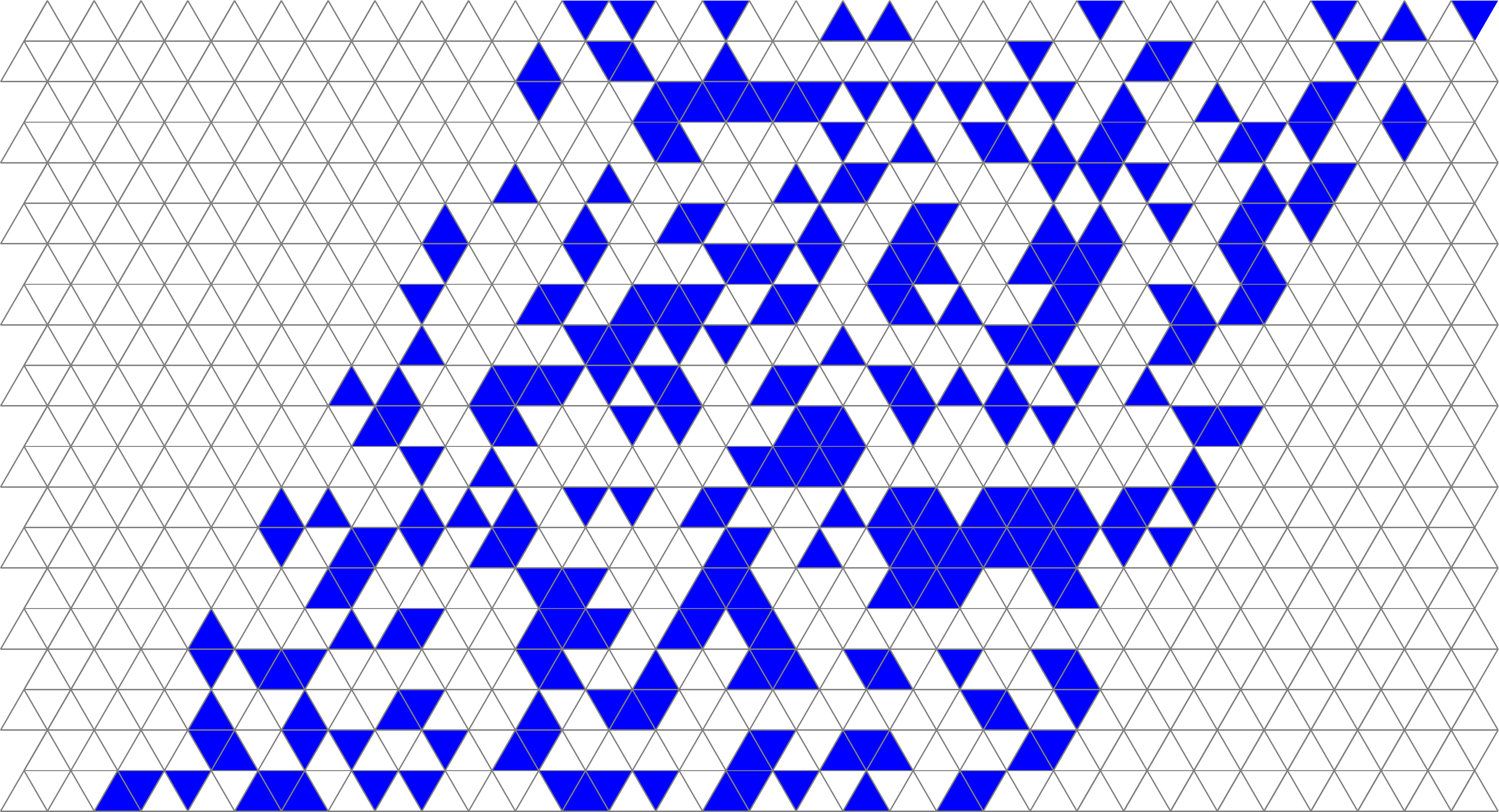}  
      \caption{Initial random configuration.}
  \end{subfigure}%
  \hfill  
  \begin{subfigure}[t]{0.475\textwidth}
    \centering
      \includegraphics[width=.95\textwidth]{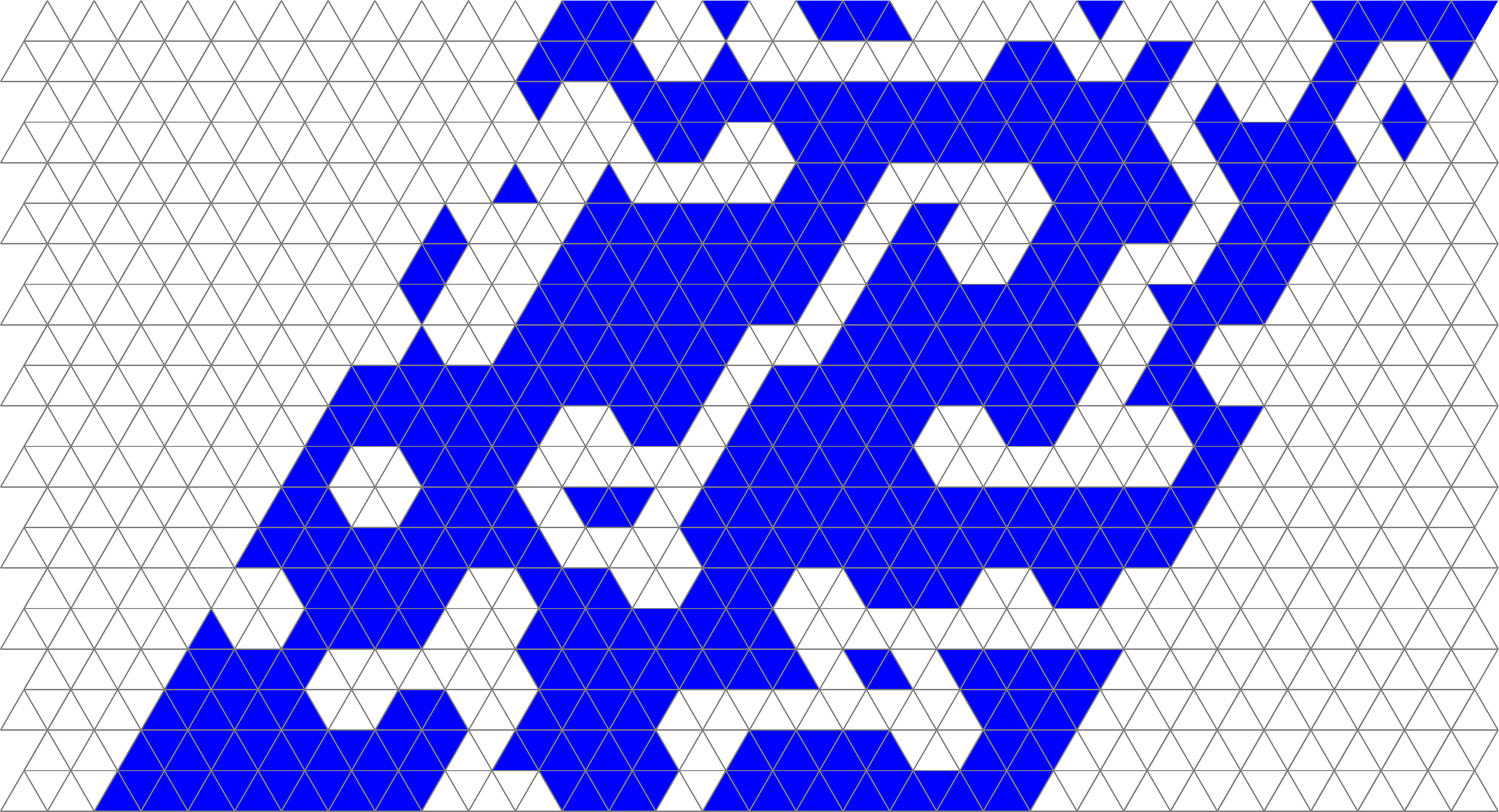}  
      \caption{Time 9 (fixed point).}
  \end{subfigure}%
  \caption{Example of fixed point for the rule $23$. The cells in state 0 in the fixed point are stable cells.}     
\end{figure} 

For the sake of completeness, we give the main ideas used to prove Theorem~\ref{thm: SyncStability 23}. The main idea is a characterization of the set of stable sites. 

\begin{proposition}[\cite{goles:hal-00914603}]\label{prop:StableMaj}
Let $Maj$ be the freezing majority vote CA defined over a graph $G$ of degree at most $4$. Let $c$ be a configuration of $G$, and let $G[0]$ be the subgraph of $G$ induced by the vertices (cells) which are inactive according to $c$. 

An inactive vertex $u$ is stable if and only if,
\begin{itemize}
\item[(i)] $u$ belongs to a cycle in $G[0]$, or 
\item[(ii)] $u$ belongs to a path $P$ in $G[0]$ where both endpoints of $P$ are contained in cycles in $G[0]$. 
\end{itemize}
Moreover, there is a fast-parallel algorithm that checks conditions (i) and (ii) in time $\cO(\log^2 n)$ using $\cO(4^2)$ processors.
\end{proposition}

Therefore, the proof of Theorem~\ref{thm: SyncStability 23} consists in (1) notice that a finite configuration on the triangular grid,  seen as a torus, is a graph of degree 3 (then in particular is a graph of degree at most $4$); (2) use the algorithm given in Proposition \ref{prop:StableMaj} to check whether the given site $u$ is stable.

 We will use the previous result to solve the stability problem for rule $2$.

\begin{theorem}\label{thm: SyncStability 2}
  \stability is in {\NC} for rule $2$.
\end{theorem}
 
\begin{proof}
When we compare rule $2$ and rule $23$, we noticed that they exhibit quite similar dynamics. Indeed, a cell $u$ which is stable for rule $23$ is also stable for rule $2$.  Therefore, to solve \stability for 2 on input configuration $x$ and cell $u$, we can first solve \stability for 23 on those inputs using the algorithm given by Theorem \ref{thm: SyncStability 23}. When the answer of \stability for 23 is \emph{Accept}, we know that \stability for 2 will have the same answer. In the following, we focus in the case where the answer of \stability for 23 is \emph{Reject}, i.e. $u$ is not stable on configuration $x$ in the dynamics of rule $23$.

Suppose that $u$ is stable for rule $2$, but it is not stable for rule $23$. Let $t$ be the first time-step where $u$ becomes active in the dynamics of rule $23$. Note that, since $u$ is stable for rule $2$, necessarily in step $t-1$ the three neighbors of $u$ are active.  Moreover, at least two of them simultaneously became active in time $t-1$. 

Let now $G$ be the graph representing the cells of the triangular grid covered by configuration $x$. Let $G[0]$ be the subgraph of $G$ induced by the initially inactive cells, and let $G[0,u]$ be the connected component of $G[0]$ containing cell $u$. We claim that $G[0,u]$, in the dynamics of rule $23$, every vertex (cell) in $G[0,u]$ must become active before $u$, i.e. in a time-step strictly smaller than $t$. \\

{\bf Claim 1:} Every vertex of $G[0,u]$ is active after $t$ applications of rule $23$.\\

Indeed, suppose that there exists a vertex (cell) $w$ in $G[0,u]$ that becomes active in a time-step greater than $t$. Call $P$ a shortest path in $G[0, u]$ that connects $u$ and $w$, and let $u^*$ be the neighbor of $u$ contained in $P$. Note that except the endpoints, all the vertices (cells) in $P$ have at least two neighbors in $P$, which are inactive. Moreover, both endpoints of $P$ are inactive at time $t$. Therefore, all the vertices in $P$ will be inactive in time $t$. This contradicts the fact that the three neighbors of $u$ become active before $u$.  \\

{\bf Claim 2:} $G[0,u]$ is a tree.\\

Indeed, $G[0,v]$ is connected, since it is defined as a connected component of $G[0]$ containing $u$. On the other hand, suppose that $G[0,u]$ contains a cycle $C$. From Proposition \ref{prop:StableMaj}, we know that all the cells in $C$ are stable, which contradicts Claim 1. \\

Call $T_u$ the tree $G[0,u]$ rooted on $u$. Let $d$ be the depth of $T_u$, i.e. longest path between $u$ and a leaf of $T_u$. \\

{\bf Claim 3:} Every vertex of $G[0,u]$, except $u$, is active after $d$ applications of rule $2$. \\

Notice that necessarily a leaf of $T_u$ has two active neighbors (because they are outside $G[0,u]$) and one inactive neighbor (its parent in $T_u$). Therefore, in one application of rule $2$, all the leafs will become active. We will reason by induction on $d$. Suppose that $d=1$. Then all vertices $w$ of $T_u$ except $u$ are leafs, so the claim is true. Suppose now that the claim is true for all trees of depth smaller or equal than $d$, but $T_u$ is a tree of depth $d+1$. We notice in one step the leafs are the only vertices of $T_u$ that become active (every other vertex has two inactive neighbors). Then, after one step, the inactive sites of $T_u$ induce a tree $T_u'$ of depth $d$. By induction hypothesis, all the cells in $T'_u$, except $u$, become active after $d$ applications of rule $2$.  We deduce the claim. 

Let $u_1, u_2, u_3$ be the three neighbors of $u$. For $i \in \{1,2,3\}$, call $T_{u_i}$ the subtree of $T_u$ rooted at $u_i$, obtained taking all the descendants of $u_i$ in $T_u$. Call $d_i$ the depth of $T_{u_i}$. Without loss of generality, $d_1 \geq d_2 \geq d_3$.\\

{\bf Claim 4:} $u$ is stable for the dynamics of rule $2$ but not for the dynamics of rule $23$, if and only if $d_1 = d_2 \geq d_3$.\\

Recall that $u$ is stable for the dynamics of rule $2$ but not for the dynamics of rule $23$ if and only if $u$ has three active neighbors at time-step $t$, and at least two of them become active at time $t-1$. The claim follows from the application of Claim 3 to trees $T_{u_1}, T_{u_2}$ and $T_{u_3}$.

We deduce the following fast-parallel algorithm solving \stability for 2: Let $x$ be the input configuration and $u$ the cell that we want to decide stability. First, use the fast-parallel algorithm given by Theorem~\ref{thm: SyncStability 23} to decide if $u$ is stable for the dynamics of rule $23$ on configuration $x$. If the answer is affirmative, then we decide that $(x,u)$ is a \emph{Accept}-instance of \stability for 2. If the answer is negative, the algorithm looks for cycles in $G[0,u]$. If there is a cycle, then the algorithm \emph{Rejects}, because Claim 2 implies that $u$ cannot be stable for rule $2$. If $G[0,u]$ is a tree, then the algorithm computes in parallel the depth $d_v$ of the subtrees $T_{v}$, for each $v \in N(u)$. Finally, the algorithm accepts if the conditions of Claim 4 are satisfied, and otherwise rejects.  The steps of the algorithm are represented in Algorithm \ref{alg: SyncStability 2}.

\begin{algorithm}[h]
\caption{Solving \stability 2}\label{alg: SyncStability 2}
\begin{algorithmic}[1]
\REQUIRE $x$ a finite configuration of dimensions $n \times n$ and $u$ a cell.

\IF{the answer of $\stability$ for rule $23$ is \emph{Accept} on input $(x,u)$}
  \RETURN \emph{Accept}
\ELSE
\STATE Compute $G[0,u]$
\STATE Compute $C$ the set of cycles of $G[0,u]$
\IF{ $C \neq \emptyset$}
\RETURN  \emph{Reject}
\ELSE
\FORALLP{ $v \in N(u)$}
  \STATE Compute $d_v$ the depth of $T_v$
\ENDFORALLP
\IF{$\exists a,b,c\in (N(u)): d_a = d_b \geq d_c$}
    \RETURN \emph{Accept}
 \ELSE 
 \RETURN \emph{Reject}
\ENDIF
\ENDIF
\ENDIF
\end{algorithmic}
\end{algorithm}

Let $N=n^{2}$ the size of the input. Algorithm \ref{alg: SyncStability 2} runs in time $\cO(\log^2 N)$ using $\mathcal{O}(N{^3}/\log N) $ processors. 
Indeed, the condition of line 1 can be checked in time $\cO(\log^2 N)$ using $\cO(N^2)$ processors according the algorithm of Theorem \ref{thm: SyncStability 23}. 
Step 4 can be done in time $\cO(\log^2 N)$ using $\cO(N^2)$ processors using a connected components algorithm given in \cite{JaJa:1992:IPA:133889}. 
Step 5 an be done in time $\cO(\log^2 N)$ using $\mathcal{O}(N{^3}/\log N) $ processors using a bi-connected components algorithm given in \cite{JaJa:1992:IPA:133889}. 
Step 10 can be solved in time $\mathcal{O}(\log N)$ using $\mathcal{O}(N) $  processors using a vertex level algorithm given in \cite{JaJa:1992:IPA:133889}. 
Finally, Step 12 can be done in $\cO(\log N)$ time in a sequential machine. 
\end{proof}

\subsection{Algebraic Rule}

We now continue with the study of rule $12$. We say that this rule is \emph{algebraic} because, as we will see, we can speed-up its dynamics using some of its algebraic properties. This speed-up will provide an algorithm that decides the stability of a cell much faster than simple simulation. In other words, we will show that $\stability$ for rule $12$ is in \NC. 

Let $x$ be a finite configuration on the triangular grid, $u$ a cell. Let $v$ be a neighbor of $u$. We define a \emph{semi-plane} $S_v$ as a partition of the triangular grid in two parts, cut by the edge of the triangle that share cell $u$ and $v$, as shown in Figure~\ref{fig: tri semi-plane}.

\begin{figure}[h]
        \centering
\includegraphics[width=.5\textwidth]{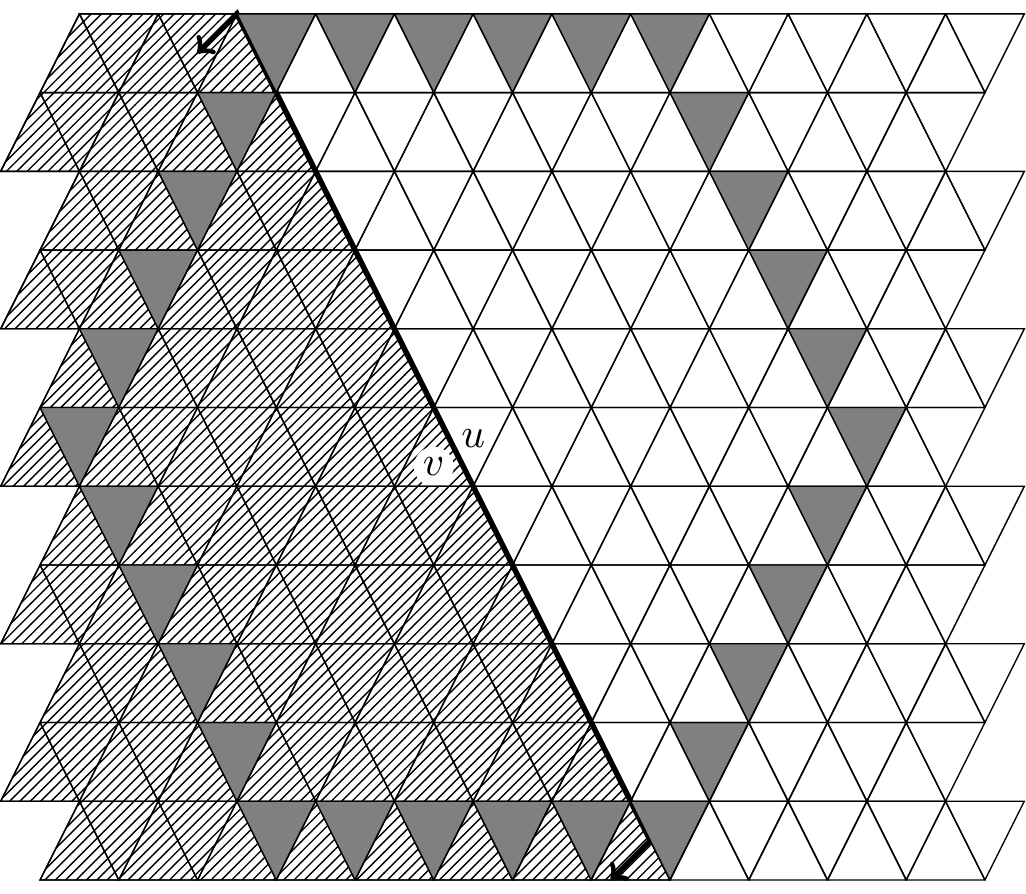}
\caption{Triangular grid divided in semi-planes according to $u$ and $v$. The hatched pattern represent the semi-plane $S_v$
Gray cells are at the same distance from $u$. }
\label{fig: tri semi-plane}
\end{figure}



\begin{lemma}\label{lem: spread}
Let $d\geq 2$ be the distance from $u$ to the nearest active cell. Then the distance to the nearest cell to $u$ in $F(c)$ is $d-1$
\end{lemma}

\begin{proof}
Let $w$ be an active cell at distance $d$ of $u$ in configuration $x$, and call $P$ a shortest $u,w$-path. Call $w_1$ the neighbor of $w$ contained in $P$, and let $w_2$ be the neighbor of $w_1$ in $P$ different than $w$ (this cell exists since $d\geq 2$). Note that $w_2$ might be equal to $u$. Since $P$ is a shortest path, $w_2$ is at distance $d-2$ from $u$. Then all the neighbors of $w_2$ are inactive, so $w_2$ it is necessarily inactive in $F(c)$. Moreover, $w_1$ has more than one active neighbor, and less than three active neighbors, so $w_1$ is active in $F(c)$. Then the distance from $u$ to the nearest active cell in $F(c)$ is $d-1$.
\end{proof}


Recall that $D_r(u)$ is the set of cells at distance $r$ from $u$. We deduce the following lemma.

\begin{lemma}\label{lem: or tec tri}
  Let $d\geq 2$ be the distance from $u$ to the nearest active cell, and let $v\in N(u)$. 
 Then $v$ is active after $d-1$ applications of rule $12$ (i.e. $F^{d-1}(c)_v = 1$) if  and only there exists an active cell in $S_v \cap D_d(u)$
\end{lemma}

\begin{proof}
We reason by induction on $d$. In the base case, $d=2$, suppose that $S_v$ does not contain an active site at distance $2$. Then every neighbor of $v$ is inactive in the initial configuration, so $v$ is inactive after one application of rule $12$ (i.e. $F(c)_v  = 0$). Conversely, if $F(c)_v = 0$, then every neighbor of $v$ is initially inactive, in particular all the sites in $S_v$ at distance $2$ from $u$. 

Suppose now that the statement of the lemma is true on configurations where the distance is $d$, and let $c$ be a configuration where the distance from $u$ to the nearest active cell is $d+1$. Let $c'$ be the configuration obtained after one application on $c$ of rule $12$ (i.e. $c' = F(c)$). \\

\noindent{\bf Claim 1:} $F^{d-1}(c')_v  = 1$ if and only if in $c'$ there exists an active cell in $S_v \cap D_d(u)$.\\

From Lemma \ref{lem: spread}, the distance from $u$ to the nearest active cell in $c'$ is $d$. The claim follows from the induction hypothesis. \\

\noindent{\bf Claim 2:} Suppose that $F^{d-1}(c')_v = 0$. Then in $c$, all the cells in $S_v \cap D_{d+1}(u)$ are inactive. \\

Notice that, from Claim 1, the fact that $F^{d-1}(c')_v=0$ implies that in $c'$ all the cells in $D_d(u) \cap S_v$ must be inactive. Suppose, by contradiction, that there is a cell $w$ in $S_v \cap D_{d+1}(u)$ that is active in $c$. Let $w'$ be a neighbor of $w$ contained in $S_v \cap D_d(u)$, and let $w''$ be a neighbor of $w'$ not contained in $D_{d+1}(u)$ (then $w'$ belongs to $D_d(u) \cup D_{d-1}(u)$). Note that $w'$ has an active neighbor in $c$, but must be inactive in $c'$. The only option is that all the neighbors of $w'$ are active in $c$, in particular $w''$ is active in $c$. This contradicts the fact the nearest active cell is at distance $d+1$ in $c$.\\

\noindent{\bf Claim 3:} Suppose that $F^{d-1}(c')_v = 1$. Then there is a cell in $ S_v \cap D_{d+1}(u)$  that is active in $c$. \\

From Claim 1, the fact that $F^{d-1}(c')_v=0$ implies that there is a cell $w \in S_v \cap D_d$ that is active in $c'$. Suppose by contradiction that all the cells in $S_v \cap D_{d+1}(u)$ are inactive in $c$. Since $w$ is active in $c'$, necessarily $w$ has at least one neighbor $w'$ that is active in $c$. Since $w'$ is not contained in $S_v \cap D_{d+1}(u)$ (because we are supposing that all those cells are inactive in $c$), we deduce that $w'$ belongs to $D_{d}(u) \cup D_{d-1}(u)$. This contradicts the fact the nearest active cell is at distance $d+1$ in $c$.\\

We deduce that $F^{d-1}(c')_v = 1$ if and only if there is a cell in $ S_v \cap D_{d+1}(u)$ that is active in $c$.  Since $c' = F(c)$, we obtain that $F^{d}(c)_v = 1$ if and only if there is a cell in $ S_v \cap D_{d+1}(u)$ that is active in $c$.
\end{proof}

\begin{theorem}\label{thm: SyncStability or}
  \stability is in {\NC} for rule $12$.
\end{theorem}
\begin{proof}

In our algorithm solving \stability for 12, we first compute the distance $d$ to the nearest active cell from $u$ (if every cell is inactive, our algorithm trivially accepts). Then, for each $v \in N(u)$, the algorithm computes the set of cells $S_v \cap D_d(u)$, and checks if that set contains an active cell. If it does, we mark $v$ as \emph{active}, and otherwise we mark $v$ as \emph{inactive}. Finally, the algorithm rejects if the three neighbors of $u$ are active, and accepts otherwise. The steps of this algorithm are described in Algorithm \ref{alg: SyncStability 12}.

From Lemma \ref{lem: or tec tri}, we know that $v$ becomes active at time $d-1$ if and only if $S_v \cap D_d(u)$ contains an active cell in the initial configuration. Since the nearest active cell from $u$ is at distance $d$, necessarily after $d-1$ steps at least one of the three neighbors of $u$ will become active. If the three neighbors of $u$ satisfy the condition of Lemma \ref{lem: or tec tri}, then the three of them will become active in time $d-1$, so $u$ will remain inactive forever. Otherwise, $u$ will have more than one and less than three active neighbors at time-step $d-1$, so it will become active at time $d$.

  \begin{algorithm}[h]
\caption{Solving \stability 12}\label{alg: SyncStability 12}
\begin{algorithmic}[1]
\REQUIRE $x$ a finite configuration of dimensions $n \times n$ and $u$ a cell.
\IF {For all cell $w$, $x_w = 0$}
\RETURN \emph{accept}
\ELSE
\STATE Compute a matrix $M = (m_{ij})$ of dimensions $2n^2 \times  2n^2$ such that 

$m_{ij}$ is the distance from cell $i$ to cell $j$.
\STATE Compute the distance $d$  to the nearest active cell from $u$.
\FORALLP{ $v \in N(u)$}
  \STATE Compute the set of cells $S_v \cap D_d(u)$
  \IF {there exists $w \in S_v \cap D_d(u)$ such that $x_w =1$}
  \STATE Mark $v$ as \emph{active}
  \ELSE
  \STATE Mark $v$ as \emph{inactive}
  \ENDIF
\ENDFORALLP
\IF{ there exists $ v$ in $N(u)$ that is marked inactive}
  \RETURN \emph{Reject}
  \ELSE
  \RETURN \emph{Accept}
\ENDIF
\ENDIF
\end{algorithmic}
\end{algorithm}

Let $N=n^{2}$ the size of the input.
This algorithm runs in time $\cO(\log N)$ using $\cO(N)$ processors. Indeed, the verifications on lines 1-3 and 8-10 can be done in time $\cO(\log N)$ using $\cO(N)$ processors using a prefix-sum algorithm. 
Finally, step $7$ can be done in time $\cO(\log N)$ using $\cO(N)$ processors, assigning one processor per cell and solving three inequations of kind $ax+by < c$.
\end{proof}

\section{Square Grid}

We now continue our study,  considering the square grid. As we said in the preliminaries section, we can define $32$ different FTCAs over this topology. Again, considering the inactive state as a quiescent state, the set of non-equivalent FTCAs is reduced to $16$.  According to our classifications, this list of FTCAs is grouped as follows:

\begin{itemize}
\item Simple rules: $\phi$, $1234$ and $4$.
\item Topological rules: $234$, $3$ and $34$.
\item Algebraic rules: $12$, $123$, and $124$.
\item Turing Universal rules: $2$, $24$.
\item Fractal growing rules: $1$, $13$, $14$ and $134$.
\end{itemize}

In complete analogy to the triangular topology, we verify that the \stability problem in Simple rules is \NC. We will directly continue then with the Topological Rules. 

\subsection{Topological Rules}

In this subsection, we study rules whose fixed points can be characterized according to some topological properties of the graph induced by initially inactive sites. More precisely, we are going to be interested in characterize the sets of stable cells, that we call \emph{stable sets}. Naturally, the structure of stable-sets depends on the rule.

\subsubsection{Rules 34 and 3.}
First, notice that the rule $34$ corresponds to freezing version of the majority automaton (Maj) over the square grid.  We remark that a finite configuration over the square grid, seen as a torus, is a regular graph of degree four.  Proposition \ref{prop:StableMaj}  that in this case the stable sets are cycles or paths between cycles in the graph induced by initially inactive sites. Moreover, Therefore, we can use the Algorithm given in  Proposition \ref{prop:StableMaj}   to check whether a given site is stable for rule $34$. We deduce the following theorem (also given in \cite{goles:hal-00914603}).

\begin{theorem}[\cite{{goles:hal-00914603}}]\label{thm:rule 34}
  \stability is in \NC for rule $34$.
\end{theorem}

Likewise, in analogy of the behavior of rule $2$ with respect to rule $23$ in the triangular grid, we can use the algorithm solving \stability for the rule $34$ to solve \stability for the rule $3$. Let $(x,u)$ be an instance of the $\stability$ problem. Clearly, if $u$ is stable for rule $34$ we have that $u$ is stable for rule $3$. Suppose now that $u$ is not stable for rule $34$ but it is stable for rule $3$. Let $G[0,u]$ be connected component of $G[0]$ containing $u$. Using the exact same proof used for rule $2$ on the triangular grid, we can deduce that $G[0,u]$ is a tree, and we call $T_u$ this tree rooted in $u$. Moreover, let $u_1, u_2, u_3, u_4$ be the four neighbors of $u$, and let $T_{u_i}$ be the subtree of $T_u$ obtained taking all the descendants of $u_i$, $i \in \{1,2,3,4\}$. Call $d_i$ the depth of $T_{u_i}$, which without loss of generality we assume that $d_1 \geq d_2 \geq d_3 \geq d_4$. We have that $u$ is stable for rule $3$ but not for rule $34$ if and only if  $d_1 = d_2 \geq d_3 \geq d_4$. We deduce that, with very slight modifications, Algorithm \ref{alg: SyncStability 2} solves \stability for rule $3$. We deduce the following theorem.

\begin{theorem}\label{thm:rule 3}
  \stability is in \NC for rule $3$.
\end{theorem}
\subsubsection{Rule 234.}

Notice that rule $234$ is the freezing version of the non-strict majority automaton, 
the CA where the cells take the state of the majority of its neighbors, and in tie case they decide to become active. In the following, we will show that the stability problem for this rule is also in \NC, characterizing the set of stable sets. This time, the topological conditions of the stable sets will be the property of being tri-connected. 

\begin{theorem}\label{thm: 234 square NC}
  \stability is in {\NC} for rule $234$.
\end{theorem}

\begin{lemma} \label{pro: estabilidad general}
  Let  $x\in \{0,1\}^{[n]\times [n]}$ be a finite configuration and $u \in [n]\times [n]$ a site. Then,  $u$ is stable for $c =c(x)$ if and only if there exist a set $S \subseteq [n]\times [n]$ such that: 
  \begin{itemize}
  \item $u \in S$, 
  \item $c_u = 0$ for every $u \in S$, and 
  \item $G[S]$  is a graph of minimum degree $3$. 
  \end{itemize}
\end{lemma}

\begin{proof} 

Suppose that $u$ is stable and  let $S$ be the subset of $[n]\times [n]$ containing all the sites that are stable for $c$. We claim that $S$ satisfy the desired properties. Indeed, since $S$ contains all the sites stable for $c$, then $u$ is contained in $S$. On the other hand,  since the automaton is freezing, all the sites in $S$ must be inactive on the configuration $c$. Finally, if $G[S]$ contains a vertex $v$ of degree less than $3$, it means that necessarily the corresponding site $v$ has two non-stable neighbors  that become $1$ in the fixed point reached from $c$, contradicting the fact that $v$ is stable. 

On the other direction suppose that $S$ contains a site that is not stable and let $t>0$ be the minimum step such that a site $v$ in $S$ changes to state $1$, i.e., $v \in S$ and $t$ are such $F^{t-1}(c)_w = 0$ for every $w \in S$, and $F^t(c)_v=1$. This implies that $v$ has at least two active neighbors in the configuration $F^{t-1}(c)$. This contradicts the fact that $v$ has three neighbors in $S$. We conclude that all the sites contained in $S$ are stable, in particular $u$.\end{proof}


For a finite configuration  $x\in \{0,1\}^{[n]\times[n]}$, let $D(x) \!\in \!\{0,1\}^{\{-n^2 -n, \dots, n^2 + 2n\} ^2}$ be the finite configuration of dimensions $m \times m$, where $m= 2n^2+3n$, constructed with repetitions of configuration $x$ in a  rectangular shape, as is depicted in Figure~\ref{fig:tildex}, and inactive sites  elsewhere.  We also call $D(c)$ the periodic configuration~$c(D(x))$.

\begin{figure}
\centering
\begin{tikzpicture}[scale=.65]



\draw (4,1) grid (11,8);

\node[scale = 2](x1) at (5.5, 4.5){$x$};
\node[scale = 2](x1) at (6.5, 4.5){$x$};
\node[scale = 2](x1) at (7.5, 4.5){$x$};
\node[scale = 2](x1) at (8.5, 4.5){$x$};
\node[scale = 2](x1) at (9.5, 4.5){$x$};

\node[scale = 2](x1) at (5.5, 3.5){$x$};
\node[scale = 2](x1) at (6.5, 3.5){$x$};
\node[scale = 2](x1) at (7.5, 3.5){$x$};
\node[scale = 2](x1) at (8.5, 3.5){$x$};
\node[scale = 2](x1) at (9.5, 3.5){$x$};

\node[scale = 2](x1) at (5.5, 5.5){$x$};
\node[scale = 2](x1) at (6.5, 5.5){$x$};
\node[scale = 2](x1) at (7.5, 5.5){$x$};
\node[scale = 2](x1) at (8.5, 5.5){$x$};
\node[scale = 2](x1) at (9.5, 5.5){$x$};

\node[scale = 2](x1) at (5.5, 2.5){$x$};
\node[scale = 2](x1) at (6.5, 2.5){$x$};
\node[scale = 2](x1) at (7.5, 2.5){$x$};
\node[scale = 2](x1) at (8.5, 2.5){$x$};
\node[scale = 2](x1) at (9.5, 2.5){$x$};

\node[scale = 2](x1) at (5.5, 6.5){$x$};
\node[scale = 2](x1) at (6.5, 6.5){$x$};
\node[scale = 2](x1) at (7.5, 6.5){$x$};
\node[scale = 2](x1) at (8.5, 6.5){$x$};
\node[scale = 2](x1) at (9.5, 6.5){$x$};

\node[scale = 1] at (7,0.5) {$0$};
\node[scale = 1] at (8,0.5) {$n$};
\draw[very thick,<->] (7,8) -- (7,1);
\draw[very thick,<->] (4,4) -- (11,4);
\node[scale = 1] at (10,0.5) {$n^2+n$};
\node[scale = 1] at (5,0.5) {$-n^2$};
\draw[densely dashed,ultra thick]  (7,5) rectangle (8,4);
\end{tikzpicture}

\caption{Construction of the finite configuration $D(x)$ obtained from a finite configuration $x$ of dimension $n \times n = 2 \times 2$. 
Note that $D(x)$ is of dimensions $7 \times 7$.}\label{fig:tildex}
\end{figure}
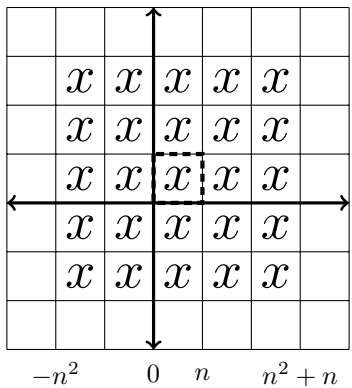

\begin{lemma}\label{lem:candDc}

Let $x \in \{0,1\}^{[n]^2}$ be a finite configuration, and let $u$ be a site in $[n] \times [n]$ such that $x_u = 0$. Then $u$ is stable for $c = c(x)$ if and only if it is stable for $D(c)$. 

\end{lemma}

\begin{proof}
Suppose first that $u$ is stable for $c$, i.e. in the fixed point $c'$ reached from $c$, $c'_u = 0$. Call $c''$ the fixed point reached from $D(c)$. Note that $D(c) \leq c$ (where $\leq$ represent the inequalities coordinate by coordinate). Since the {\FNSM} automata is monotonic, we have that $c'' \leq c'$, so $c''_u = 0$. Then $u$ is stable for $D(c)$. 

Conversely, suppose that  $u \in [n]\times [n]$ is not stable for $c $, and let $S$ be the set of all  sites at distance at most $n^2$ from $u$.  We know that in each step on the dynamics of $c$, at least one site in the periodic configuration changes its state, then in at most $n^2$ steps the site $u$ will be activated. In other words, the state of $u$ depends only on the states of the sites at distance at most $n^2$ from $u$. Note that for every $v \in S$, $c_v = D(c)_v$. Therefore, $u$ is not stable in $D(c)$. \end{proof}

Note that the perimeter of width $n$ of  $D(x)$ contains only inactive sites. We call this perimeter the \emph{border} of $D(x)$, and $D(x) - B$ the \emph{interior} of $D(x)$. Note that $B$ is tri-connected and forms a set of sites stable for $D(c)$ thanks to Lemma~\ref{pro: estabilidad general}. We call $Z$ the set of sites $w$ in $[m]\times [m]$ such that $D(x)_w = 0$.

\begin{lemma}\label{lem:3caminos}

Let $u$ be a site in $[n] \times [n]$ stable for $D(c)$.  
Then, there exist three disjoint paths on $G[Z]$ connecting $u$ with sites of the border $B$. 
Moreover, the paths contain only sites that are stable for $B(c)$.
 
\end{lemma}
\begin{proof} 

Suppose that $ u$ is stable. From Lemma \ref{pro: estabilidad general} this implies that $u$ has three stable neighbors. Let $0 \leq i,j \leq n$ be such that $u = (i, j)$. 
We divide the interior of $D(c)$ in four quadrants: 
\begin{itemize}
\item The first quadrant contains all the sites in $D(x)$ with coordinates at the north-east of $u$, i.e., all the sites $v = (k,l)$ such that $k \geq i$ and $l \geq j$. 
\item The second quadrant contains all the sites in $D(x)$ with coordinates at the north-west of $u$, i.e., all the sites $v = (k,l)$ such that $k \leq i$ and $l \geq j$. 
\item The third quadrant contains all the sites in $D(x)$ with coordinates at the south-west of $u$, i.e., all the sites $v = (k,l)$ such that $k \leq i$ and $l \leq j$. 
\item The fourth quadrant contains all the sites in $D(x)$ with coordinates at the south-east of $u$, i.e., all the sites $v = (k,l)$ such that $k \geq i$ and $l \leq j$. 
\end{itemize}

We will construct three disjoint  paths in $G[Z]$ connecting $u$ with the border, each one passing through a different quadrant. The idea is to first choose three quadrants, and then extend three paths starting from $u$ iteratively picking different stable sites in the chosen quadrants, until the paths reach the border. 

Suppose without loss of generality that we choose the first, second and third quadrants, and  let  $u_1, u_2$ and $u_3$ be three stable neighbors of $u$, named according to Figure \ref{fig:u1u2u3}. 

\begin{figure}
        \begin{centering}
        \hfill
        \begin{subfigure}[b]{0.23\textwidth}
          \centering
          \begin{tikzpicture}
            \node (v2) at (-0.5,0.5) {$u$};
            \node (v1) at (-1.5,0.5) {$u_3$};
            \node (v3) at (0.5,0.5)  {$u_1$};
            \node (v4) at (-0.5,1.5) {$u_2$};
            \node (v5) at (-0.5,-0.5) {$?$};
            \draw (v1) -- (v2);
            \draw (v3) -- (v2) -- (v4) -- (v2) -- (v5) -- (v2);
          \end{tikzpicture}
                \caption{Case 1}
                \label{fig: camino Norte a}
        \end{subfigure}%
        \hfill
        \begin{subfigure}[b]{0.23\textwidth}
          \centering
        \begin{tikzpicture}
            \node (v2) at (-0.5,0.5) {$u$};
            \node (v1) at (-1.5,0.5) {$?$};
            \node (v3) at (0.5,0.5)  {$u_1$};
            \node (v4) at (-0.5,1.5) {$u_2$};
            \node (v5) at (-0.5,-0.5) {$u_3$};
            \draw (v1) -- (v2);
            \draw (v3) -- (v2) -- (v4) -- (v2) -- (v5) -- (v2);
          \end{tikzpicture}      
                \caption{Case 2}\label{fig: camino Norte b}
        \end{subfigure}%
        \hfill
        \begin{subfigure}[b]{0.23\textwidth}
                \centering
        \begin{tikzpicture}
            \node (v2) at (-0.5,0.5) {$u$};
            \node (v1) at (-1.5,0.5) {$u_2$};
            \node (v3) at (0.5,0.5)  {$u_1$};
            \node (v4) at (-0.5,1.5) {$?$};
            \node (v5) at (-0.5,-0.5) {$u_3$};
            \draw (v1) -- (v2);
            \draw (v3) -- (v2) -- (v4) -- (v2) -- (v5) -- (v2);
          \end{tikzpicture}           
                \caption{Case 3}\label{fig: camino Norte c}
        \end{subfigure}%
        \hfill
        \begin{subfigure}[b]{0.23\textwidth}
                \centering
          \begin{tikzpicture}
            \node (v2) at (-0.5,0.5) {$u$};
            \node (v1) at (-1.5,0.5) {$u_2$};
            \node (v3) at (0.5,0.5)  {$?$};
            \node (v4) at (-0.5,1.5) {$u_1$};
            \node (v5) at (-0.5,-0.5) {$u_3$};
            \draw (v1) -- (v2);
            \draw (v3) -- (v2) -- (v4) -- (v2) -- (v5) -- (v2);
          \end{tikzpicture}
                \caption{Caso 4}\label{fig: camino Norte d}
        \end{subfigure}%
        \caption{Four possible cases for $u_1, u_2$ and $u_3$. Note that one of these four cases must exist, since $u$ has at least three stable neighbors.  From $u_1$ we will extend a path through the first quadrant, from $u_2$ a path through the second quadrant, and from $u_3$ a path through the third quadrant.}
        \label{fig:u1u2u3}
        \end{centering}
  \end{figure}
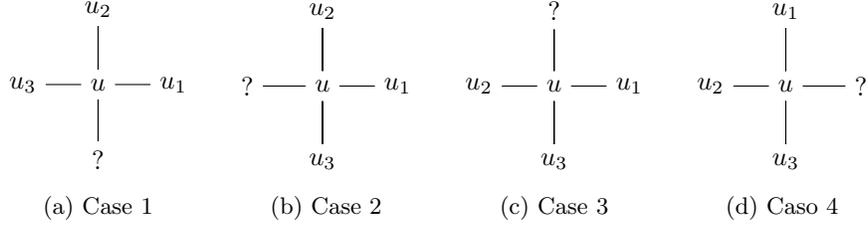
Starting from $u, u_1$, we extend the path $P_1$ through the endpoint different than $u$, picking iteratively a stable site at the east, or at the north if the site in the north is not stable. Such sites will always exist since by construction the current endpoint of the path will be a stable site, and stable sites must have three stable neighbors (so either one neighbor at east or one neighbor at north). The iterative process finishes when $P_1$ reaches the border. Note that necessarily $P_1$ is contained in the first quadrant. Analogously, we define paths $P_2$ and $P_3$, starting from $u_2$ and $u_3$, respectively, and extending the corresponding paths picking neighbors at the north-west or south-west, respectively. We obtain that $P_2$ and $P_3$ belong to the second and third quadrants, and are disjoint from $P_1$ and from each other. 

This argument is analogous for any choice of three quadrants. We conclude there exist three disjoint paths of stable sites from $u$ to the border~$B$. \end{proof}

\begin{lemma}\label{teo: TRI}
Let $u,v$ be two sites in $[n] \times [n]$ stable for $D(c)$.  Then, there exist three disjoint $u,v$-paths  in $G[Z]$ consisting only of  sites that are stable for $D(c)$.

\end{lemma}

\begin{proof}
Let $ u, v $ be stable vertices. Without loss of generality, we can suppose that $u = (i,j)$, $v = (k,l)$ with $i \leq k$ and $j \leq l$ (otherwise we can rotate $x$ to obtain this property). In this case $u$ and $v$ divide the interior of $D(x)$ into nine regions (see Figure \ref{fig:nineregions}). Let $P_{u,2}, P_{u,3}, P_{u,4}$ be three disjoint paths that connect $u$ with the border through the second, third and fourth quadrants of $u$.  These paths exist according to the proof of Lemma  \ref{lem:3caminos}. Similarly, define $P_{v,1}, P_{v,2}, P_{v,3}$
three disjoint paths that connect $v$ to the border through the first, second and third quadrants of $v$.

\begin{figure}
  \centering
\usetikzlibrary{snakes}

 \begin{tikzpicture}[scale=0.8]
\draw[<->,thick,gray!50] (2.5,-1.5) -- (2.5,3.5);
\draw[<->,thick,gray!50] (-1,2) -- (4,2);
\draw [->,snake=zigzag](2.5,2) -- (3.5,4);  
\draw [->,snake=zigzag](2.5,2) -- (-1.5,3.5);  
\draw [->,snake=zigzag](2.5,2) -- (4.5,-1);  
\draw[<->,thick,gray!50] (0.5,-1.5) -- (0.5,3.5);
\draw[<->,thick,gray!50] (-1,0) -- (4,0);

\draw [->,snake=zigzag](0.5,0) -- (-1,4);   
\draw [->,snake=zigzag](0.5,0) -- (-1.5,-1);    
\draw [->,snake=zigzag](0.5,0) -- (3.5,-2);     
\node at (0.7815,0.3003) {$u$};
\node at (1.3051,2.8092) {$P_{v,2}$};
\node at (3.5,3.0524) {$P_{v,1}$};
\node at (3.4808,1.1287) {$P_{v,4}$};
\node at (2.2383,1.6429) {$v$};
\node at (-0.35,1.0002) {$P_{u,2}$};
\node at (-0.2134,-0.8) {$P_{u,3}$};
\node at (1.686,-0.3376) {$P_{u,4}$};
\draw[double]  (-1,3.5) rectangle (4,-1.5);
\node at (2.5,2) {$\bullet$};
\node at (0.5,0) {$\bullet$};

\node  at (-2,-2.5) {$0$};
\node  at (-2,3.5) {$0$};
\node  at (-2,1.5) {$0$};
\node  at (-2,2.5) {$0$};
\node  at (-2,-0.5) {$0$};
\node  at (-2,0.5) {$0$};
\node  at (-2,-1.5) {$0$};

\node  at (5,3.5) {$0$};
\node  at (5,1.5) {$0$};
\node  at (5,2.5) {$0$};
\node  at (5,-0.5) {$0$};
\node  at (5,0.5) {$0$};
\node  at (5,-1.5) {$0$};

\node  at (-1,4.5) {$0$};
\node  at (-2,4.5) {$0$};
\node  at (1,4.5) {$0$};
\node  at (0,4.5) {$0$};
\node  at (3,4.5) {$0$};
\node  at (2,4.5) {$0$};
\node  at (5,4.5) {$0$};
\node  at (4,4.5) {$0$};

\node  at (-1,-2.5) {$0$};
\node  at (1,-2.5) {$0$};
\node  at (0,-2.5) {$0$};
\node  at (3,-2.5) {$0$};
\node  at (2,-2.5) {$0$};
\node  at (5,-2.5) {$0$};
\node  at (4,-2.5) {$0$};
\node at (-0.6134,3.1901) {$\bullet$};
\draw[dashed]  plot[smooth, tension=.7] coordinates {(3.5,-2) (4.5,-2) (4.5,-1)};
\draw[dashed]  plot[smooth, tension=.7] coordinates {(-1.5,-1) (-2,3)  (-0.5,4.5) (3.5,4)};
\end{tikzpicture}

    \caption{Vertices $ u $ and $ v $ divides the interior of $D(x)$ into four regions each one.
Together they split the space into nine regions. According to Lemma \ref{lem:3caminos}, we can choose three disjoint paths connecting $u$ and $v$, in such a way that each of the nine regions intersect at most one path. We use the border of $D(x)$ to connect the paths that do not intersect in the interior of $D(x)$.}
\label{fig:nineregions}
    \end{figure}
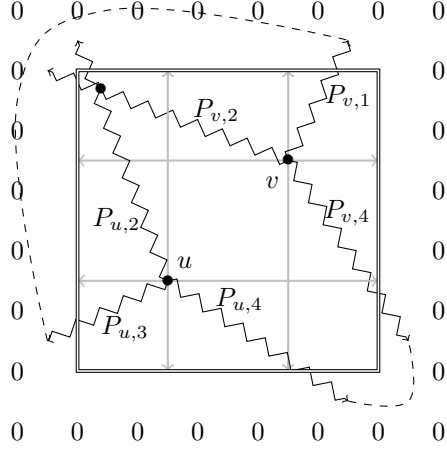

Observe fist that $P_{u,3}$  touch regions that are disjoint from the ones touched by $P_{v,1}, P_{v,2}$ and $P_{v,3}$. The same is true for $P_{v,1}$ with respect to $P_{u,2}, P_{u,3}, P_{u,4}$.
The first observation implies that paths $P_{u,3}$ and $P_{v,1}$ reach the border without intersecting any other path.  Let $w_1$ and $w_2$ be respectively the intersections of $P_{u,3}$ and $P_{v,1}$ with the border. Let now $P_{w_1, w_2}$ be any path in $G_B$ connecting $w_1$ and $w_2$. We call $P_{1,3}$ the path induced by $P_{u,3} \cup P_{w_1, w_2} \cup P_{v,1}$. 
 
Observe now that $P_{u,2}$ and $P_{v,4}$ must be disjoint, as well as $P_{u,4}$ and $P_{v,2}$. This observation implies that $P_{u,2}$ either intersects $P_{v,2}$ or it do not intersect any other path, and the same is true for  $P_{u,4}$ and $P_{v,4}$. If $P_{u,2}$ does not intersect $P_{v,2}$, then we define a path $P_{2,2}$ in a similar way than $P_{1,3}$, i.e., we connect the endpoints of $P_{u,2}$  and $P_{v,2}$ through a path in the border (we can choose this path disjoint from $P_{1,3}$ since the border is tri-connected).  Suppose now that  $P_{u,2}$ intersects $P_{v,2}$. Let $w$ the first site where $P_{u,2}$ and $P_{v,2}$ intersect, let $P_{u,w}$ be the $u,w$-path contained in $P_{u,2}$, and let $P_{w,v}$ be the $w,v$-path contained in $P_{v,2}$. We call in this case $P_{2,2}$ the path $P_{u,w} \cup P_{w,v}$. Note that also in this case $P_{2,2}$ is disjoint from $P_{1,3}$.  Finally, we define $P_{4,4}$ in a similar way using paths $P_{u,4}$ and $P_{v,4}$. We conclude that $P_{1,3}$, $P_{2,2,}$ and $P_{4,4}$ are three disjoint paths of stable sites connecting $u$ and $v$ in $G[Z]$. 

\end{proof}

We are now ready to show our characterization of stable sets of vertices.

\begin{lemma}\label{col: TRI}

Let $x \in \{0,1\}^{[n]\times[n]}$ be a finite configuration, and let $u$ be a site in $[n] \times [n]$. Then, $u$ is stable for $c = c(x)$ if and only if $u$ is contained in a tri-connected component of $G[Z]$. 


\end{lemma}

\begin{proof}

From Lemma \ref{lem:candDc}, we know that $u$ is stable for $c$  if and only if it is stable for $D(c)$. Let $S$ be the set of sites stable for $D(c)$. We claim that $S$ is a tri-connected component of $G[Z]$. From Lemma \ref{teo: TRI}, we know that for every pair of sites in $S$ there exist three disjoint paths in $G[S]$ connecting them, so the set $S$ must be contained in some tri-connected component $T$ of $G[Z]$. Since $G[T]$ is a graph of degree at least three, and the sites in $T$ are contained in $Z$, then Lemma \ref{pro: estabilidad general} implies that $T$ must form a stable set of vertices, then $T$ equals $S$. 

On the other direction,  Lemma \ref{pro: estabilidad general} implies that any tri-connected component of $G[Z]$ must form a stable set of vertices for $D(c)$, so $u$ is stable for $c$. \end{proof}

We are now ready to study the complexity of \stability for this rule.

\begin{proof}[Proof of Theorem \ref{thm: 234 square NC} ]
Let $(x, u)$ be an input of {\stability}, i.e. $x$ is a finite configuration of dimensions $n\times n$, and $u$ is a site in $[n]\times [n]$. Our algorithm for {\stability} first computes from $x$ the finite configuration $D(x)$. Then, the algorithm uses the algorithm of Proposition \ref{prop:trijaja} to compute the tri-connected components of $G[Z]$, where $Z$ is the set of sites $w$ such that $D(x)_w = 0$. Finally, the algorithm answers \emph{no} if $u$  belongs to some tri-connected component of $G[Z]$, and answer \emph{yes} otherwise.

\begin{algorithm}[h]
\caption{Solving \stability 234}\label{alg: PRED 234}
\begin{algorithmic}[1]
\REQUIRE $x$ a finite configuration of dimensions $n \times n$ and $u \in [n]\times [n]$ such that $x_u =0$.   
\STATE Compute the finite configuration $D(x)$ of dimensions $m\times m$ with $m= 2n^2+3n$
\STATE Compute the set $Z = \{ w \in [m]\times [m] ~:~ D(x)_w = 0\}$.
\STATE Compute the graph $G[Z]$.
\STATE Compute the set $\mathcal{T}$ of tri-connected components of $G[Z]$.
\FORALLP{ $T \in \mathcal{T}$}
  \IF{$u \in T$}
    \RETURN \emph{Accept}
  \ENDIF  
\ENDFORALLP
\RETURN  \emph{Reject}
\end{algorithmic}
\end{algorithm}

The correctness of Algorithm \ref{alg: PRED 234} is given by Lemma \ref{col: TRI}. Indeed, the algorithm answers \emph{Reject} on input $(x, u)$ only when $u$ does not belong to a tri-connected component of $G[Z]$. From Lemma \ref{col: TRI}, it means that $u$ is not stable, so there exists $t>0$ such that $F^t(c(x))_u = 1$. 

Let $N=n^{2}$ the size of the input.
Step {\bf 1} can be done in $\cO(\log N)$ time with $m^2  = \cO(N^2)$ processors: one processor for each site of $B(x)$ computes from $x$ the value of the corresponding site in $B(x)$. 
Step {\bf 2} can be done in time in $\cO(\log N)$ with $\cO(N^2)$ processors, representing $Z$ as a vector in $\{0,1\}^{m^2}$, each coordinate is computed by a processor. 
Step {\bf 3} can be done in time $\cO(\log N)$ and $\cO(N^2)$ processors: we give one processor to each site in $Z$, which fill the corresponding four coordinates of the adjacency matrix of $G[Z]$. Step {\bf 4} can be done in time $\mathcal{O}(\log^2N)$ with $\mathcal{O}((N^2)^4)$ processors using the algorithm of Proposition \ref{prop:trijaja}. Finally, steps {\bf 5} to {\bf 10} can be done in time $\cO(\log N)$ with $\cO(N^2)$ processors: the algorithm checks in parallel if $u$ is contained in each tri-connected components.
All together the algorithm runs in time $\mathcal{O}(\log^2N)$ with $\cO(N^8)$ processors. 
\end{proof}

\subsection{Algebraic Rules}
We will now study the family of FTCA where the cells become active with one or two neighbors. We consider there the rules $12$, $123$, $124$. Of course, rule $1234$ will fit in our analysis, but we already know that this rule is trivial. As we already mentioned, these rules are \emph{algebraic} in the sense that, in order to answer the \stability problem, we will \emph{accelerate} the dynamics using algebraic properties of these rules.

In the following, we assume that the cells are placed in the Cartesian coordinate system, where each cell is placed in a coordinate in $\mathbb{N}\times \mathbb{N}$. Moreover, without loss of generality,  our decision cell is $u=(0,0)$ and the configuration $c$ has at least one active cell. Let $\tau>1$ be the distance from $u$ to the first active cell. Remember that we called $D_{\tau}$ the set of cells at distance $\tau$ from $u$. We also call $d_I(\tau)$ the \emph{diagonal at distance $\tau$ of $u$ in the first quadrant}, defined as follows: 
$$d_I(\tau):=\{(i,j)\in \mathbb{N}^{2}: |i-j|=\tau \textrm{ and } i,j>0\}.$$
Then, we place ourselves in the case where all the cells in $D_{\tau-1}$ are inactive.

Let $c'$ be the configuration obtained after one step, i.e. $c' = F(c)$, where $F$ is one of the rules in $\{12, 123, 124\}$. Notice that all the cells in $D_{\tau-1}$ will remain inactive in $c'$. Moreover, the states of cells in $d_I(\tau-1)$ can be computed as follows (see Figure \ref{fig: or tec}):
$$\forall(i,j)\in d_I(\tau-1), ~c'_{i,j} = c_{i+1,j}\vee c_{i,j+1}.$$ Where $\vee$ is the OR operator (i.e. $c_{i,j}' = 1$ if $ c_{i+1,j}=1$ or  $c_{i,j+1}=1$).
If we inductively apply this formula, we deduce:
$$F^{\tau-2}(c)_{1,1} = \bigvee_{(i,j)\in d_I(\tau)}c_{i,j}.$$
Note that if the cell $(1,1)$ is inactive at time $\tau-1$, then necessarily all the cells in $d_I(\tau)$ are inactive at time $0$. Moreover, by we can apply the same ideas to every cell $(i,j) \in \bigcup_{1 \leq k \leq \tau} d_I(k)$ such that $i,j \geq 1$, obtaining:
  \begin{equation}\label{eqn:OR12}
F^{\tau-1-i}(c)_{i,1} = \bigvee_{\scalebox{0.75}{$\begin{array}{c}
      (k,j)\in D(\tau)\\
      k\geq i \\
    \end{array}$}}c_{k,j} \quad \textrm{and} \quad F^{\tau-1-j}(c)_{1,j} = \bigvee_{\scalebox{0.75}{$\begin{array}{c}
      (i,k)\in D(\tau)\\
      k \geq j \\
    \end{array}$}}c_{i,k}. 
    \end{equation}
Analogously, we can define $d_{II}(\tau)$ (resp. $d_{III}(\tau)$, $d_{IV}(\tau)$) the \emph{diagonals at distance $\tau$ of $u$ in the second (resp. third, fourth) quadrant}, and deduce similar formulas in the other three quadrants.  Concretely we can compute the states the states of cells $(\pm i,\pm 1)$, $i =1,...,\tau-1$ in time $\tau - i$ and the states of cells $(\pm 1,\pm j)$, $j =1,...,\tau-1$ in time $\tau - j$. These cells are represented as the hatched patterns in Figure \ref{fig: diamond}.
This way of computing cells we call it \emph{OR technique}.

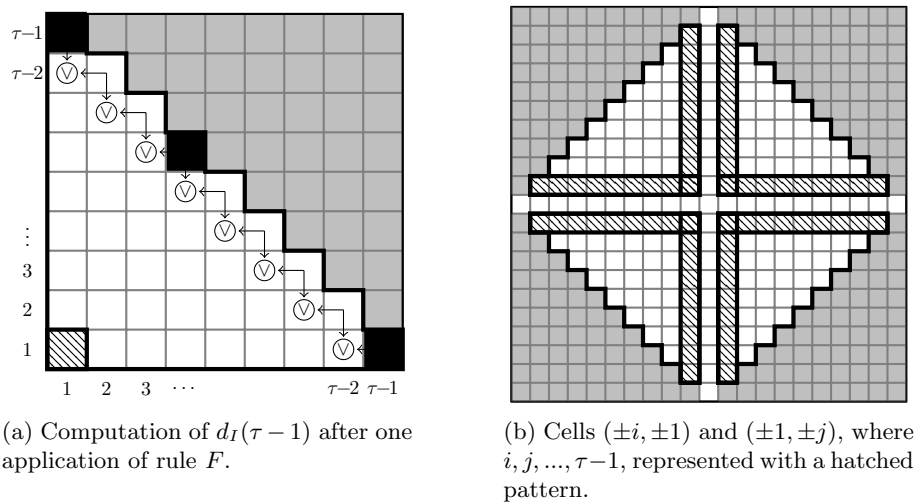
\begin{figure}
  \centering
    \begin{subfigure}[t]{0.45\textwidth}
    \centering
      \begin{tikzpicture}[scale=0.525, every node/.style={scale=0.75}]
       
\draw [](1.5,8.5)circle [radius=0.26] node (v4) {$\vee$};
\draw [](2.5,7.5)circle [radius=0.26] node (v5) {$\vee$};
\draw [](3.5,6.5)circle [radius=0.26] node (v6) {$\vee$};
\draw [](4.5,5.5)circle [radius=0.26] node (v7) {$\vee$};
\draw [](5.5,4.5)circle [radius=0.26] node (v8) {$\vee$};
\draw [](6.5,3.5)circle [radius=0.26] node (v9) {$\vee$};
\draw [](7.5,2.5)circle [radius=0.26] node (v10) {$\vee$};
\draw [](8.5,1.5)circle [radius=0.26] node (v11) {$\vee$};

\fill [color=gray,ultra thick,opacity=0.5](2,10) -- (2,9) node (v13) {} -- (3,9) node (v14) {} -- (3,8) -- (4,8) -- (4,7) --  (5,7) node (v1) {} -- (5,6) -- (6,6) -- (6,5) -- (7,5) -- (7,4) -- (8,4) -- (8,3) -- (9,3) -- (9,2) -- (10,2) node (v2) {} -- (10,1) -- (10,10) node (v3) {} -- (1,10)--(2,10) node (v12) {};

\fill[black] (1,9) rectangle (2,10); 
\fill[black] (4,6) rectangle (5,7); 
\fill[black] (9,1) rectangle (10,2); 

\draw[color=gray, thick] (1,1) grid (10,10);
\draw[color=black, thick] (1,1) rectangle (10,10);

\draw [black,ultra thick](2,10) -- (2,9) node (v13) {} -- (3,9) node (v14) {} -- (3,8) -- (4,8) -- (4,7) --  (5,7) node (v1) {} -- (5,6) -- (6,6) -- (6,5) -- (7,5) -- (7,4) -- (8,4) -- (8,3) -- (9,3) -- (9,2) -- (10,2) node (v2) {} -- (10,1) -- (1,1) node (v3) {} -- (1,10)--(2,10) node (v12) {};
\draw [black,ultra thick,pattern=north west lines, pattern color=black]  (2,2) rectangle (v3);

\draw [->] (1.5,9.5) -- (v4) {};
\draw [->] (2.5,8.5) -- (v4);

\draw [->] (2.5,8.5) -- (v5) {};
\draw [->] (3.5,7.5) -- (v5);

\draw [->] (3.5,7.5) -- (v6) {};
\draw [->] (4.5,6.5) -- (v6);

\draw [->] (4.5,6.5) -- (v7) {};
\draw [->] (5.5,5.5) -- (v7);

\draw [->] (5.5,5.5) -- (v8) {};
\draw [->] (6.5,4.5) -- (v8);

\draw [->] (6.5,4.5) -- (v9) {};
\draw [->] (7.5,3.5) -- (v9);

\draw [->] (7.5,3.5) -- (v10) {};
\draw [->] (8.5,2.5) -- (v10);

\draw [->] (8.5,2.5) -- (v11) {};
\draw [->] (9.5,1.5) -- (v11);

\draw (v12) -- (v13) -- (v14);

\node at (0.5,9.5) {$\tau\!\!-\!\!1$};
\node at (0.5,8.5) {$\tau\!\!-\!\!2$};
\node at (0.5,4.5) {$\vdots$};
\node at (0.5,3.5) {$3$};
\node at (0.5,2.5) {$2$};
\node at (0.5,1.5) {$1$};

\node at (9.5,0.5) {$\tau\!\!-\!\!1$};
\node at (8.5,0.5) {$\tau\!\!-\!\!2$};
\node at (4.5,0.5) {$\cdots$};
\node at (3.5,0.5) {$3$};
\node at (2.5,0.5) {$2$};
\node at (1.5,0.5) {$1$};
      \end{tikzpicture}
    \caption{Computation of $d_I(\tau-1)$ after one application of rule $F$.} \label{fig: or tec} 
  \end{subfigure}
  \hfill
  \begin{subfigure}[t]{0.45\textwidth}
    \centering
      \begin{tikzpicture}[scale=0.25]
      \draw[color=gray, thick] (-11,-11) grid (10,10);
      \draw[color=black, thick] (-11,-11) rectangle (10,10);

      \fill [color=gray,ultra thick,opacity=0.5](1,9) -- (1,8)-- (2,8) -- (2,7) -- (3,7) -- (3,6) --  (4,6) -- (4,5) -- (5,5) -- (5,4) -- (6,4) -- (6,3) -- (7,3) -- (7,2) -- (8,2) -- (8,1) -- (9,1) -- (9,0) -- (10,0) -- (10,10)  -- (0,10)-- (0,9)--(1,9);
      \fill [color=gray,ultra thick,opacity=0.5](-2,9) -- (-2,8)-- (-3,8) -- (-3,7) -- (-4,7) -- (-4,6) --  (-5,6) -- (-5,5) -- (-6,5) -- (-6,4) -- (-7,4) -- (-7,3) -- (-8,3) -- (-8,2) -- (-9,2) -- (-9,1) -- (-10,1) -- (-10,0) -- (-11,0) -- (-11,10)  -- (-1,10)-- (-1,9)--(-2,9);
      \fill [color=gray,ultra thick,opacity=0.5](1,-10) -- (1,-9)-- (2,-9) -- (2,-8) -- (3,-8) -- (3,-7) --  (4,-7) -- (4,-6) -- (5,-6) -- (5,-5) -- (6,-5) -- (6,-4) -- (7,-4) -- (7,-3) -- (8,-3) -- (8,-2) -- (9,-2) -- (9,-1) -- (10,-1) -- (10,-11)  -- (0,-11)  -- (0,-10)--(1,-10);
      \fill [color=gray,ultra thick,opacity=0.5](-2,-10) -- (-2,-9)-- (-3,-9) -- (-3,-8) -- (-4,-8) -- (-4,-7) --  (-5,-7) -- (-5,-6) -- (-6,-6) -- (-6,-5) -- (-7,-5) -- (-7,-4) -- (-8,-4) -- (-8,-3) -- (-9,-3) -- (-9,-2) -- (-10,-2) -- (-10,-1) --(-11,-1) -- (-11,-11)  -- (-1,-11)-- (-1,-10)--(-2,-10);

      \draw [black,ultra thick](1,9) -- (1,8)-- (2,8) -- (2,7) -- (3,7) -- (3,6) --  (4,6) -- (4,5) -- (5,5) -- (5,4) -- (6,4) -- (6,3) -- (7,3) -- (7,2) -- (8,2) -- (8,1) -- (9,1) -- (9,0) -- (0,0)  -- (0,9) node (v1) {}--(1,9);z
      \draw [black,ultra thick](-2,9) -- (-2,8)-- (-3,8) -- (-3,7) -- (-4,7) -- (-4,6) --  (-5,6) -- (-5,5) -- (-6,5) -- (-6,4) -- (-7,4) -- (-7,3) -- (-8,3) -- (-8,2) -- (-9,2) -- (-9,1) -- (-10,1) -- (-10,0) -- (-1,0)  -- (-1,9)--(-2,9);
      \draw [black,ultra thick](1,-10) -- (1,-9)-- (2,-9) -- (2,-8) -- (3,-8) -- (3,-7) --  (4,-7) -- (4,-6) -- (5,-6) -- (5,-5) -- (6,-5) -- (6,-4) -- (7,-4) -- (7,-3) -- (8,-3) -- (8,-2) -- (9,-2) -- (9,-1) node (v2) {} -- (0,-1)  -- (0,-10) node (v3) {}--(1,-10);
      \draw [black,ultra thick](-2,-10) -- (-2,-9)-- (-3,-9) -- (-3,-8) -- (-4,-8) -- (-4,-7) --  (-5,-7) -- (-5,-6) -- (-6,-6) -- (-6,-5) -- (-7,-5) -- (-7,-4) -- (-8,-4) -- (-8,-3) -- (-9,-3) -- (-9,-2) -- (-10,-2) node (v5) {} -- (-10,-1) -- (-1,-1)  -- (-1,-10)--(-2,-10) node (v4) {};

      \draw [black,ultra thick,pattern=north west lines, pattern color=black] (-2,9) rectangle  (-1,0);
      \draw [black,ultra thick,pattern=north west lines, pattern color=black] (-10,1) rectangle  (-1,0);
      \draw [black,ultra thick,pattern=north west lines, pattern color=black] (v1) rectangle (1,0);
      \draw [black,ultra thick,pattern=north west lines, pattern color=black] (0,1) rectangle (9,0);
      \draw [black,ultra thick,pattern=north west lines, pattern color=black] (0,-2) rectangle (v2);
      \draw [black,ultra thick,pattern=north west lines, pattern color=black] (v3) rectangle (1,-1);
      \draw [black,ultra thick,pattern=north west lines, pattern color=black] (v4) rectangle (-1,-1) node (v6) {};
      \draw [black,ultra thick,pattern=north west lines, pattern color=black] (v5) rectangle (v6);
      \end{tikzpicture}   
    \caption{Cells $(\pm i,\pm 1)$ and $(\pm 1, \pm j)$,  where $i,j,...,\tau-1$, represented with a hatched pattern. }
    \label{fig: diamond}
  \end{subfigure}
  \caption{Computation in a disc of radius $\tau$, where only the cells at distance $\tau$ from $(0,0)$ can be initially active.}
\end{figure}

We define the following sets of cells.
\begin{itemize}
  \item The {\bf north-east triangle} is set of cells in the first quadrant between the cells in the hatched pattern (including them) and the gray zone, 
  i.e. is the set $D_{\tau, I} = \{(i,j)\in \mathbb{N}^2: |i-j| \leq \tau \textrm{ and } i,j \geq 1 \}$. 
  Analogously we define north-west, south-west and south-east triangles, and denote them $D_{\tau, II}, D_{\tau, III}$ and $D_{\tau, IV}$, respectively.
  \item The {\bf north corridor} is the set of cells in the positive x-axis contained in the disc , i.e. is the set $\{(0,i)\in \mathbb{N}^2: 1 \leq i\leq \tau \}$.
  Analogously we define west, south and east corridors.
\end{itemize}

Consider now $B_2$, the ball of radius $2$ centered in $u = (0,0)$, and name the vertices of the ball as depicted in Figure  \ref{fig: von Neumann r2}. We can compute the states of cells $b,d,f$ and $h$ in time $\tau-1$ using the OR technique. In order to solve \stability, the use of this information will depend on the which rule that we are considering. In the following, we will show how to use this information to solve stability for rule $123$, then for rule $12$ and finally for rule $124$.

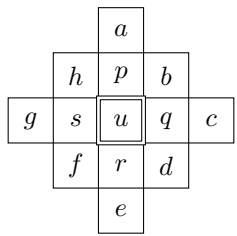
\begin{figure}
    \centering
    \begin{tikzpicture}[scale=0.3]

    \draw  (-8,11) rectangle (2,9);
    \draw  (-4,15) rectangle (-2,5);
    \draw  (-6,13) rectangle (0,7);

    \node at (-3,14) {$a$};
    \node at (-1,12) {$b$};
    \node at (1,10)  {$c$};
    \node at (-1,8)  {$d$};
    \node at (-3,6)  {$e$};
    \node at (-5,8)  {$f$};
    \node at (-7,10) {$g$};
    \node at (-5,12) {$h$};
    \node at (-3,12) {$p$};
    \node at (-1,10) {$q$};
    \node at (-3,8)  {$r$};
    \node at (-5,10) {$s$};
    \node at (-3,10) {$u$};

    \draw [double distance=1pt] (-4,11) rectangle (-2,9); 

    \end{tikzpicture}
    \caption{Names for the cells in the ball of $B_2$. Cells  $p,q,r$ and $s$ correspond to the neighbors of cell $u$.}
    \label{fig: von Neumann r2} 
 \end{figure}

\subsubsection{Solving \stability for rule $123$}

Rule $123$ is the simplest of the Algebraic rules. Its simplicity, follows mainly from the following claim. Remember that $\tau$ is the distance from $u$ to the nearest active cell.\\

{\bf Claim 1:} Either $u$ becomes active at time $\tau$ or $u$ is stable. \\

We know that at least one of the cells in $\{b, d,  f, h\}$ will be active at time $\tau -2$. Indeed the states of those cells depend only on the logical disjunction of the sites in the border of the disc $D_{\tau}$, and we are assuming that there is at least one active site in $D_{\tau} \setminus D_{\tau-1}$. Therefore, in time $\tau -1$, necessarily at least one neighbor $w \in N(u)$ will become active, since it will have more than one active neighbor, and less than four (because $u$ is inactive at time $\tau-1$). Suppose now that $u$ does not become active in time $\tau$. Since $u$ has one active neighbor in time $\tau-1$, the only possibility is that the four neighbors of $u$ are active $\tau-1$. Since the rule is freezing,  $u$ will remain stable in inactive state. 

At this point, we know how to compute the states of $b,d,f$ and $h$ in time $\tau -2$, and we know that the only possibility for $u$ to become active is on time $\tau$. Therefore, in order to decide {\stability} for rule $123$ we need to compute the states of cells $p, q, r$ and $s$ in time $\tau-1$. 
In the following, we show how to compute the state of site $p$ in time $\tau-1$. The arguments for computing cells $q,r$ and $s$ will be deduced by analogy (considering the same arguments in another quadrant). 

Call $x_{(i,j)}^t$ the state of cell $(i,j)$ in time $t$, with the convention of $x^0_{(i,j)}$ is the input state of $(i,j)$. 
First, note that, for all $i\in \{0, \dots, \tau-2\}$, the state of $(0, i)$ in time $1$ will be inactive. Moreover, the state of cell $(0, \tau-1)$ will be active if and only if at least one of its three neighbors $(-1, \tau-1)$, $(1, \tau-1)$ or $(0, \tau)$ is active at time $0$.  Then, we deduce the following formula for $x_{(0, \tau-1 )}^1$:
$$x^{1}_{(0,\tau-1)}=x^{0}_{(-1,\tau-1)} \vee x^{0}_{(0,\tau)} \vee x^{0}_{(1,\tau-1)}$$
For the same reasons, we notice that at time $j>0$ the nearest neighbor from $u$ is in the border of $D_{\tau-j}$. Therefore,
$$x^{j}_{(0,\tau-j)}=x^{j-1}_{(-1,\tau-j)} \vee x^{j-1}_{(0,\tau)} \vee x^{j-1}_{(1,\tau-j)}$$
In particular
$$x^{\tau-1}_p = x^{\tau-1}_{(0,1)}=x^{\tau-2}_{(-1,1)} \vee x^{\tau-2}_{(0,1)} \vee x^{\tau-2}_{(1,1)}.$$
Remember that we know how to compute $x^{\tau-2}_{(-1,1)}$ and  $x^{j-1}_{(1,1)}$ according to Equation \ref{eqn:OR12}. We deduce that we can compute ,$x_p^{\tau-1}$ as follows:

\begin{equation}\label{eqn:calculop12} x^{\tau-1}_{p}=\bigvee_{k=1}^{\tau} x^{0}_{(-k,\tau-k)}\vee x^{0}_{(0,\tau)} \vee \bigvee_{k=1}^{\tau} x^{0}_{(k,\tau-k)}
\end{equation}
In words, the state of $p$ at time $\tau-1$ can be computed as the OR of all the cells to the north of the $u$ contained in $D_{\tau} \setminus D_{\tau-1}$.
Analogously we can compute $x^{\tau-1}_{q}$, $x^{\tau-1}_{r}$ and $x^{\tau-1}_{s}$.


%

\subsubsection{Solving stability for rule $12$}

For rule $12$ the computation of $a$, $c$, $e$ and $g$ is not so simple as in the previous case. First of all, there is one case when cell $u$ remains active, though we can also assume {\bf Claim 1} for this rule. 

Indeed, remember that we know that at least one of cells in $\{b, d, f, h\}$ will be active at time $\tau -2$. Suppose that $u$ remains inactive at time $\tau$. There are two three possibilities: (1) none of the neighbors of $u$ will become active at time $\tau-1$, (2) three neighbors of $u$ become active at time $\tau-1$; and (3) the four neighbors of $u$ become active at time $\tau-1$. Note that in (2) and (3) we directly obtain that $u$ is stable, because the rule is freezing.

The case when the sum in its neighborhood is $0$ is slightly more complicated. As we said, we know that at least one of $\{b, d, f, h\}$ becomes active at time $\tau -2$. Suppose, without loss of generality, that $b$ satisfy this condition. On the other hand, we are assuming that $p$ and $q$ remain inactive at time $\tau-1$. Since this cells have one active neighbor at time $\tau-2$, the sole possibility is that cells $h, a, c$ and $d$ are active at time $\tau-1$. Applying the same arguments to cells $r$ and $s$, we deduce that all cells $a, b, c, d, e, f, g$ and $h$ will be active at time $\tau-2$. Since the rule is freezing, we deduce that cells $p, q, r$ and $s$ are stable, obtaining that also $u$ is stable. 

From Equation \ref{eqn:OR12}, we know how to compute the states of cells $b, d, f$ and $h$ in time $\tau-2$. To decide the stability of $u$, we need to compute the states of $p,q,r$ and $s$ in time $\tau-1$. In this case, however, the dynamics in the corridors is more complicated. In the following, we will show how to compute the east corridor (in order to compute $q$), depicted in Figure \ref{fig: comp1}. We will study only this case, since the other three corridors are analogous.

\begin{figure}
  \centering
  \includegraphics[width=.985\textwidth]{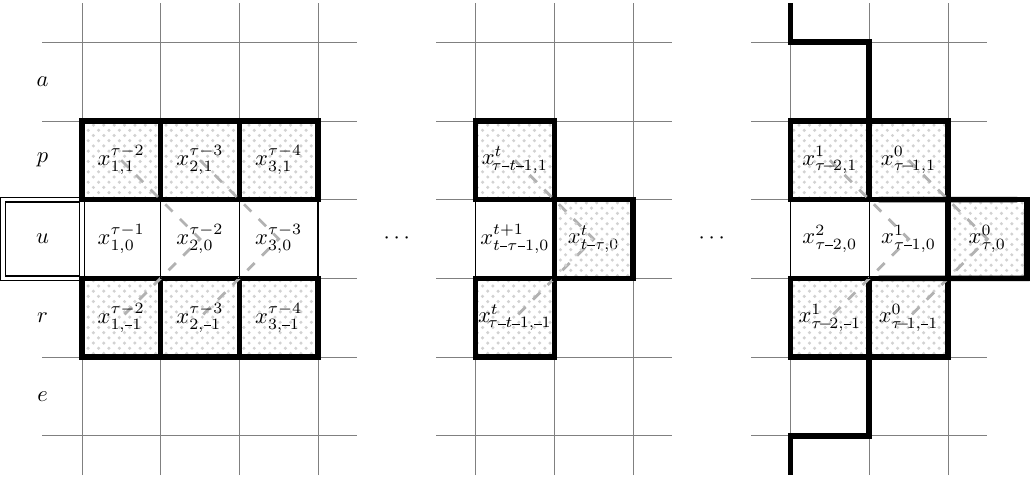}
  \caption[Computation of cells $(\pm t,\pm 1)$ and $(\pm 1,\pm t)$, $t\leq \tau$.]
  {Computation of east corridor. 
  Recall that $x_{(i,j)}^{t}$ is the state of cell $(i,j)$ in time-step $t$.  
  The gray cells were previously computed using the OR technique. 
  Dashed lines connect cells which potentially change states at the same time.}
  \label{fig: comp1}
\end{figure}

Remember that, using Equation \ref{eqn:OR12} we can compute the values of $x^{\tau -1 - i}_{(i,1)}$ and $x^{\tau -1 - i}_{(-i,1)}$, for every $i \in \{1, \dots, \tau-2\}$.
Notice first that, if $x^{\tau -1 - i}_{(i,1)} \neq x^{\tau -1 - i}_{(-i,1)}$, then necessarily $x^{\tau -i}_{(i,0)} = 1$. Indeed, we know that $x^{\tau -i-1}_{(i-1,0)} = 0$ (otherwise we contradict the definition of $\tau$). Then, $x^{\tau -1 - i}_{(i,1)} \neq x^{\tau -1 - i}_{(-i,1)}$ implies that in time $\tau-1-i$ the cell $(i,0)$ will have more than one and less than three active neighbors, so it will become active in time $\tau -i$. 

Let $i^*$ be the minimum value of $i \in \{1, \dots, \tau\}$ such that $x^{\tau -1 - i}_{(i,1)} \neq x^{\tau -1 - i}_{(-i,1)}$. If no such $i$ exists, then fix $i^*= \tau$. Call $I^*$ the set $\{1, \dots, i^*-1\}$. In other words,  we know that $x^{\tau -1 - i}_{(i,1)} =  x^{\tau -1 - i}_{(-i,1)}$ for every $i \in I^*$. Moreover, we also know that $x^{\tau - i^*}_{(i^*,0)} = 1$.  

We now identify two situations, concerning the values of   $x^{\tau -1 - i}_{(i,1)}$, for $i \in I^*$.

\begin{description}
  \item[If $x^{\tau-1-i}_{(i,1)}=x^{\tau-1-i}_{(-i,1)}=0$ then,]  the value of $x^{\tau-i}_{(i,0)}$ will equal the value of $x^{\tau-1-i}_{(i+1,0)}$. Indeed, in time $\tau-1-i$, the cell $(i,0)$ will have three inactive neighbors ( $(-i,1), (i,1), (i-1,0)$). Then it will take the same state than cell $(i+1,0)$ at time $\tau-i-1$. 

  \item[If  $x^{\tau-1-i}_{(i,1)}=x^{\tau-1-i}_{(-i,1)}=1$ then,]  the value of $x^{\tau-i}_{(i,0)}$ will be the opposite than value of $x^{\tau-1-i}_{(i+1,0)}$. Indeed, in time $\tau-1-i$, the cell $(i,0)$ will have two active neighbors ( $(-i,1), (i,1)$) and one inactive neighbor ($(i-1,0)$). Then cell $(i,0)$ is active at time $\tau-i$ if and only if cell $(i+1,0)$  is inactive at time $\tau-1-i$.

\end{description}
We imagine that a signal drive along the corridor. The signal starts at $(i^*,0)$ with value $x^{\tau-i^*}_{(i^*,1)}$. The movement of the signal satisfies that, each time it encounters an $i\in I^*$ such that $x^{\tau-1-i}_{(i,1)}=x^{\tau-1-i}_{(-i,1)}=1$, the state switches to the opposite value.
Let $z = |\{i \in I^* : x^{\tau-1-i}_{(i,1)}=x^{\tau-1-i}_{(-i,1)}=1 \}|$ (i.e. $z$ is the number of switches). From the two situations explained above, we deduce the following lemma.

\begin{lemma}\label{lem:corridor12}
 $x^{\tau-1}_{(1,0)}$ equals $x^{\tau-i^*}_{(i^*,1)}$ if $z$ is even, and $x^{\tau-1}_{(1,0)}$ is different than $x^{\tau-i^*}_{(i^*,1)}$ when $z$ is odd.
\end{lemma}

Therefore, to solve \stability for rule $12$, we compute the values of $x^{\tau-1}_p$, $x^{\tau-1}_q$, $x^{\tau-1}_r$, $x^{\tau-1}_s$  according to Lemma \ref{lem:corridor12}. 

\subsubsection{Solving \stability for rule $124$}

The analysis for the rule $124$ is more complicated than the one we did for rule $12$ and $123$. In fact, one great difference is that {\bf Claim 1} is no longer true for this rule. In other words, $u$ might not be stable but change after time-step $\tau$.

For this rule, the cases when the cell $u$ remain inactive are the cases when $u$ has zero or three active neighbors. The possible cases when the cell $u$ remain inactive at time $\tau$ are given in Figure \ref{fig: casos 124}.

  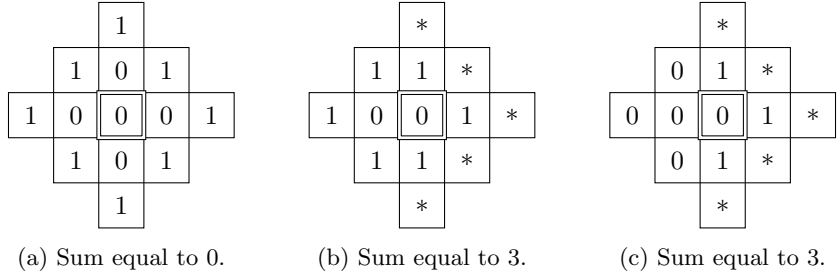
\begin{figure} 
    \centering
      \begin{subfigure}[t]{0.32\textwidth}
              \centering
                \begin{tikzpicture}[scale=0.3]
                \draw  (-8,11) rectangle (2,9);
                \draw  (-4,15) rectangle (-2,5);
                \draw  (-6,13) rectangle (0,7);
                \node at (-3,14) {$1$};
                \node at (-1,12) {$1$};
                \node at (1,10)  {$1$};
                \node at (-1,8)  {$1$};
                \node at (-3,6)  {$1$};
                \node at (-5,8)  {$1$};
                \node at (-7,10) {$1$};
                \node at (-5,12) {$1$};
                \node at (-3,12) {$0$};
                \node at (-1,10) {$0$};
                \node at (-3,8)  {$0$};
                \node at (-5,10) {$0$};
                \node at (-3,10) {$0$};
                \draw [double distance=1pt] (-4,11) rectangle (-2,9); 
                \end{tikzpicture}
              \caption{Sum equal to 0.}
              \label{fig: casos 124 0}
    \end{subfigure}                    
    \begin{subfigure}[t]{0.32\textwidth}
              \centering
                \begin{tikzpicture}[scale=0.3]
                \draw  (-8,11) rectangle (2,9);
                \draw  (-4,15) rectangle (-2,5);
                \draw  (-6,13) rectangle (0,7);
                \node at (-3,14) {$*$};
                \node at (-1,12) {$*$};
                \node at (1,10)  {$*$};
                \node at (-1,8)  {$*$};
                \node at (-3,6)  {$*$};
                \node at (-5,8)  {$1$};
                \node at (-7,10) {$1$};
                \node at (-5,12) {$1$};
                \node at (-3,12) {$1$};
                \node at (-1,10) {$1$};
                \node at (-3,8)  {$1$};
                \node at (-5,10) {$0$};
                \node at (-3,10) {$0$};
                \draw [double distance=1pt] (-4,11) rectangle (-2,9); 
                \end{tikzpicture}
              \caption{Sum equal to 3.}
              \label{fig: casos 124 33}
    \end{subfigure}
    \begin{subfigure}[t]{0.32\textwidth}
              \centering
                \begin{tikzpicture}[scale=0.3]
                \draw  (-8,11) rectangle (2,9);
                \draw  (-4,15) rectangle (-2,5);
                \draw  (-6,13) rectangle (0,7);
                \node at (-3,14) {$*$};
                \node at (-1,12) {$*$};
                \node at (1,10)  {$*$};
                \node at (-1,8)  {$*$};
                \node at (-3,6)  {$*$};
                \node at (-5,8)  {$0$};
                \node at (-7,10) {$0$};
                \node at (-5,12) {$0$};
                \node at (-3,12) {$1$};
                \node at (-1,10) {$1$};
                \node at (-3,8)  {$1$};
                \node at (-5,10) {$0$};
                \node at (-3,10) {$0$};
                \draw [double distance=1pt] (-4,11) rectangle (-2,9); 
                \end{tikzpicture}
              \caption{Sum equal to 3.}
              \label{fig: casos 124 30}
    \end{subfigure} 
  \caption{Possibles cases of rule $124$ at time $\tau$ such that $u$ remain inactive.}
  \label{fig: casos 124}
\end{figure}

The case when $u$ has four inactive neighbors at time $\tau-1$ is exactly the same that we explained for rule $12$ (see Figure \ref{fig: casos 124 0}). Suppose that $u$ has three active neighbors, and without loss of generality assume that $p, q, r$ are active and $s$ is inactive. Then there are two possibilities, either $s$ has three active neighbors ($h,g, f$), in which case $u$ and $s$ remain inactive (see Figure \ref{fig: casos 124 33}).
The difference with the rule $12$ is that we can not decide immediately if the cell $u$ remains inactive when the sum at time $\tau$ is 3. Indeed, in the case depicted in Figure \ref{fig: casos 124 33}, it is possible that $s$ becomes active in a time-step later than $\tau-1$. 


Thus we need study only the case when the sum at time $\tau-1$ of the states of neighbors of $s$ is 0 or equivalently $x_f^{\tau-2}=x_g^{\tau-2}=x_h^{\tau-2}=0$. Note that, by the OR technique, the fact that $x_f^{\tau-2}=x_g^{\tau-2}=x_h^{\tau-2}=0$ means that all the cells in the left side border of $D_\tau$ are initially inactive, as shown in Figure \ref{fig: comp2}.

\begin{figure}
  \centering
  \begin{subfigure}[t]{0.485\textwidth}
    \centering  
    \includegraphics[width=.95\textwidth]{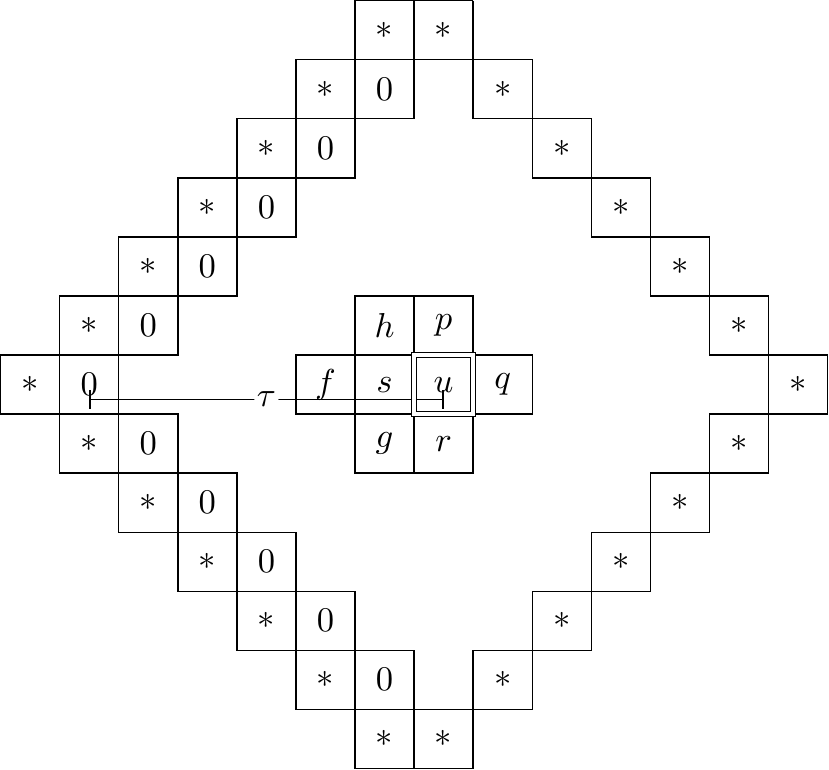}  
    \caption{Configuration at time $0$.}
  \end{subfigure}      
  \begin{subfigure}[t]{0.485\textwidth}
    \centering 
    \includegraphics[width=.95\textwidth]{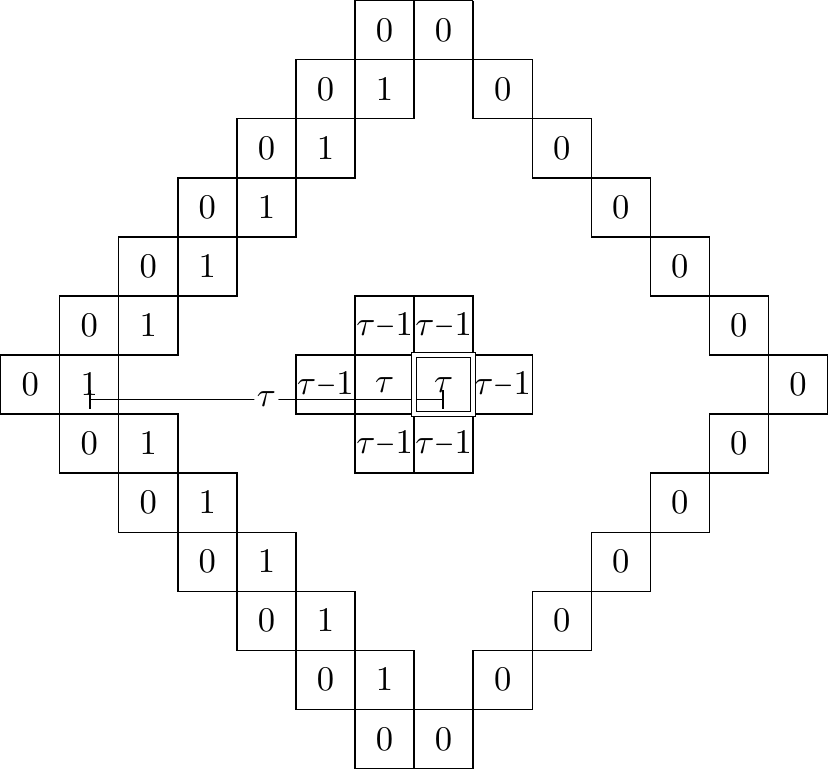}   
    \caption{Diagram with the times when a cell can be activated.}
  \end{subfigure}   
  \caption{Schedule for to compute $x_u^{\tau}$ if  $x_f^{\tau-2}=x_g^{\tau-2}=x_h^{\tau-2}=0$.}     
  \label{fig: comp2}
\end{figure} 

Knowing that $x_f^{\tau-2}=x_g^{\tau-2}=x_h^{\tau-2}=0$, we can compute their states at time-step $\tau-1$, considering the OR-techinique in the disc $D_{\tau+1}$. Using this information, we can compute the state of $s$ in time $\tau$, i.e. compute $x_s^{\tau}$.

Remember that we are in the case where $x_p^{\tau-1}=x_q^{\tau-1}=x_r^{\tau-1}=1$ and $x_s^{\tau-1}=0$. If the cell $s$ becomes  active at time $\tau$, then $u$ will have four active neighbors at time $\tau$, and it will become active. Now we suppose that $s$ also remain inactive at time $\tau$. Again, we have two possible cases:

  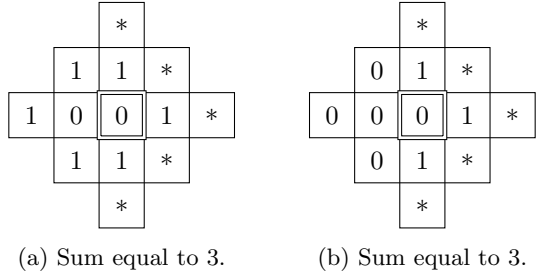
\begin{figure}
    \centering         
    \begin{subfigure}[t]{0.32\textwidth}
              \centering
                \begin{tikzpicture}[scale=0.3]
                \draw  (-8,11) rectangle (2,9);
                \draw  (-4,15) rectangle (-2,5);
                \draw  (-6,13) rectangle (0,7);
                \node at (-3,14) {$*$};
                \node at (-1,12) {$*$};
                \node at (1,10)  {$*$};
                \node at (-1,8)  {$*$};
                \node at (-3,6)  {$*$};
                \node at (-5,8)  {$1$};
                \node at (-7,10) {$1$};
                \node at (-5,12) {$1$};
                \node at (-3,12) {$1$};
                \node at (-1,10) {$1$};
                \node at (-3,8)  {$1$};
                \node at (-5,10) {$0$};
                \node at (-3,10) {$0$};
                \draw [double distance=1pt] (-4,11) rectangle (-2,9); 
                \end{tikzpicture}
              \caption{Sum equal to 3.}
              \label{fig: casos 124 303}
    \end{subfigure}
    \begin{subfigure}[t]{0.32\textwidth}
              \centering
                \begin{tikzpicture}[scale=0.3]
                \draw  (-8,11) rectangle (2,9);
                \draw  (-4,15) rectangle (-2,5);
                \draw  (-6,13) rectangle (0,7);
                \node at (-3,14) {$*$};
                \node at (-1,12) {$*$};
                \node at (1,10)  {$*$};
                \node at (-1,8)  {$*$};
                \node at (-3,6)  {$*$};
                \node at (-5,8)  {$0$};
                \node at (-7,10) {$0$};
                \node at (-5,12) {$0$};
                \node at (-3,12) {$1$};
                \node at (-1,10) {$1$};
                \node at (-3,8)  {$1$};
                \node at (-5,10) {$0$};
                \node at (-3,10) {$0$};
                \draw [double distance=1pt] (-4,11) rectangle (-2,9); 
                \end{tikzpicture}
              \caption{Sum equal to 3.}
              \label{fig: casos 124 300}
    \end{subfigure} 
  \caption{Possible cases of rule $124$ at time $\tau+1$ such that $u$ remain at state 0.}
  \label{fig: casos 124 N0}
\end{figure}

In the case in Figure \ref{fig: casos 124 303} (i.e., when $s$ remains inactive at time $\tau$ because $f,g$ and $h$ were active at time $\tau-1$) the cell $u$ is stable.

For the case shown in Figure \ref{fig: casos 124 300} (i.e. $s$ remains inactive at time $\tau$ because e $f,g$ and $h$ were inactive at time $\tau-1$) we must repeat the previous analysis.  Indeed, we know that $x_f^{\tau-1}=x_g^{\tau-1}=x_h^{\tau-1}=0$. The OR technique implies that every cell in the left  border on the disc $D_{\tau+1}$ (see Figure \ref{fig: comp2}) have to be initially inactive too. In this case, however we study the next disc $D_{\tau_+1}$, shifting it one cell to the left, as the Figure \ref{fig: comp3}.

\begin{figure}
  \centering
  \begin{subfigure}[t]{0.485\textwidth}
    \centering  
    \includegraphics[width=.95\textwidth]{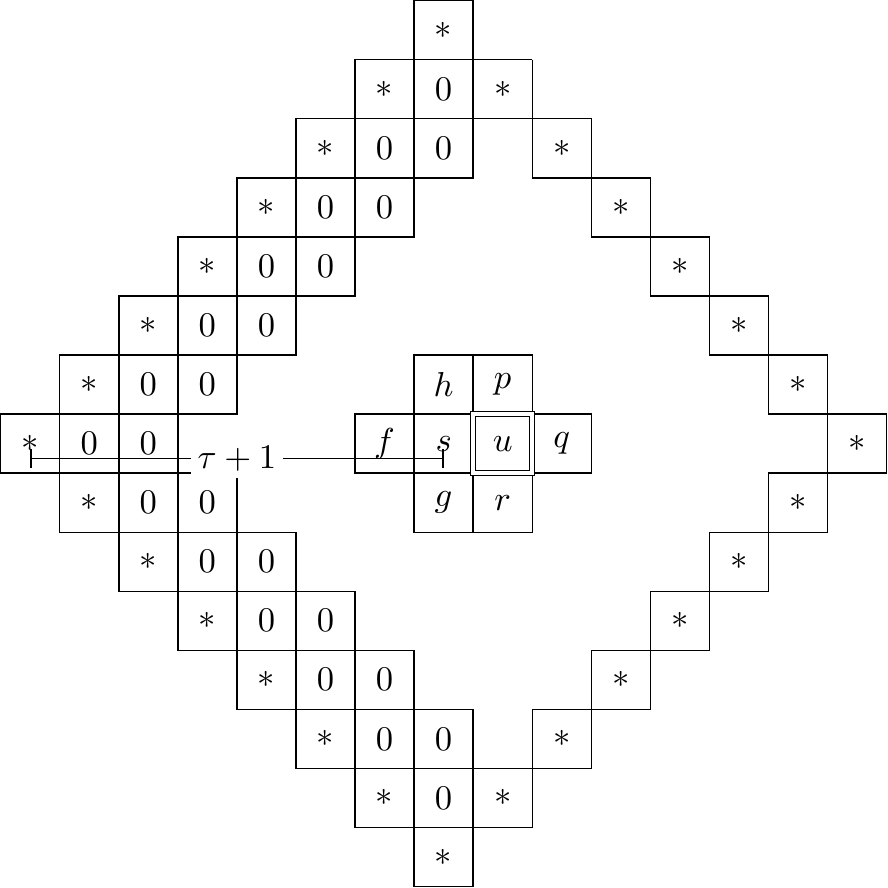}  
    \caption{Configuration at time $0$.}
  \end{subfigure}      
  \begin{subfigure}[t]{0.485\textwidth}
    \centering 
    \includegraphics[width=.95\textwidth]{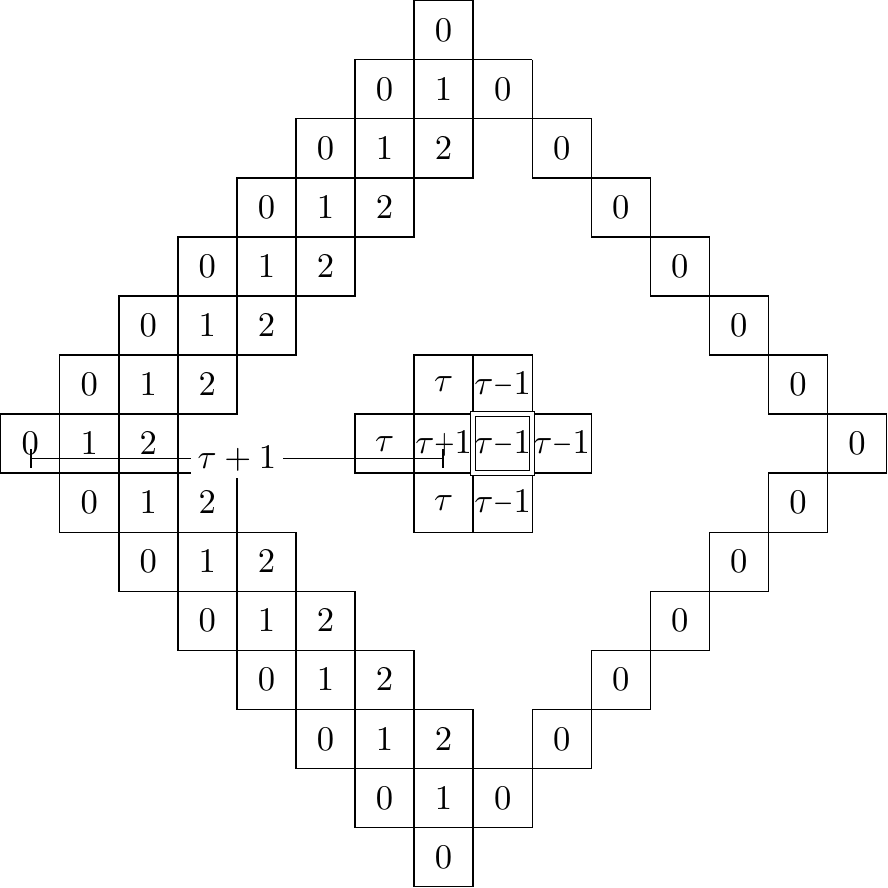}   
    \caption{Diagram with the times when a cell can be activated.}
  \end{subfigure}   
  \caption{Schedule for to compute $x^{\tau}_s$  if $x_f^{\tau-1}=x_g^{\tau-1}=x_h^{\tau-1}=0$.}     
  \label{fig: comp3}
\end{figure} 

This new disc is centered in the cell $s$ and (considering only the sides at the north-west and south-west), consists in the sites at distance $\tau+1$ from $s$. Again, using the OR technique, we can compute the states of cells  $f$, $g$ and $h$ at time $\tau$, and then the sate of cell $s$ at time $\tau+1$. 

Again, if the cell $s$ becomes active at time $\tau+1$, then the problem is solved, because $u$ becomes active at time $\tau+2$.  Further, suppose that $s$ is not active at time $\tau+1$. Notice that this means that $g, h$ and $f$ must have the same state at time $\tau$. Remember that $p$ and $r$ are active at time $\tau-1$, then cells $h$ and $f$ have at least one active neighbor at time $\tau-1$. If $h, f$ and $g$ are active at time $\tau$, then $s$ will be stable, as well as $u$. If $h, f$ and $g$ are inactive at time $\tau$, it means that $h$ and $f$ have three active neighbors at time $\tau-1$, including cells $(-2,1)$ and $(-2,-1)$. Since these cells are also neighbors of $g$, and $g$ remains inactive at time $\tau$, necessarily cell $(-3,0)$ must be active at time $\tau-1$. This means that $f,g$ and $h$ have three active neighbors, so they are stable. Implying also that $s$ and $u$ are stable.

We deduce that either $u$ becomes active at time $\tau, \tau+1$ or $\tau+2$,  or $u$ is stable. 


\begin{lemma}
Let $F$ be rule $124$. Given a finite configuration $x$, a cell $u$ and $\tau$ the distance from $u$ to the nearest active cell. Then $u$ becomes active at time $\tau, \tau +1$ or $\tau+2$, or $u$ is stable.
\end{lemma}

Now we give an algorithm for to decide the \stability in \NC for the rule $124$. The algorithms for rules $12$ and $123$ can be deduced from this algorithm.

\begin{theorem}\label{thm: 12** square NC}
  \stability is in {\NC} for rules $12$, $123$ and $124$.
\end{theorem}

\begin{proof}
Let $(x, u)$ be an input of {\stability}, $x$ is a finite configuration of dimensions $n\times n$ and $u$ is a site in $[n]\times [n]$. The following parallel algorithm is able to decide \stability using the fast computation of the first neighbors of $u$ by the OR technique. 
Let ${N^2(u)}$ be the set of cells at distance at most $2$ from $u$. For $t\geq 0$, we call $x_{N(u)}^t$ the set of states at time $t$ of all cells in $N(u)$.

\begin{algorithm}[h]
\caption{Solving \stability 124 }\label{alg: PRED 12** and 1L**}
\begin{algorithmic}[1]
\REQUIRE $x$ a finite configuration of dimensions $n \times n$ and $u \in [n]\times [n]$.
\STATE Compute $\tau$ the distance form $u$ to the nearest  active cell in $x$. 
\STATE Compute $x^{\tau-2}_{N^2(u)} $ using the OR technique and the corridors. 
\IF{$x^{\tau}_{u}=1 $}
  \RETURN \emph{Reject}
\ENDIF
\IF{$x^{\tau}_{u}=0$ and $x^{\tau-2}_{N^2(u)}$ is as Figure \ref{fig: casos 124 33}}
  \RETURN \emph{Accept}
\ENDIF
\STATE Compute $s$ the neighbor of  $u$ such that $x^{\tau-1}_{s}=0 $. 
\STATE Compute $x^{\tau-1}_{N^2(s)} $ using the OR technique and the corridors. 
\IF{$x^{\tau+1}_{u}=1 $}
  \RETURN \emph{Reject}
\ENDIF
\IF{$x^{\tau+1}_{u}=0$ and $x^{\tau-1}_{N^2(u)}$ is as Figure \ref{fig: casos 124 33}}
  \RETURN \emph{Accept}
\ENDIF
\RETURN \emph{Reject}
\end{algorithmic}
\end{algorithm}

Let $N=n^{2}$ the size of the input.
Step {\bf 1} can be done in $\cO(\log N)$ time with $ \cO(N)$ processors: one processor for each cell for to choose the actives cells and to compute its distances with $u$, then in  $\cO(\log N)$ compute the nears cell to $u$.
Steps {\bf 2} and {\bf 10} can be done in $\cO(\log N)$ time with $ \cO(N)$ processors: the OR technique and the corridors can be computed with prefix sum algorithm (see Proposition \ref{prop: prefix-sum}) for the computation of consecutive $\vee$ and parity of $z$ in the corridors.
The others steps can be computed in $\cO(\log n)$ time in using a sequential algorithm. 
\end{proof}

\subsection{ Turing Universal Rules}\label{sec: p-com square}

For the rules $2$ and $24$ the  \stability problem is \Pt-Complete by reducing a restricted version of the Circuit Value Problem \cite{Greenlaw:1995} to this problem. 
Instances of circuit value problem are encoded into a configuration of the CA $2$ and  $24$ using the idea in the proof of the {\Pt}-completeness of Planar Circuit Value Problem (PCV) \cite{Goldschlager1977}. Moreover, we use an aproach given in \cite{DBLP:journals/jcss/GolesMPT18}, were the authors show that a two-dimensional automaton capable of simulating \emph{wires}, \emph{OR gates}, \emph{AND gates} and \emph{crossing gadgets} is {\Pt}-Complete.

In Figure \ref{fig: gates} we will give the gadgets that simulate this structures for rules $2$ and $24$. We remark that both rules have the same structures, because the patterns with four active neighbors  never appear in the gadgets.

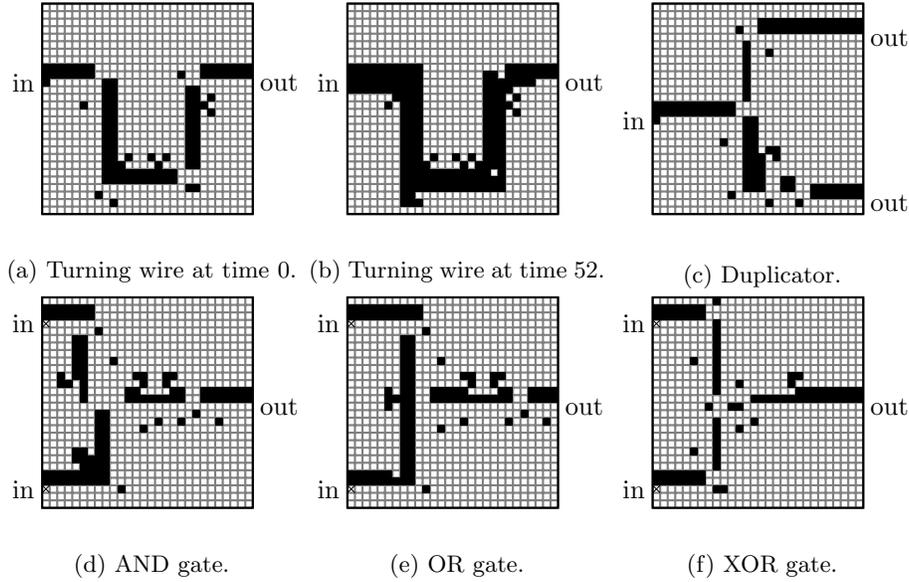
\begin{figure}
        \centering
        \begin{subfigure}[t]{0.325\textwidth}
            \begin{center}
                \begin{tikzpicture}[scale=0.1]
\draw[color=gray, thick] (1,1) grid (29,29);
\draw[color=black, thick] (1,1) rectangle (29,29);

\fill[black] (1,20) rectangle (2,21); 
\fill[black] (2,20) rectangle (3,21); 
\fill[black] (3,20) rectangle (4,21); 
\fill[black] (4,20) rectangle (5,21); 
\fill[black] (5,20) rectangle (6,21); 
\fill[black] (6,20) rectangle (7,21); 
\fill[black] (7,20) rectangle (8,21); 
\fill[black] (22,20) rectangle (23,21); 
\fill[black] (23,20) rectangle (24,21); 
\fill[black] (24,20) rectangle (25,21); 
\fill[black] (25,20) rectangle (26,21); 
\fill[black] (26,20) rectangle (27,21); 
\fill[black] (27,20) rectangle (28,21); 
\fill[black] (28,20) rectangle (29,21); 
\fill[black] (1,19) rectangle (2,20); 
\fill[black] (2,19) rectangle (3,20); 
\fill[black] (3,19) rectangle (4,20); 
\fill[black] (4,19) rectangle (5,20); 
\fill[black] (5,19) rectangle (6,20); 
\fill[black] (6,19) rectangle (7,20); 
\fill[black] (7,19) rectangle (8,20); 
\fill[black] (19,19) rectangle (20,20); 
\fill[black] (22,19) rectangle (23,20); 
\fill[black] (23,19) rectangle (24,20); 
\fill[black] (24,19) rectangle (25,20); 
\fill[black] (25,19) rectangle (26,20); 
\fill[black] (26,19) rectangle (27,20); 
\fill[black] (27,19) rectangle (28,20); 
\fill[black] (28,19) rectangle (29,20); 
\fill[black] (1,18) rectangle (2,19); 
\fill[black] (9,18) rectangle (10,19); 
\fill[black] (10,18) rectangle (11,19); 
\fill[black] (9,17) rectangle (10,18); 
\fill[black] (10,17) rectangle (11,18); 
\fill[black] (20,17) rectangle (21,18); 
\fill[black] (21,17) rectangle (22,18); 
\fill[black] (9,16) rectangle (10,17); 
\fill[black] (10,16) rectangle (11,17); 
\fill[black] (20,16) rectangle (21,17); 
\fill[black] (21,16) rectangle (22,17); 
\fill[black] (23,16) rectangle (24,17); 
\fill[black] (6,15) rectangle (7,16); 
\fill[black] (9,15) rectangle (10,16); 
\fill[black] (10,15) rectangle (11,16); 
\fill[black] (20,15) rectangle (21,16); 
\fill[black] (21,15) rectangle (22,16); 
\fill[black] (22,15) rectangle (23,16); 
\fill[black] (9,14) rectangle (10,15); 
\fill[black] (10,14) rectangle (11,15); 
\fill[black] (20,14) rectangle (21,15); 
\fill[black] (21,14) rectangle (22,15); 
\fill[black] (23,14) rectangle (24,15); 
\fill[black] (9,13) rectangle (10,14); 
\fill[black] (10,13) rectangle (11,14); 
\fill[black] (20,13) rectangle (21,14); 
\fill[black] (21,13) rectangle (22,14); 
\fill[black] (9,12) rectangle (10,13); 
\fill[black] (10,12) rectangle (11,13); 
\fill[black] (20,12) rectangle (21,13); 
\fill[black] (21,12) rectangle (22,13); 
\fill[black] (9,11) rectangle (10,12); 
\fill[black] (10,11) rectangle (11,12); 
\fill[black] (20,11) rectangle (21,12); 
\fill[black] (21,11) rectangle (22,12); 
\fill[black] (9,10) rectangle (10,11); 
\fill[black] (10,10) rectangle (11,11); 
\fill[black] (20,10) rectangle (21,11); 
\fill[black] (21,10) rectangle (22,11); 
\fill[black] (9,9) rectangle (10,10); 
\fill[black] (10,9) rectangle (11,10); 
\fill[black] (20,9) rectangle (21,10); 
\fill[black] (21,9) rectangle (22,10); 
\fill[black] (9,8) rectangle (10,9); 
\fill[black] (10,8) rectangle (11,9); 
\fill[black] (12,8) rectangle (13,9); 
\fill[black] (15,8) rectangle (16,9); 
\fill[black] (17,8) rectangle (18,9); 
\fill[black] (20,8) rectangle (21,9); 
\fill[black] (21,8) rectangle (22,9); 
\fill[black] (9,7) rectangle (10,8); 
\fill[black] (10,7) rectangle (11,8); 
\fill[black] (11,7) rectangle (12,8); 
\fill[black] (16,7) rectangle (17,8); 
\fill[black] (20,7) rectangle (21,8); 
\fill[black] (21,7) rectangle (22,8); 
\fill[black] (9,6) rectangle (10,7); 
\fill[black] (10,6) rectangle (11,7); 
\fill[black] (11,6) rectangle (12,7); 
\fill[black] (12,6) rectangle (13,7); 
\fill[black] (13,6) rectangle (14,7); 
\fill[black] (14,6) rectangle (15,7); 
\fill[black] (15,6) rectangle (16,7); 
\fill[black] (16,6) rectangle (17,7); 
\fill[black] (17,6) rectangle (18,7); 
\fill[black] (18,6) rectangle (19,7); 
\fill[black] (9,5) rectangle (10,6); 
\fill[black] (10,5) rectangle (11,6); 
\fill[black] (11,5) rectangle (12,6); 
\fill[black] (12,5) rectangle (13,6); 
\fill[black] (13,5) rectangle (14,6); 
\fill[black] (14,5) rectangle (15,6); 
\fill[black] (15,5) rectangle (16,6); 
\fill[black] (16,5) rectangle (17,6); 
\fill[black] (17,5) rectangle (18,6); 
\fill[black] (18,5) rectangle (19,6); 
\fill[black] (20,4) rectangle (21,5); 
\fill[black] (21,4) rectangle (22,5); 
\fill[black] (8,3) rectangle (9,4); 
\fill[black] (10,2) rectangle (11,3); 
\node at (-1.5,18.5) {in};

\node at (32.5,18.5) {out};

\end{tikzpicture} 
            \end{center}
            \caption{Turning wire at time 0.}
            \label{fig: wire1}
        \end{subfigure}            
        \begin{subfigure}[t]{0.325\textwidth}
            \begin{center}
             \begin{tikzpicture}[scale=0.1]
\draw[color=gray, thick] (1,1) grid (29,29);
\draw[color=black, thick] (1,1) rectangle (29,29);
\fill[black] (1,20) rectangle (2,21); 
\fill[black] (2,20) rectangle (3,21); 
\fill[black] (3,20) rectangle (4,21); 
\fill[black] (4,20) rectangle (5,21); 
\fill[black] (5,20) rectangle (6,21); 
\fill[black] (6,20) rectangle (7,21); 
\fill[black] (7,20) rectangle (8,21); 
\fill[black] (8,20) rectangle (9,21); 
\fill[black] (9,20) rectangle (10,21); 
\fill[black] (10,20) rectangle (11,21); 
\fill[black] (22,20) rectangle (23,21); 
\fill[black] (23,20) rectangle (24,21); 
\fill[black] (24,20) rectangle (25,21); 
\fill[black] (25,20) rectangle (26,21); 
\fill[black] (26,20) rectangle (27,21); 
\fill[black] (27,20) rectangle (28,21); 
\fill[black] (28,20) rectangle (29,21); 
\fill[black] (1,19) rectangle (2,20); 
\fill[black] (2,19) rectangle (3,20); 
\fill[black] (3,19) rectangle (4,20); 
\fill[black] (4,19) rectangle (5,20); 
\fill[black] (5,19) rectangle (6,20); 
\fill[black] (6,19) rectangle (7,20); 
\fill[black] (7,19) rectangle (8,20); 
\fill[black] (8,19) rectangle (9,20); 
\fill[black] (9,19) rectangle (10,20); 
\fill[black] (10,19) rectangle (11,20); 
\fill[black] (19,19) rectangle (20,20); 
\fill[black] (20,19) rectangle (21,20); 
\fill[black] (22,19) rectangle (23,20); 
\fill[black] (23,19) rectangle (24,20); 
\fill[black] (24,19) rectangle (25,20); 
\fill[black] (25,19) rectangle (26,20); 
\fill[black] (26,19) rectangle (27,20); 
\fill[black] (27,19) rectangle (28,20); 
\fill[black] (28,19) rectangle (29,20); 
\fill[black] (1,18) rectangle (2,19); 
\fill[black] (2,18) rectangle (3,19); 
\fill[black] (3,18) rectangle (4,19); 
\fill[black] (4,18) rectangle (5,19); 
\fill[black] (5,18) rectangle (6,19); 
\fill[black] (6,18) rectangle (7,19); 
\fill[black] (7,18) rectangle (8,19); 
\fill[black] (8,18) rectangle (9,19); 
\fill[black] (9,18) rectangle (10,19); 
\fill[black] (10,18) rectangle (11,19); 
\fill[black] (19,18) rectangle (20,19); 
\fill[black] (20,18) rectangle (21,19); 
\fill[black] (21,18) rectangle (22,19); 
\fill[black] (22,18) rectangle (23,19); 
\fill[black] (23,18) rectangle (24,19); 
\fill[black] (24,18) rectangle (25,19); 
\fill[black] (25,18) rectangle (26,19); 
\fill[black] (1,17) rectangle (2,18); 
\fill[black] (2,17) rectangle (3,18); 
\fill[black] (3,17) rectangle (4,18); 
\fill[black] (4,17) rectangle (5,18); 
\fill[black] (5,17) rectangle (6,18); 
\fill[black] (6,17) rectangle (7,18); 
\fill[black] (7,17) rectangle (8,18); 
\fill[black] (8,17) rectangle (9,18); 
\fill[black] (9,17) rectangle (10,18); 
\fill[black] (10,17) rectangle (11,18); 
\fill[black] (19,17) rectangle (20,18); 
\fill[black] (20,17) rectangle (21,18); 
\fill[black] (21,17) rectangle (22,18); 
\fill[black] (22,17) rectangle (23,18); 
\fill[black] (7,16) rectangle (8,17); 
\fill[black] (8,16) rectangle (9,17); 
\fill[black] (9,16) rectangle (10,17); 
\fill[black] (10,16) rectangle (11,17); 
\fill[black] (19,16) rectangle (20,17); 
\fill[black] (20,16) rectangle (21,17); 
\fill[black] (21,16) rectangle (22,17); 
\fill[black] (23,16) rectangle (24,17); 
\fill[black] (6,15) rectangle (7,16); 
\fill[black] (8,15) rectangle (9,16); 
\fill[black] (9,15) rectangle (10,16); 
\fill[black] (10,15) rectangle (11,16); 
\fill[black] (19,15) rectangle (20,16); 
\fill[black] (20,15) rectangle (21,16); 
\fill[black] (21,15) rectangle (22,16); 
\fill[black] (22,15) rectangle (23,16); 
\fill[black] (8,14) rectangle (9,15); 
\fill[black] (9,14) rectangle (10,15); 
\fill[black] (10,14) rectangle (11,15); 
\fill[black] (19,14) rectangle (20,15); 
\fill[black] (20,14) rectangle (21,15); 
\fill[black] (21,14) rectangle (22,15); 
\fill[black] (23,14) rectangle (24,15); 
\fill[black] (8,13) rectangle (9,14); 
\fill[black] (9,13) rectangle (10,14); 
\fill[black] (10,13) rectangle (11,14); 
\fill[black] (19,13) rectangle (20,14); 
\fill[black] (20,13) rectangle (21,14); 
\fill[black] (21,13) rectangle (22,14); 
\fill[black] (8,12) rectangle (9,13); 
\fill[black] (9,12) rectangle (10,13); 
\fill[black] (10,12) rectangle (11,13); 
\fill[black] (19,12) rectangle (20,13); 
\fill[black] (20,12) rectangle (21,13); 
\fill[black] (21,12) rectangle (22,13); 
\fill[black] (8,11) rectangle (9,12); 
\fill[black] (9,11) rectangle (10,12); 
\fill[black] (10,11) rectangle (11,12); 
\fill[black] (19,11) rectangle (20,12); 
\fill[black] (20,11) rectangle (21,12); 
\fill[black] (21,11) rectangle (22,12); 
\fill[black] (8,10) rectangle (9,11); 
\fill[black] (9,10) rectangle (10,11); 
\fill[black] (10,10) rectangle (11,11); 
\fill[black] (19,10) rectangle (20,11); 
\fill[black] (20,10) rectangle (21,11); 
\fill[black] (21,10) rectangle (22,11); 
\fill[black] (8,9) rectangle (9,10); 
\fill[black] (9,9) rectangle (10,10); 
\fill[black] (10,9) rectangle (11,10); 
\fill[black] (19,9) rectangle (20,10); 
\fill[black] (20,9) rectangle (21,10); 
\fill[black] (21,9) rectangle (22,10); 
\fill[black] (8,8) rectangle (9,9); 
\fill[black] (9,8) rectangle (10,9); 
\fill[black] (10,8) rectangle (11,9); 
\fill[black] (12,8) rectangle (13,9); 
\fill[black] (15,8) rectangle (16,9); 
\fill[black] (17,8) rectangle (18,9); 
\fill[black] (19,8) rectangle (20,9); 
\fill[black] (20,8) rectangle (21,9); 
\fill[black] (21,8) rectangle (22,9); 
\fill[black] (8,7) rectangle (9,8); 
\fill[black] (9,7) rectangle (10,8); 
\fill[black] (10,7) rectangle (11,8); 
\fill[black] (11,7) rectangle (12,8); 
\fill[black] (16,7) rectangle (17,8); 
\fill[black] (18,7) rectangle (19,8); 
\fill[black] (19,7) rectangle (20,8); 
\fill[black] (20,7) rectangle (21,8); 
\fill[black] (21,7) rectangle (22,8); 
\fill[black] (8,6) rectangle (9,7); 
\fill[black] (9,6) rectangle (10,7); 
\fill[black] (10,6) rectangle (11,7); 
\fill[black] (11,6) rectangle (12,7); 
\fill[black] (12,6) rectangle (13,7); 
\fill[black] (13,6) rectangle (14,7); 
\fill[black] (14,6) rectangle (15,7); 
\fill[black] (15,6) rectangle (16,7); 
\fill[black] (16,6) rectangle (17,7); 
\fill[black] (17,6) rectangle (18,7); 
\fill[black] (18,6) rectangle (19,7); 
\fill[black] (19,6) rectangle (20,7); 
\fill[black] (21,6) rectangle (22,7); 
\fill[black] (8,5) rectangle (9,6); 
\fill[black] (9,5) rectangle (10,6); 
\fill[black] (10,5) rectangle (11,6); 
\fill[black] (11,5) rectangle (12,6); 
\fill[black] (12,5) rectangle (13,6); 
\fill[black] (13,5) rectangle (14,6); 
\fill[black] (14,5) rectangle (15,6); 
\fill[black] (15,5) rectangle (16,6); 
\fill[black] (16,5) rectangle (17,6); 
\fill[black] (17,5) rectangle (18,6); 
\fill[black] (18,5) rectangle (19,6); 
\fill[black] (19,5) rectangle (20,6); 
\fill[black] (20,5) rectangle (21,6); 
\fill[black] (21,5) rectangle (22,6); 
\fill[black] (8,4) rectangle (9,5); 
\fill[black] (9,4) rectangle (10,5); 
\fill[black] (10,4) rectangle (11,5); 
\fill[black] (11,4) rectangle (12,5); 
\fill[black] (12,4) rectangle (13,5); 
\fill[black] (13,4) rectangle (14,5); 
\fill[black] (14,4) rectangle (15,5); 
\fill[black] (15,4) rectangle (16,5); 
\fill[black] (16,4) rectangle (17,5); 
\fill[black] (17,4) rectangle (18,5); 
\fill[black] (18,4) rectangle (19,5); 
\fill[black] (19,4) rectangle (20,5); 
\fill[black] (20,4) rectangle (21,5); 
\fill[black] (21,4) rectangle (22,5); 
\fill[black] (8,3) rectangle (9,4); 
\fill[black] (9,3) rectangle (10,4); 
\fill[black] (8,2) rectangle (9,3); 
\fill[black] (9,2) rectangle (10,3); 
\fill[black] (10,2) rectangle (11,3); 
\node at (-1.5,18.5) {in};

\node at (32.5,18.5) {out};

\end{tikzpicture} 
            \end{center}
            \caption{Turning wire at time 52.}
            \label{fig: wire2}
        \end{subfigure}            
        \begin{subfigure}[t]{0.325\textwidth}
            \begin{center}
             \begin{tikzpicture}[scale=0.1]
\draw[color=gray, thick] (1,1) grid (29,29);
\draw[color=black, thick] (1,1) rectangle (29,29);
\fill[black] (15,26) rectangle (16,27); 
\fill[black] (16,26) rectangle (17,27); 
\fill[black] (17,26) rectangle (18,27); 
\fill[black] (18,26) rectangle (19,27); 
\fill[black] (19,26) rectangle (20,27); 
\fill[black] (20,26) rectangle (21,27); 
\fill[black] (21,26) rectangle (22,27); 
\fill[black] (22,26) rectangle (23,27); 
\fill[black] (23,26) rectangle (24,27); 
\fill[black] (24,26) rectangle (25,27); 
\fill[black] (25,26) rectangle (26,27); 
\fill[black] (26,26) rectangle (27,27); 
\fill[black] (27,26) rectangle (28,27); 
\fill[black] (28,26) rectangle (29,27); 
\fill[black] (12,25) rectangle (13,26); 
\fill[black] (15,25) rectangle (16,26); 
\fill[black] (16,25) rectangle (17,26); 
\fill[black] (17,25) rectangle (18,26); 
\fill[black] (18,25) rectangle (19,26); 
\fill[black] (19,25) rectangle (20,26); 
\fill[black] (20,25) rectangle (21,26); 
\fill[black] (21,25) rectangle (22,26); 
\fill[black] (22,25) rectangle (23,26); 
\fill[black] (23,25) rectangle (24,26); 
\fill[black] (24,25) rectangle (25,26); 
\fill[black] (25,25) rectangle (26,26); 
\fill[black] (26,25) rectangle (27,26); 
\fill[black] (27,25) rectangle (28,26); 
\fill[black] (28,25) rectangle (29,26); 
\fill[black] (13,23) rectangle (14,24); 
\fill[black] (13,22) rectangle (14,23); 
\fill[black] (16,22) rectangle (17,23); 
\fill[black] (13,21) rectangle (14,22); 
\fill[black] (13,20) rectangle (14,21); 
\fill[black] (13,19) rectangle (14,20); 
\fill[black] (13,18) rectangle (14,19); 
\fill[black] (13,17) rectangle (14,18); 
\fill[black] (13,16) rectangle (14,17); 
\fill[black] (1,15) rectangle (2,16); 
\fill[black] (2,15) rectangle (3,16); 
\fill[black] (3,15) rectangle (4,16); 
\fill[black] (4,15) rectangle (5,16); 
\fill[black] (5,15) rectangle (6,16); 
\fill[black] (6,15) rectangle (7,16); 
\fill[black] (7,15) rectangle (8,16); 
\fill[black] (8,15) rectangle (9,16); 
\fill[black] (9,15) rectangle (10,16); 
\fill[black] (10,15) rectangle (11,16); 
\fill[black] (11,15) rectangle (12,16); 
\fill[black] (1,14) rectangle (2,15); 
\fill[black] (2,14) rectangle (3,15); 
\fill[black] (3,14) rectangle (4,15); 
\fill[black] (4,14) rectangle (5,15); 
\fill[black] (5,14) rectangle (6,15); 
\fill[black] (6,14) rectangle (7,15); 
\fill[black] (7,14) rectangle (8,15); 
\fill[black] (8,14) rectangle (9,15); 
\fill[black] (9,14) rectangle (10,15); 
\fill[black] (10,14) rectangle (11,15); 
\fill[black] (11,14) rectangle (12,15); 
\fill[black] (1,13) rectangle (2,14); 
\fill[black] (13,13) rectangle (14,14); 
\fill[black] (14,13) rectangle (15,14); 
\fill[black] (13,12) rectangle (14,13); 
\fill[black] (14,12) rectangle (15,13); 
\fill[black] (13,11) rectangle (14,12); 
\fill[black] (14,11) rectangle (15,12); 
\fill[black] (10,10) rectangle (11,11); 
\fill[black] (13,10) rectangle (14,11); 
\fill[black] (14,10) rectangle (15,11); 
\fill[black] (13,9) rectangle (14,10); 
\fill[black] (14,9) rectangle (15,10); 
\fill[black] (16,9) rectangle (17,10); 
\fill[black] (17,9) rectangle (18,10); 
\fill[black] (13,8) rectangle (14,9); 
\fill[black] (14,8) rectangle (15,9); 
\fill[black] (15,8) rectangle (16,9); 
\fill[black] (17,8) rectangle (18,9); 
\fill[black] (13,7) rectangle (14,8); 
\fill[black] (14,7) rectangle (15,8); 
\fill[black] (15,7) rectangle (16,8); 
\fill[black] (13,6) rectangle (14,7); 
\fill[black] (14,6) rectangle (15,7); 
\fill[black] (15,6) rectangle (16,7); 
\fill[black] (13,5) rectangle (14,6); 
\fill[black] (14,5) rectangle (15,6); 
\fill[black] (15,5) rectangle (16,6); 
\fill[black] (18,5) rectangle (19,6); 
\fill[black] (19,5) rectangle (20,6); 
\fill[black] (13,4) rectangle (14,5); 
\fill[black] (14,4) rectangle (15,5); 
\fill[black] (15,4) rectangle (16,5); 
\fill[black] (18,4) rectangle (19,5); 
\fill[black] (19,4) rectangle (20,5); 
\fill[black] (22,4) rectangle (23,5); 
\fill[black] (23,4) rectangle (24,5); 
\fill[black] (24,4) rectangle (25,5); 
\fill[black] (25,4) rectangle (26,5); 
\fill[black] (26,4) rectangle (27,5); 
\fill[black] (27,4) rectangle (28,5); 
\fill[black] (28,4) rectangle (29,5); 
\fill[black] (11,3) rectangle (12,4); 
\fill[black] (22,3) rectangle (23,4); 
\fill[black] (23,3) rectangle (24,4); 
\fill[black] (24,3) rectangle (25,4); 
\fill[black] (25,3) rectangle (26,4); 
\fill[black] (26,3) rectangle (27,4); 
\fill[black] (27,3) rectangle (28,4); 
\fill[black] (28,3) rectangle (29,4); 
\fill[black] (16,2) rectangle (17,3); 
\fill[black] (20,2) rectangle (21,3); 

\node at (-1.5,13.5) {in};

\node at (32.5,24.5) {out};
\node at (32.5,2.5) {out};
\end{tikzpicture}
            \end{center}
            \caption{Duplicator.}
            \label{fig: x2}
        \end{subfigure}            

        \begin{subfigure}[t]{0.325\textwidth}
            \begin{center}
              \begin{tikzpicture}[scale=0.1]
\draw[color=gray, thick] (1,1) grid (29,29);
\draw[color=black, thick] (1,1) rectangle (29,29);

\fill[black] (1,27) rectangle (2,28); 
\fill[black] (2,27) rectangle (3,28); 
\fill[black] (3,27) rectangle (4,28); 
\fill[black] (4,27) rectangle (5,28); 
\fill[black] (5,27) rectangle (6,28); 
\fill[black] (6,27) rectangle (7,28); 
\fill[black] (7,27) rectangle (8,28); 
\fill[black] (1,26) rectangle (2,27); 
\fill[black] (2,26) rectangle (3,27); 
\fill[black] (3,26) rectangle (4,27); 
\fill[black] (4,26) rectangle (5,27); 
\fill[black] (5,26) rectangle (6,27); 
\fill[black] (6,26) rectangle (7,27); 
\fill[black] (7,26) rectangle (8,27); 
\fill[black] (8,24) rectangle (9,25); 
\fill[black] (5,23) rectangle (6,24); 
\fill[black] (6,23) rectangle (7,24); 
\fill[black] (5,22) rectangle (6,23); 
\fill[black] (6,22) rectangle (7,23); 
\fill[black] (5,21) rectangle (6,22); 
\fill[black] (6,21) rectangle (7,22); 
\fill[black] (5,20) rectangle (6,21); 
\fill[black] (6,20) rectangle (7,21); 
\fill[black] (10,20) rectangle (11,21); 
\fill[black] (5,19) rectangle (6,20); 
\fill[black] (6,19) rectangle (7,20); 
\fill[black] (3,18) rectangle (4,19); 
\fill[black] (5,18) rectangle (6,19); 
\fill[black] (6,18) rectangle (7,19); 
\fill[black] (13,18) rectangle (14,19); 
\fill[black] (14,18) rectangle (15,19); 
\fill[black] (17,18) rectangle (18,19); 
\fill[black] (18,18) rectangle (19,19); 
\fill[black] (3,17) rectangle (4,18); 
\fill[black] (4,17) rectangle (5,18); 
\fill[black] (6,17) rectangle (7,18); 
\fill[black] (14,17) rectangle (15,18); 
\fill[black] (17,17) rectangle (18,18); 
\fill[black] (6,16) rectangle (7,17); 
\fill[black] (12,16) rectangle (13,17); 
\fill[black] (13,16) rectangle (14,17); 
\fill[black] (18,16) rectangle (19,17); 
\fill[black] (19,16) rectangle (20,17); 
\fill[black] (22,16) rectangle (23,17); 
\fill[black] (23,16) rectangle (24,17); 
\fill[black] (24,16) rectangle (25,17); 
\fill[black] (25,16) rectangle (26,17); 
\fill[black] (26,16) rectangle (27,17); 
\fill[black] (27,16) rectangle (28,17); 
\fill[black] (28,16) rectangle (29,17); 
\fill[black] (6,15) rectangle (7,16); 
\fill[black] (12,15) rectangle (13,16); 
\fill[black] (13,15) rectangle (14,16); 
\fill[black] (14,15) rectangle (15,16); 
\fill[black] (15,15) rectangle (16,16); 
\fill[black] (16,15) rectangle (17,16); 
\fill[black] (17,15) rectangle (18,16); 
\fill[black] (18,15) rectangle (19,16); 
\fill[black] (19,15) rectangle (20,16); 
\fill[black] (22,15) rectangle (23,16); 
\fill[black] (23,15) rectangle (24,16); 
\fill[black] (24,15) rectangle (25,16); 
\fill[black] (25,15) rectangle (26,16); 
\fill[black] (26,15) rectangle (27,16); 
\fill[black] (27,15) rectangle (28,16); 
\fill[black] (28,15) rectangle (29,16); 
\fill[black] (8,13) rectangle (9,14); 
\fill[black] (9,13) rectangle (10,14); 
\fill[black] (21,13) rectangle (22,14); 
\fill[black] (8,12) rectangle (9,13); 
\fill[black] (9,12) rectangle (10,13); 
\fill[black] (16,12) rectangle (17,13); 
\fill[black] (19,12) rectangle (20,13); 
\fill[black] (24,12) rectangle (25,13); 
\fill[black] (8,11) rectangle (9,12); 
\fill[black] (9,11) rectangle (10,12); 
\fill[black] (14,11) rectangle (15,12); 
\fill[black] (8,10) rectangle (9,11); 
\fill[black] (9,10) rectangle (10,11); 
\fill[black] (8,9) rectangle (9,10); 
\fill[black] (9,9) rectangle (10,10); 
\fill[black] (5,8) rectangle (6,9); 
\fill[black] (6,8) rectangle (7,9); 
\fill[black] (8,8) rectangle (9,9); 
\fill[black] (9,8) rectangle (10,9); 
\fill[black] (5,7) rectangle (6,8); 
\fill[black] (6,7) rectangle (7,8); 
\fill[black] (7,7) rectangle (8,8); 
\fill[black] (8,7) rectangle (9,8); 
\fill[black] (9,7) rectangle (10,8); 
\fill[black] (6,6) rectangle (7,7); 
\fill[black] (7,6) rectangle (8,7); 
\fill[black] (8,6) rectangle (9,7); 
\fill[black] (9,6) rectangle (10,7); 
\fill[black] (1,5) rectangle (2,6); 
\fill[black] (2,5) rectangle (3,6); 
\fill[black] (3,5) rectangle (4,6); 
\fill[black] (4,5) rectangle (5,6); 
\fill[black] (5,5) rectangle (6,6); 
\fill[black] (6,5) rectangle (7,6); 
\fill[black] (7,5) rectangle (8,6); 
\fill[black] (8,5) rectangle (9,6); 
\fill[black] (9,5) rectangle (10,6); 
\fill[black] (1,4) rectangle (2,5); 
\fill[black] (2,4) rectangle (3,5); 
\fill[black] (3,4) rectangle (4,5); 
\fill[black] (4,4) rectangle (5,5); 
\fill[black] (5,4) rectangle (6,5); 
\fill[black] (6,4) rectangle (7,5); 
\fill[black] (7,4) rectangle (8,5); 
\fill[black] (8,4) rectangle (9,5); 
\fill[black] (9,4) rectangle (10,5); 
\fill[black] (11,3) rectangle (12,4); 

\node at (-1.5,25.5) {in};
\node at (-1.5,3.5) {in};
\node at (32.5,14.5) {out};
\draw (2,25) -- (1,26);
\draw (1,25) -- (2,26);
\draw (2,3) -- (1,4);
\draw (1,3) -- (2,4);

\end{tikzpicture}
            \end{center}
            \caption{AND gate.}\label{fig: AND gate}
        \end{subfigure}
        \begin{subfigure}[t]{0.325\textwidth}
            \begin{center}
              \begin{tikzpicture}[scale=0.1]
\draw[color=gray, thick] (1,1) grid (29,29);
\draw[color=black, thick] (1,1) rectangle (29,29);

\fill[black] (1,27) rectangle (2,28); 
\fill[black] (2,27) rectangle (3,28); 
\fill[black] (3,27) rectangle (4,28); 
\fill[black] (4,27) rectangle (5,28); 
\fill[black] (5,27) rectangle (6,28); 
\fill[black] (6,27) rectangle (7,28); 
\fill[black] (7,27) rectangle (8,28); 
\fill[black] (8,27) rectangle (9,28); 
\fill[black] (9,27) rectangle (10,28); 
\fill[black] (10,27) rectangle (11,28); 
\fill[black] (1,26) rectangle (2,27); 
\fill[black] (2,26) rectangle (3,27); 
\fill[black] (3,26) rectangle (4,27); 
\fill[black] (4,26) rectangle (5,27); 
\fill[black] (5,26) rectangle (6,27); 
\fill[black] (6,26) rectangle (7,27); 
\fill[black] (7,26) rectangle (8,27); 
\fill[black] (8,26) rectangle (9,27); 
\fill[black] (9,26) rectangle (10,27); 
\fill[black] (10,26) rectangle (11,27); 
\fill[black] (11,24) rectangle (12,25); 
\fill[black] (8,23) rectangle (9,24); 
\fill[black] (9,23) rectangle (10,24); 
\fill[black] (8,22) rectangle (9,23); 
\fill[black] (9,22) rectangle (10,23); 
\fill[black] (8,21) rectangle (9,22); 
\fill[black] (9,21) rectangle (10,22); 
\fill[black] (8,20) rectangle (9,21); 
\fill[black] (9,20) rectangle (10,21); 
\fill[black] (13,20) rectangle (14,21); 
\fill[black] (8,19) rectangle (9,20); 
\fill[black] (9,19) rectangle (10,20); 
\fill[black] (8,18) rectangle (9,19); 
\fill[black] (9,18) rectangle (10,19); 
\fill[black] (16,18) rectangle (17,19); 
\fill[black] (17,18) rectangle (18,19); 
\fill[black] (20,18) rectangle (21,19); 
\fill[black] (21,18) rectangle (22,19); 
\fill[black] (8,17) rectangle (9,18); 
\fill[black] (9,17) rectangle (10,18); 
\fill[black] (17,17) rectangle (18,18); 
\fill[black] (20,17) rectangle (21,18); 
\fill[black] (6,16) rectangle (7,17); 
\fill[black] (8,16) rectangle (9,17); 
\fill[black] (9,16) rectangle (10,17); 
\fill[black] (12,16) rectangle (13,17); 
\fill[black] (13,16) rectangle (14,17); 
\fill[black] (14,16) rectangle (15,17); 
\fill[black] (15,16) rectangle (16,17); 
\fill[black] (16,16) rectangle (17,17); 
\fill[black] (21,16) rectangle (22,17); 
\fill[black] (22,16) rectangle (23,17); 
\fill[black] (25,16) rectangle (26,17); 
\fill[black] (26,16) rectangle (27,17); 
\fill[black] (27,16) rectangle (28,17); 
\fill[black] (28,16) rectangle (29,17); 
\fill[black] (6,15) rectangle (7,16); 
\fill[black] (7,15) rectangle (8,16); 
\fill[black] (8,15) rectangle (9,16); 
\fill[black] (9,15) rectangle (10,16); 
\fill[black] (12,15) rectangle (13,16); 
\fill[black] (13,15) rectangle (14,16); 
\fill[black] (14,15) rectangle (15,16); 
\fill[black] (15,15) rectangle (16,16); 
\fill[black] (16,15) rectangle (17,16); 
\fill[black] (17,15) rectangle (18,16); 
\fill[black] (18,15) rectangle (19,16); 
\fill[black] (19,15) rectangle (20,16); 
\fill[black] (20,15) rectangle (21,16); 
\fill[black] (21,15) rectangle (22,16); 
\fill[black] (22,15) rectangle (23,16); 
\fill[black] (25,15) rectangle (26,16); 
\fill[black] (26,15) rectangle (27,16); 
\fill[black] (27,15) rectangle (28,16); 
\fill[black] (28,15) rectangle (29,16); 
\fill[black] (6,14) rectangle (7,15); 
\fill[black] (8,14) rectangle (9,15); 
\fill[black] (9,14) rectangle (10,15); 
\fill[black] (8,13) rectangle (9,14); 
\fill[black] (9,13) rectangle (10,14); 
\fill[black] (24,13) rectangle (25,14); 
\fill[black] (8,12) rectangle (9,13); 
\fill[black] (9,12) rectangle (10,13); 
\fill[black] (16,12) rectangle (17,13); 
\fill[black] (22,12) rectangle (23,13); 
\fill[black] (27,12) rectangle (28,13); 
\fill[black] (8,11) rectangle (9,12); 
\fill[black] (9,11) rectangle (10,12); 
\fill[black] (13,11) rectangle (14,12); 
\fill[black] (8,10) rectangle (9,11); 
\fill[black] (9,10) rectangle (10,11); 
\fill[black] (8,9) rectangle (9,10); 
\fill[black] (9,9) rectangle (10,10); 
\fill[black] (8,8) rectangle (9,9); 
\fill[black] (9,8) rectangle (10,9); 
\fill[black] (8,7) rectangle (9,8); 
\fill[black] (9,7) rectangle (10,8); 
\fill[black] (8,6) rectangle (9,7); 
\fill[black] (9,6) rectangle (10,7); 
\fill[black] (1,5) rectangle (2,6); 
\fill[black] (2,5) rectangle (3,6); 
\fill[black] (3,5) rectangle (4,6); 
\fill[black] (4,5) rectangle (5,6); 
\fill[black] (5,5) rectangle (6,6); 
\fill[black] (6,5) rectangle (7,6); 
\fill[black] (8,5) rectangle (9,6); 
\fill[black] (9,5) rectangle (10,6); 
\fill[black] (1,4) rectangle (2,5); 
\fill[black] (2,4) rectangle (3,5); 
\fill[black] (3,4) rectangle (4,5); 
\fill[black] (4,4) rectangle (5,5); 
\fill[black] (5,4) rectangle (6,5); 
\fill[black] (6,4) rectangle (7,5); 
\fill[black] (7,4) rectangle (8,5); 
\fill[black] (8,4) rectangle (9,5); 
\fill[black] (9,4) rectangle (10,5); 
\fill[black] (11,3) rectangle (12,4); 
	\node at (-1.5,25.5) {in};
	\node at (-1.5,3.5) {in};
	\node at (32.5,14.5) {out};
	\draw (2,25) -- (1,26);
	\draw (1,25) -- (2,26);
	\draw (2,3) -- (1,4);
	\draw (1,3) -- (2,4);
	
	\end{tikzpicture} 
            \end{center}
            \caption{OR gate.}\label{fig: OR gate}
        \end{subfigure}
        \begin{subfigure}[t]{0.325\textwidth}
            \begin{center}
              \begin{tikzpicture}[scale=0.1]
\draw[color=gray, thick] (1,1) grid (29,29);
\draw[color=black, thick] (1,1) rectangle (29,29);
\fill[black] (9,28) rectangle (10,29); 
\fill[black] (1,27) rectangle (2,28); 
\fill[black] (2,27) rectangle (3,28); 
\fill[black] (3,27) rectangle (4,28); 
\fill[black] (4,27) rectangle (5,28); 
\fill[black] (5,27) rectangle (6,28); 
\fill[black] (6,27) rectangle (7,28); 
\fill[black] (7,27) rectangle (8,28); 
\fill[black] (1,26) rectangle (2,27); 
\fill[black] (2,26) rectangle (3,27); 
\fill[black] (3,26) rectangle (4,27); 
\fill[black] (4,26) rectangle (5,27); 
\fill[black] (5,26) rectangle (6,27); 
\fill[black] (6,26) rectangle (7,27); 
\fill[black] (7,26) rectangle (8,27); 
\fill[black] (9,25) rectangle (10,26); 
\fill[black] (9,24) rectangle (10,25); 
\fill[black] (9,23) rectangle (10,24); 
\fill[black] (9,22) rectangle (10,23); 
\fill[black] (9,21) rectangle (10,22); 
\fill[black] (6,20) rectangle (7,21); 
\fill[black] (9,20) rectangle (10,21); 
\fill[black] (9,19) rectangle (10,20); 
\fill[black] (9,18) rectangle (10,19); 
\fill[black] (19,18) rectangle (20,19); 
\fill[black] (20,18) rectangle (21,19); 
\fill[black] (9,17) rectangle (10,18); 
\fill[black] (12,17) rectangle (13,18); 
\fill[black] (19,17) rectangle (20,18); 
\fill[black] (9,16) rectangle (10,17); 
\fill[black] (20,16) rectangle (21,17); 
\fill[black] (21,16) rectangle (22,17); 
\fill[black] (22,16) rectangle (23,17); 
\fill[black] (23,16) rectangle (24,17); 
\fill[black] (24,16) rectangle (25,17); 
\fill[black] (25,16) rectangle (26,17); 
\fill[black] (26,16) rectangle (27,17); 
\fill[black] (27,16) rectangle (28,17); 
\fill[black] (28,16) rectangle (29,17); 
\fill[black] (14,15) rectangle (15,16); 
\fill[black] (15,15) rectangle (16,16); 
\fill[black] (16,15) rectangle (17,16); 
\fill[black] (17,15) rectangle (18,16); 
\fill[black] (18,15) rectangle (19,16); 
\fill[black] (19,15) rectangle (20,16); 
\fill[black] (20,15) rectangle (21,16); 
\fill[black] (21,15) rectangle (22,16); 
\fill[black] (22,15) rectangle (23,16); 
\fill[black] (23,15) rectangle (24,16); 
\fill[black] (24,15) rectangle (25,16); 
\fill[black] (25,15) rectangle (26,16); 
\fill[black] (26,15) rectangle (27,16); 
\fill[black] (27,15) rectangle (28,16); 
\fill[black] (28,15) rectangle (29,16); 
\fill[black] (8,14) rectangle (9,15); 
\fill[black] (11,14) rectangle (12,15); 
\fill[black] (12,14) rectangle (13,15); 
\fill[black] (9,12) rectangle (10,13); 
\fill[black] (14,12) rectangle (15,13); 
\fill[black] (9,11) rectangle (10,12); 
\fill[black] (12,11) rectangle (13,12); 
\fill[black] (9,10) rectangle (10,11); 
\fill[black] (9,9) rectangle (10,10); 
\fill[black] (6,8) rectangle (7,9); 
\fill[black] (9,8) rectangle (10,9); 
\fill[black] (9,7) rectangle (10,8); 
\fill[black] (9,6) rectangle (10,7); 
\fill[black] (1,5) rectangle (2,6); 
\fill[black] (2,5) rectangle (3,6); 
\fill[black] (3,5) rectangle (4,6); 
\fill[black] (4,5) rectangle (5,6); 
\fill[black] (5,5) rectangle (6,6); 
\fill[black] (6,5) rectangle (7,6); 
\fill[black] (7,5) rectangle (8,6); 
\fill[black] (1,4) rectangle (2,5); 
\fill[black] (2,4) rectangle (3,5); 
\fill[black] (3,4) rectangle (4,5); 
\fill[black] (4,4) rectangle (5,5); 
\fill[black] (5,4) rectangle (6,5); 
\fill[black] (6,4) rectangle (7,5); 
\fill[black] (7,4) rectangle (8,5); 
\fill[black] (9,3) rectangle (10,4); 
\fill[black] (10,3) rectangle (11,4); 
\node at (-1.5,25.5) {in};
\node at (-1.5,3.5) {in};
\node at (32.5,14.5) {out};
\draw (2,25) -- (1,26);
\draw (1,25) -- (2,26);
\draw (2,3) -- (1,4);
\draw (1,3) -- (2,4);

\end{tikzpicture}
            \end{center}
            \caption{XOR gate.}\label{fig: XOR gate}
        \end{subfigure}
        \caption{Gadgets for the implementation of logic circuits for the rules $2$ and $24$.}\label{fig: gates}
\end{figure}

We represent the information flowing through wires, which are based, roughly, on a line of active sites. Then, all the sites over (under) this line will have one active neighbor. If a cell over the line becomes active, then in the next step a neighbor of this cell will become active, so the information flows over the wire. 

These constructions of the gates are quite standard. Maybe one exception is the XOR gate. The crucial observation is that we manage to simulate the XOR using the syncronisity of information. An XOR gate consists roughly in two confluent wires. If a signal arrives from one of the two wires, the signal simply passes.  If two signals arrive at the same time, the next cell in the wire will have more than two neighbors, so it will remain active.

Using the XOR gate (Figure \ref{fig: XOR gate}), one can build a planar crossing gadget, concluding the \Pt-completeness constructions.

\begin{theorem}\label{thm: p-complete square}
   \stability is \Pt-complete for rules $2$ and $24$.
\end{theorem}

Remark: In our construction we strongly use neighborhoods composed only of inactive cells, so these constructions can not be used for rules $02$ and $024$, where zero is not a quiescent state. So in these cases stability could have less complexity.

\section{Concluding Remarks}
\subsection{Summary of our results}
In this paper we have studied the complexity of the \stability problem for the set of binary Freezing Totalistic Cellular Automata (FTCA) on the triangular and square grid with von Neumann neighborhood.

We find different complexities for this FTCA, including a \Pt-complete case on the square grid.
For the rules where \stability is in \NC we have considered two approaches: a topological approach (Theorems \ref{thm: SyncStability 2}, \ref{thm: SyncStability 23}, \ref{thm:rule 34}, \ref{thm:rule 3}, and \ref{thm: 234 square NC}) and an algebraic  approach (Theorems \ref{thm: SyncStability or} and \ref{thm: 12** square NC}). A summary of our results is given in Tables \ref{table: complexity 3} and \ref{table: complexity 4}.

\begin{table}[h]
\centering 
{\footnotesize
\[\begin{array}{|c|c|c|}
   \hline
 \text{Rule}& \stability & Theorem \\   
  \hline   
 \phi         &  \text{$\mathcal{O}(1)$ }&\text{ Trivial                 }       \\
 3           &  \text{\text{ \NC} }&\text{ Trivial                }        \\
 2           &  \text{\text{ \NC}   }& \text{Thm \ref{thm: SyncStability 2}}  \\
 23          &  \text{\text{ \NC}   }& \text{Thm \ref{thm: SyncStability 23}}  \\
 12          &  \text{\text{ \NC}   }& \text{Thm \ref{thm: SyncStability or}}  \\
 123         &  \text{\text{ \NC}   }& \text{Trivial                         } \\ 
  \hline
 \end{array}\]}
  \caption
 {Summary of rules and their complexity of \stability on the triangular grid.}
  \label{table: complexity 3} 
\end{table}

\begin{table}[h]
\centering 
{\footnotesize
\[\begin{array}{|c|c|c|c|c|c|}
   \hline
 \text{Rule}&  \stability & Theorem & \text{Rule}&  \stability & Theorem \\   
  \hline   
 4    &\text{\text{ \NC} }  &\text{ Trivial }                         &   234   &\text{\text{ \NC}   }     & \text{Thm \ref{thm: 234 square NC}} \\
 3    &\text{\text{ \NC}   }     & \text{Thm \ref{thm:rule 3} }   &   12    &  \text{\NC}& \text{Thm \ref{thm: 12** square NC}} \\
 34   &\text{\text{ \NC}   }     & \text{Thm \ref{thm:rule 34} }   &   124   &  \text{\NC}& \text{Thm \ref{thm: 12** square NC}} \\
 2    &\text{\text{\Pt-Complete}}& \text{Thm \ref{thm: p-complete square}} &   123   &  \text{\NC}& \text{Thm \ref{thm: 12** square NC}}  \\
 24   &\text{\text{\Pt-Complete}}& \text{Thm \ref{thm: p-complete square}} &   1234  &  \text{$\mathcal{O}(1)$ }  &\text{ Trivial  } \\ 
  \hline
 \end{array}\]
 }
 \caption
 {Summary of rules and their complexity of \stability on the square grid.}
  \label{table: complexity 4} 
\end{table}

\subsection{About Fractal-Growing Rules}
In this paper we have not included a study of fractal growing rules. In fact, the complexity of \stability remains open for these rules, even for fractal rules defined over a triangular grid. 


To have an intuition about the dynamical complexity of those rules, see Figures \ref{fig: rules1 sq} and \ref{fig: rules1 tri}, where starting with only the center active we obtain a fractal behavior.

\begin{figure}
        \centering
  \begin{subfigure}[t]{0.24\textwidth}
    \centering
    \includegraphics[width=.95\textwidth]{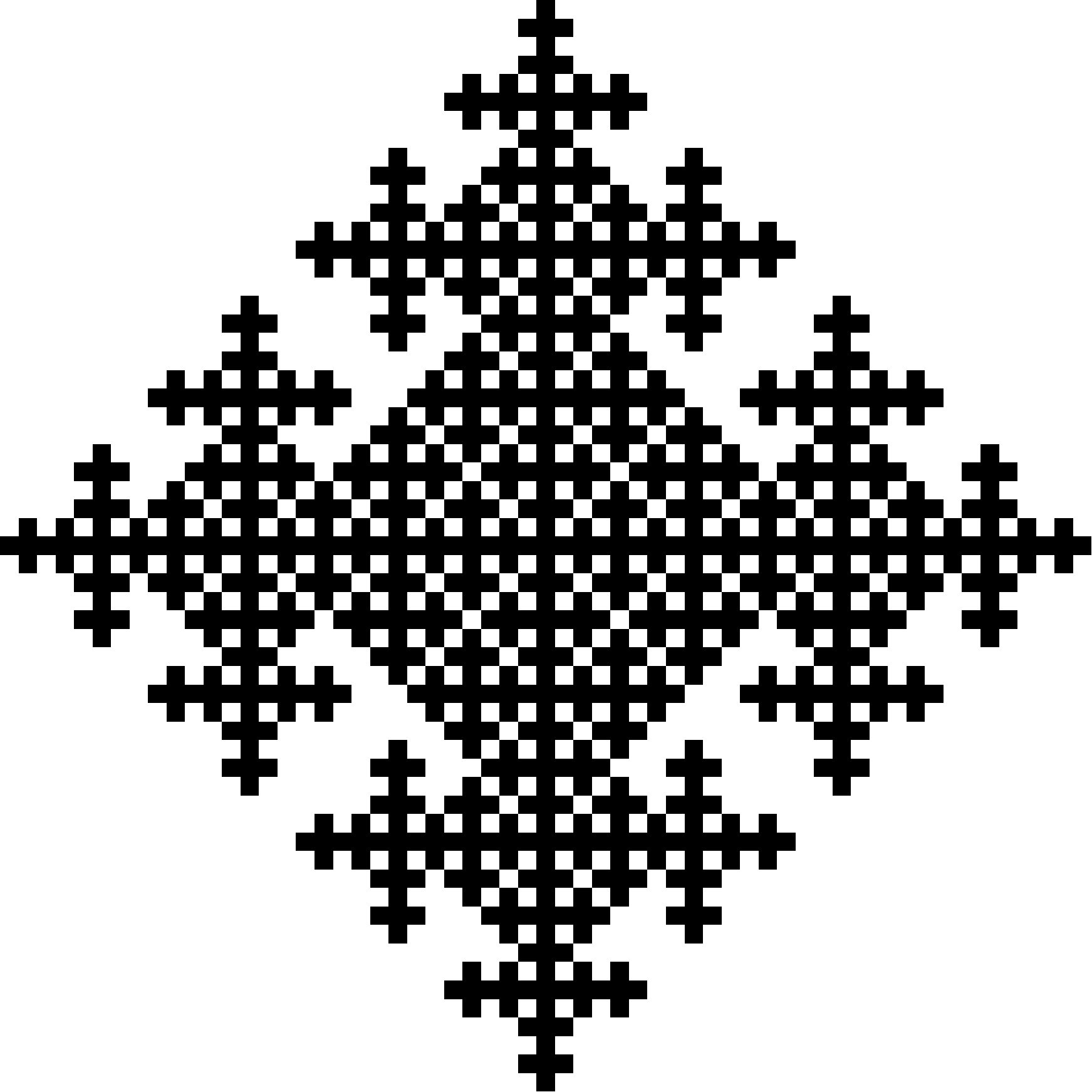}  
    \caption{Rule 1}
  \end{subfigure}   
  \begin{subfigure}[t]{0.24\textwidth}
    \centering  
    \includegraphics[width=.95\textwidth]{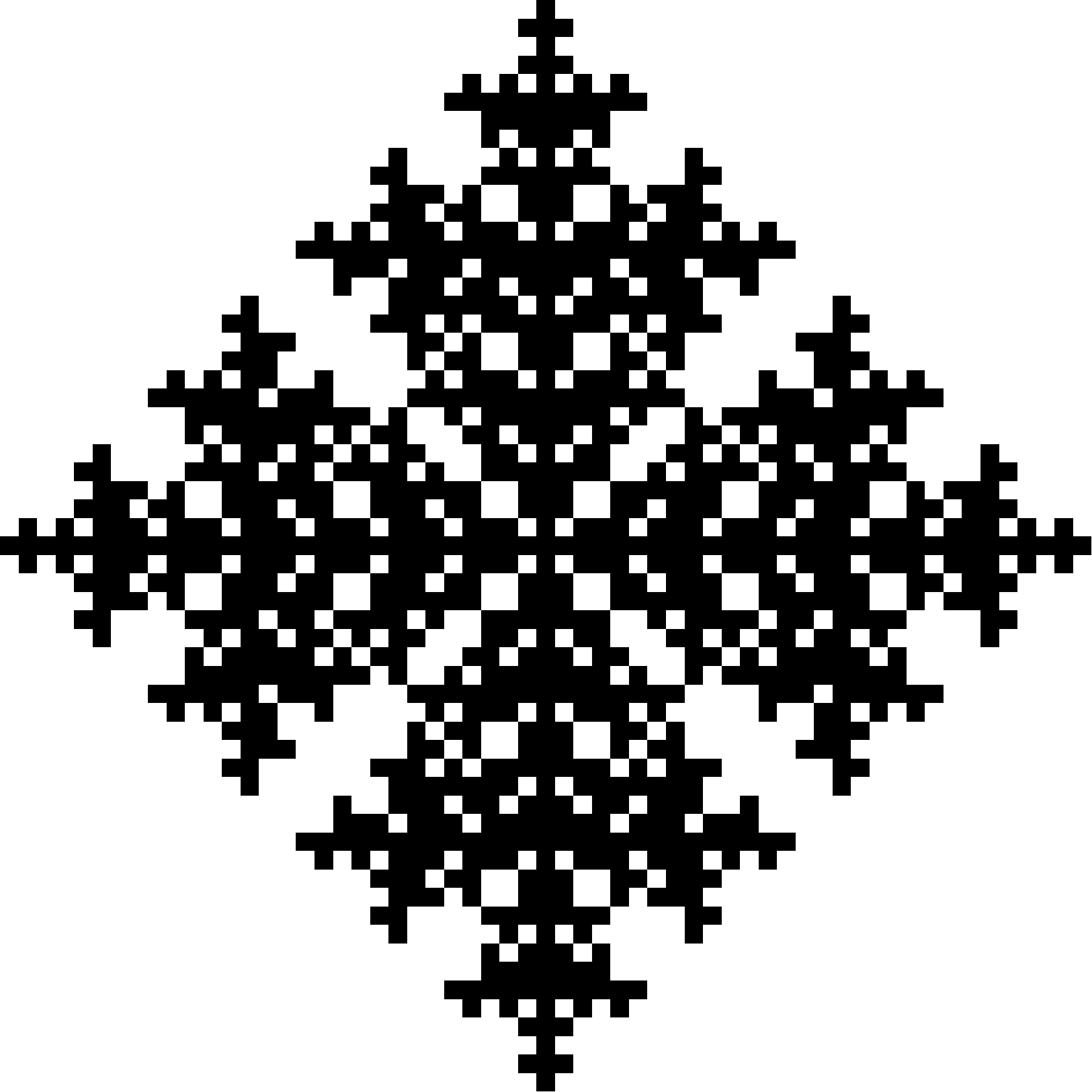}  
    \caption{Rule 13}
  \end{subfigure}      
  \begin{subfigure}[t]{0.24\textwidth}
    \centering 
    \includegraphics[width=.95\textwidth]{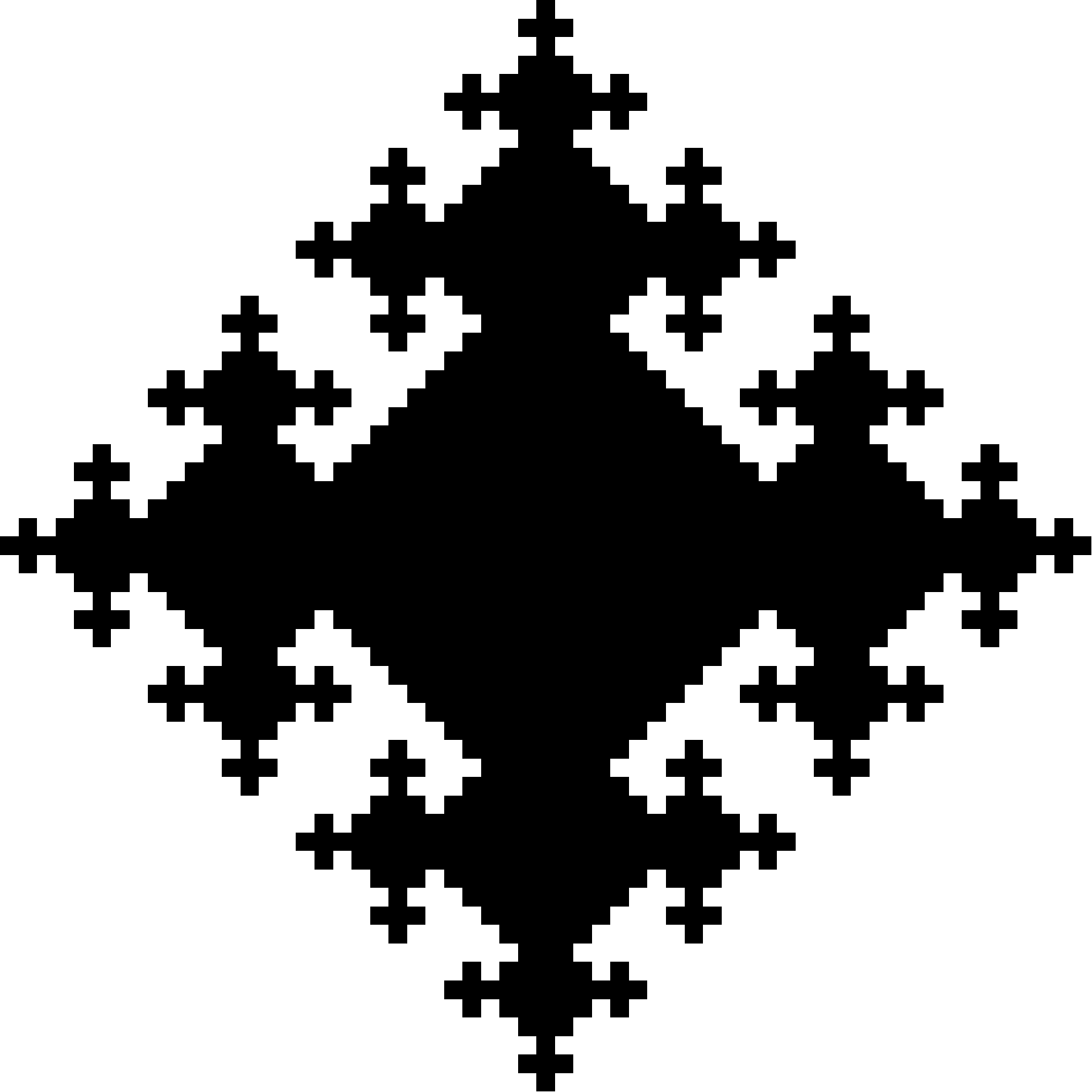} 
    \caption{Rule 14}
  \end{subfigure}             
  \begin{subfigure}[t]{0.24\textwidth}
    \centering 
    \includegraphics[width=.95\textwidth]{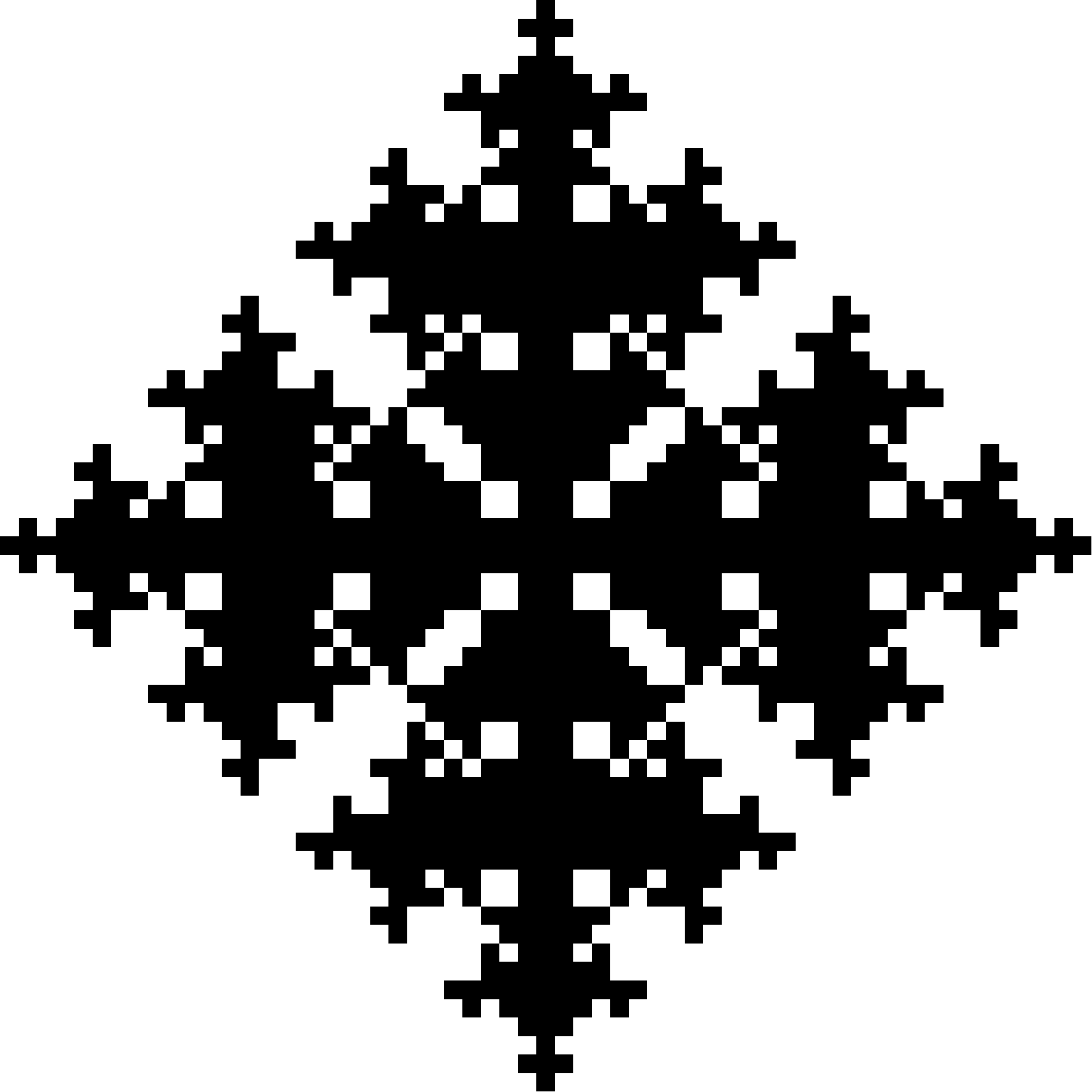} 
    \caption{Rule 134}
  \end{subfigure}             
        \caption{Examples of different rules with the similar fractal dynamics starting with a single active cell on square grid.}\label{fig: rules1 sq}
\end{figure}
\begin{figure}
        \centering
  \begin{subfigure}[t]{0.3\textwidth}
    \centering
    \includegraphics[width=.95\textwidth]{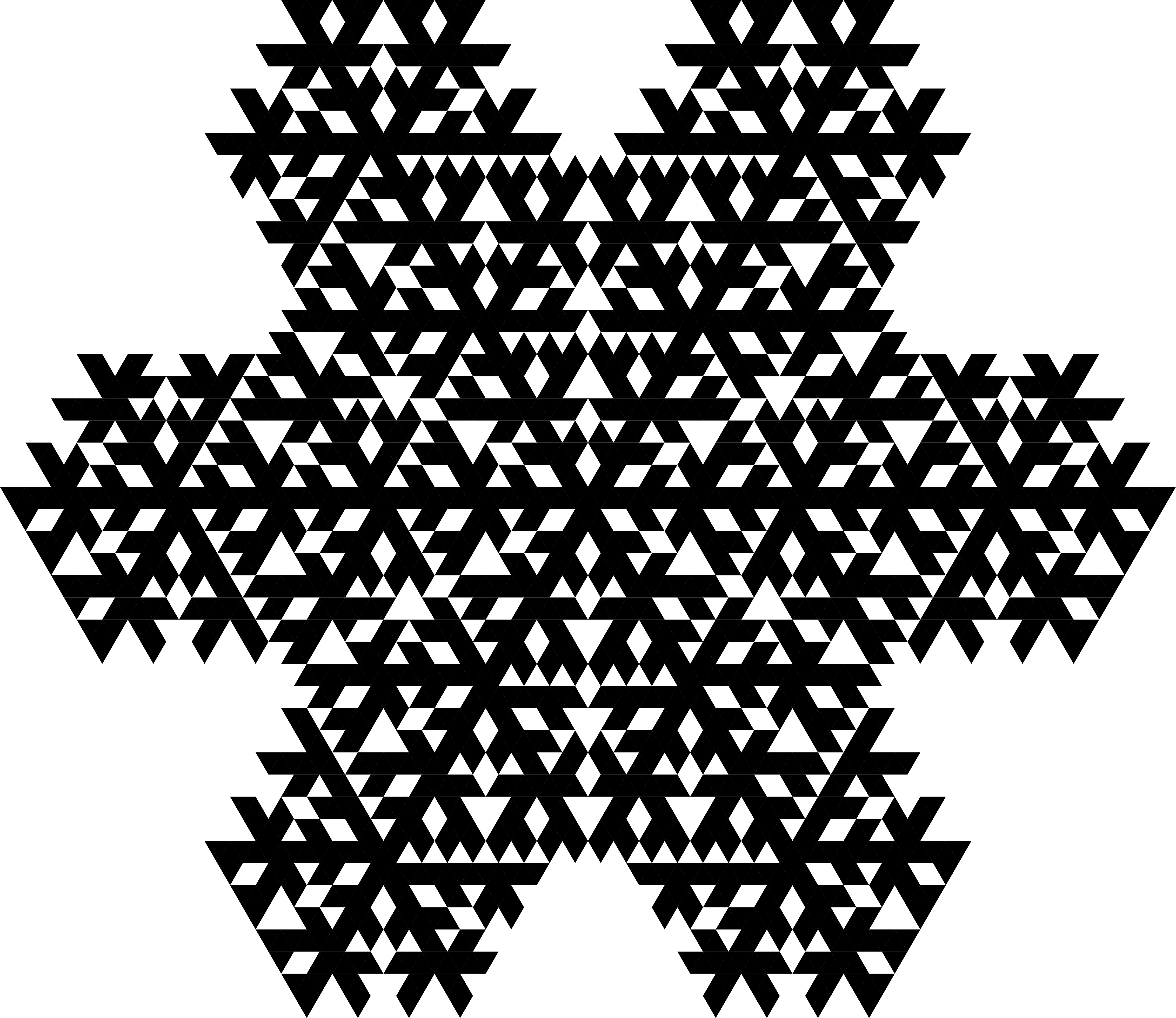}  
    \caption{Rule 1}
  \end{subfigure}   
  \qquad
  \begin{subfigure}[t]{0.3\textwidth}
    \centering  
    \includegraphics[width=.95\textwidth]{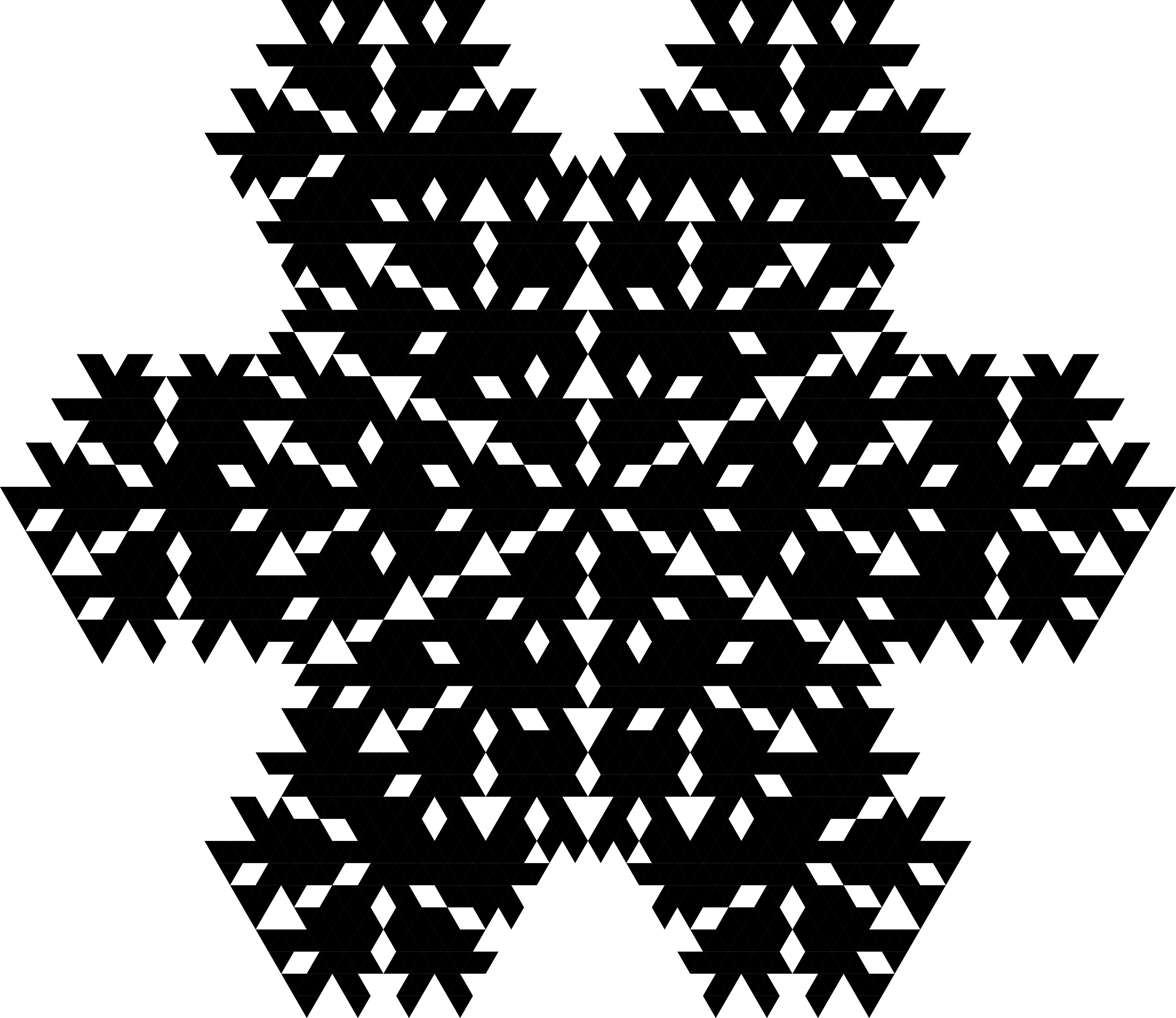}  
    \caption{Rule 13}
  \end{subfigure}                 
        \caption{Examples of different rules with the similar fractal dynamics starting with a single active cell on triangular grid.}\label{fig: rules1 tri}
\end{figure}

It  interesting to remark that non-freezing version of rule $13$ is the usual XOR between the four neighbors, which is a linear cellular automaton. Using a prefix-sum algorithm, we can compute any step of a linear cellular automaton, so the \emph{non-freezing} rule $13$ is in \NC. Although we might imagine that adding \emph{freezing} property simplifies the dynamics of a rule,  the restriction of this rule to freezing dynamics introduces a non-linearity that to prevented us to characterize its complexity.

\subsection{On P-Completeness on the triangular grid}
It is important to point out that for triangular graph (despite rule $1$ or $13$ might be candidates), we do not exhibit a rule such that \stability is \Pt-complete. The reader might think that, like in the one-dimensional case, every freezing rule defined in a triangular grid is \NC. This is not the case. Moreover, three states (that we call $0$, $1$, $2$) suffice to define a \Pt-complete FTCA. The general freezing property means that states may only grow (so, in this case state 2 is stable). In this context, for a triangular grid,  consider the following the local function.

$$f\left(x_u,\sum_{z\in N+u}x_z\right)=\begin{cases}
  1,  &\text{ if } x_u=0 \wedge (\sum_{z\in N+u}x_z=2 \vee \sum_{z\in N}x_z=12)\\
  1,  &\text{ if } x_u=10 \wedge \sum_{z\in N+u}x_z=11\\
  x_u, &\text{ otherwise } 
\end{cases}
$$
The proof of \Pt-Completeness follows similar arguments than the ones we used for rules  $2$ and $24$ in the square grid (Theorem \ref{thm: p-complete square}), i.e. reducing the Circuit Value Problem (CVP) to \stability on this rule. 
Instances of CVP are encoded into a configuration of this FTCA using the idea in the construction of the logical gates. In Figures \ref{fig: AND gate} and \ref{fig: xor tri} we exhibit the gadgets.
 
\begin{figure}
        \centering
  \begin{subfigure}[t]{0.24\textwidth}
    \centering  
    \includegraphics[width=.95\textwidth]{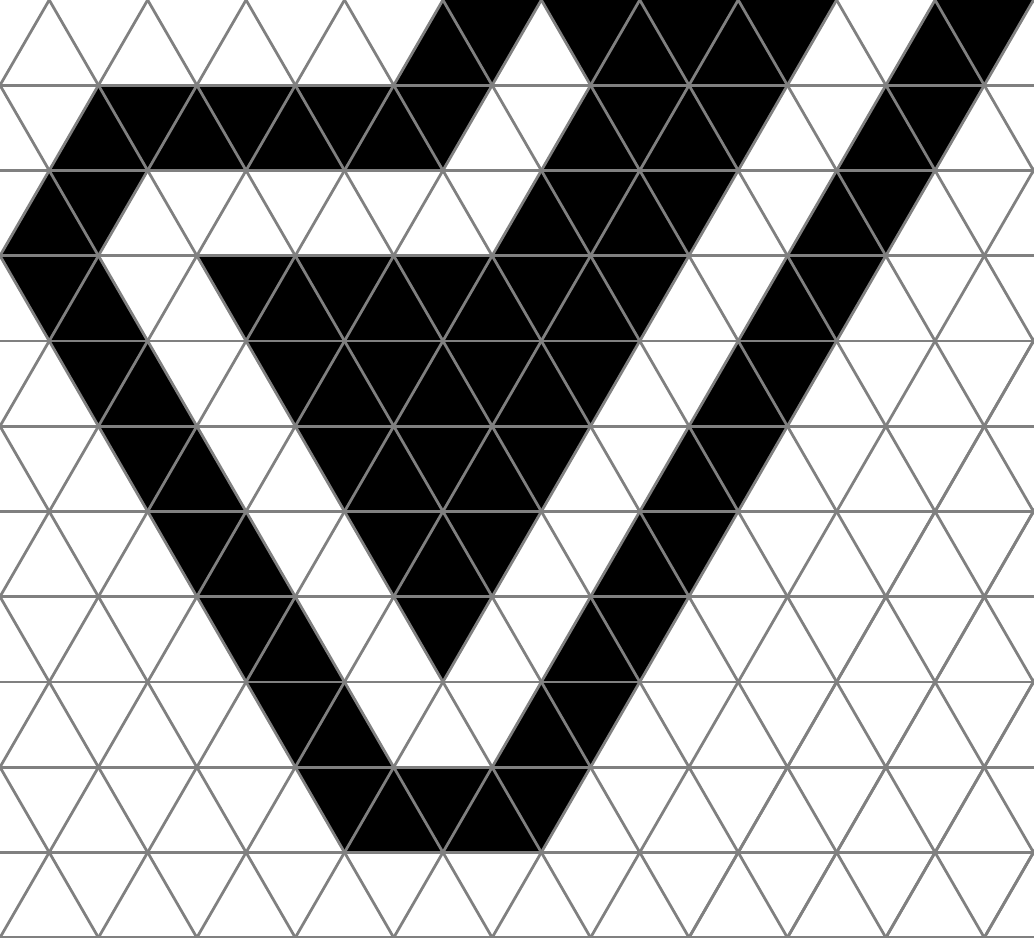}  
    \caption{Wire at time $0$.}\label{fig: wire1t}
  \end{subfigure}       
  \begin{subfigure}[t]{0.24\textwidth}
    \centering 
    \includegraphics[width=.95\textwidth]{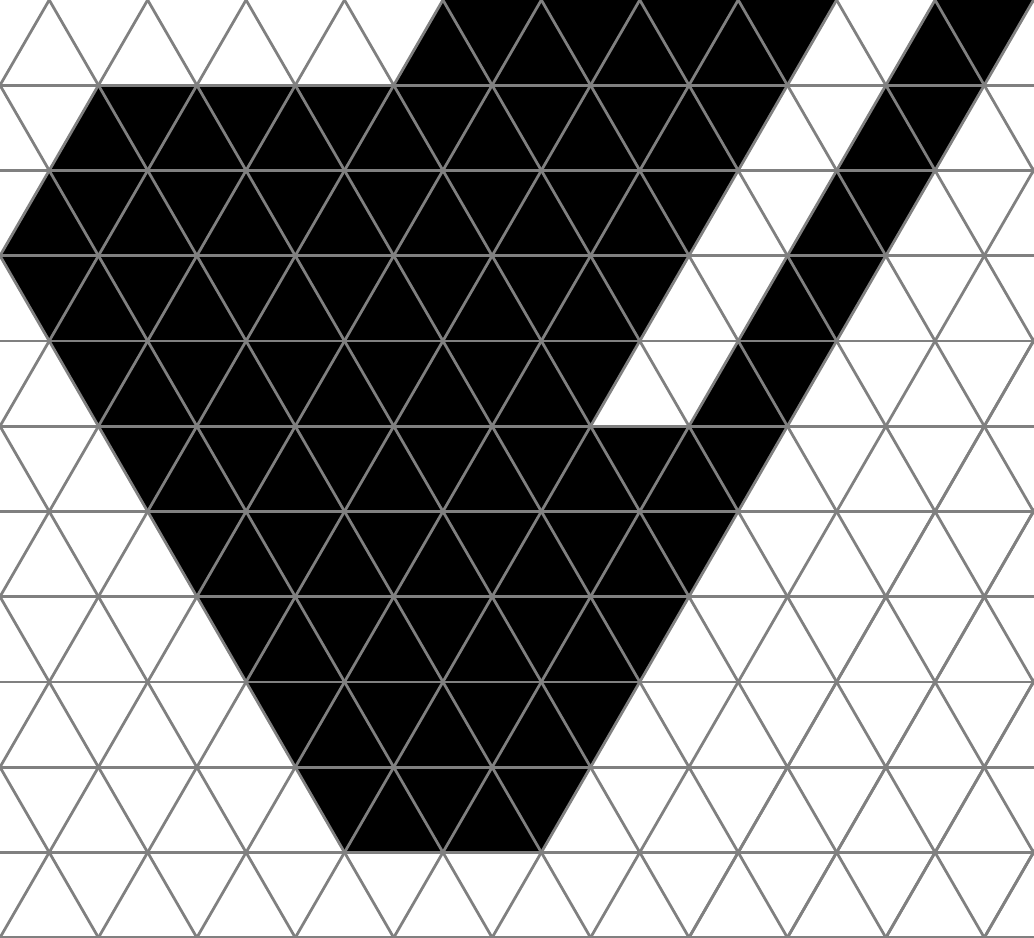}   
    \caption{Wire at time $30$.}
  \end{subfigure}             
\begin{subfigure}[t]{0.24\textwidth}
    \centering  
    \includegraphics[width=.95\textwidth]{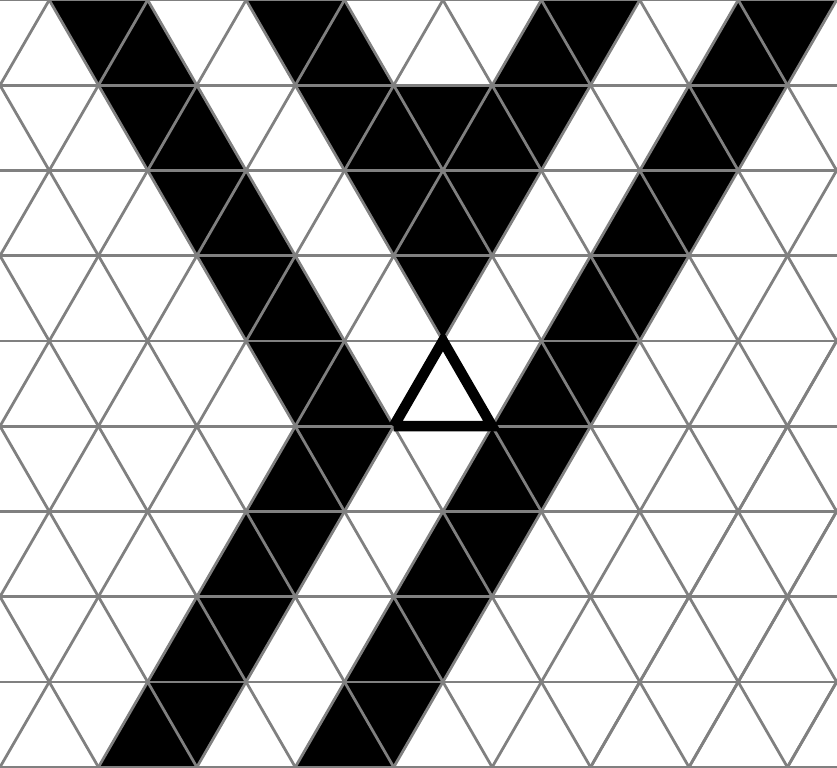}  
    \caption{AND gate.}
    \label{fig: and tri}
  \end{subfigure}
  \vspace{0.16\textwidth}      
  \begin{subfigure}[t]{0.24\textwidth}
    \centering 
    \includegraphics[width=.95\textwidth]{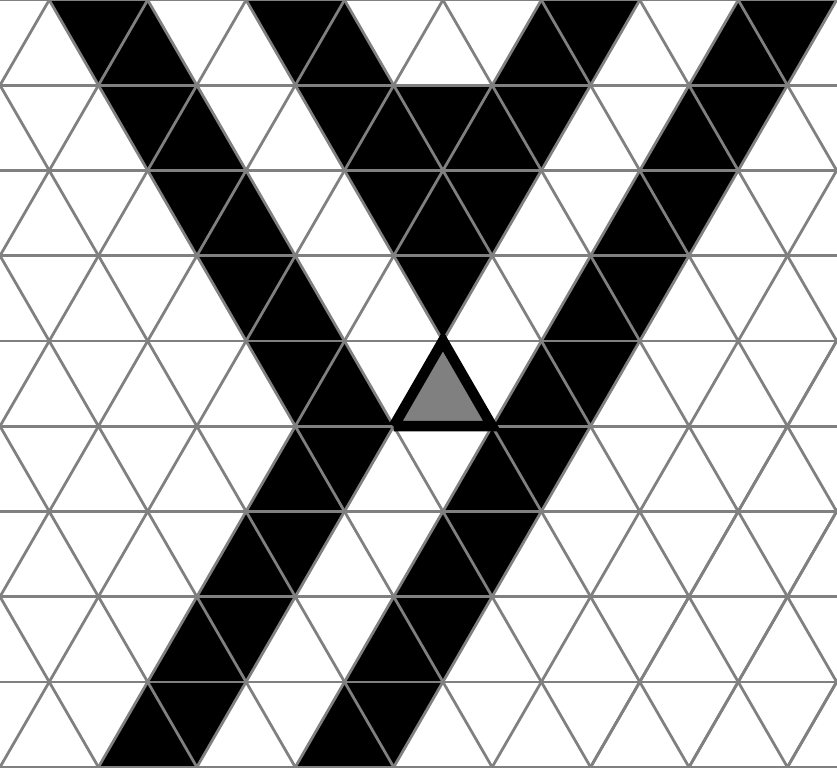}   
    \caption{XOR gate.}
    \label{fig: xor tri}
  \end{subfigure}   
  \vspace{-60pt}        
  \caption{Gadgets for the implementation of logic circuits.
    The thick line marks the cell that makes the calculation from signals. 
    The color code is: \Trianguloooo{black} : 1, \Trianguloooo{gray} : 10 and \Trianguloooo{white} : 0}\label{fig: gatest}
\end{figure}

\subsection{About non-quiescent rules}

Finally, it is convenient to say a word about rules where cells become active with zero active neighbors, i.e., rules where state $0$ is not quiescent.  Clearly, after one step for those rules, every cell will have at least one active neighbor. Then, their complexity is upper-bounded by the complexity of the same rule not considering the case of zero active neighbors as an activating state. For example, consider rule $034$ in the square grid. After one step of rule $034$, the dynamics are exactly the same that the one of rule $34$. Therefore, rule $034$ is in \NC. 
Although, there are some interesting cases. First, notice that rule $01$ is trivial (in the triangular or the square grid), because after only one step the rule reaches a fixed point. This contrasts with rule $1$, which is a Fractal-Growing rule. Second, consider rule $02$ or $024$  in the square grid. We know that rule $2$ and $24$ are \Pt-Complete. However, the reader can verify that the gadgets used to reduce CVP to \stability do not work for rule $02$ and $024$. This fact opens the possibility that rules $02$ and $024$ belong to \NC.

\section{Acknowledgment}
This work was supported by CONICYT + PAI + CONVOCATORIA NACIONAL SUBVENC\'OIN A INSTALACI\'ON EN LA ACADEMIA CONVOCATORIA A\~NO 2017 + PAI77170068 (P.M.), FONDECYT 11190482 (P.M.), and CONICYT via Programa Regional STIC-AmSud (CoDANet) cód. 19-STIC-03 (E.G. and P.M.).
\section{Bibliography}
\bibliographystyle{elsarticle-num} 
\bibliography{biblio}

\end{document}